\documentclass[11pt,leqno]{article}
\usepackage{amsthm,amsmath, mathtools}
\usepackage[round]{natbib}
\usepackage{amsfonts}
\usepackage{ifthen}
\usepackage{enumitem}
\usepackage{amssymb,dsfont,url,color,booktabs}
 \usepackage{tikz}
 \usepackage{setspace}
\usepackage{epstopdf}
\usepackage{chngcntr}
\usepackage{bbm}
\usepackage{array}
\usepackage[export]{adjustbox}[2011/08/13] 
\usepackage{geometry}
\usepackage{rotating}
\usepackage{todonotes}
\usepackage{xspace}

 \usepackage[english]{babel}

\usepackage{caption}
\usepackage{subcaption}
\usepackage{multirow}
\usepackage{booktabs}
\usepackage[font={footnotesize}]{caption}

\usepackage{placeins} 		

\newcommand*{\IR}{\mathbb{R}}

\theoremstyle{plain}
\newtheorem{theorem}{Theorem}[section]
\newtheorem{lemma}[theorem]{Lemma}
\newtheorem{proposition}[theorem]{Proposition}
\newtheorem{corollary}[theorem]{Corollary}

\newtheorem{remark}[theorem]{Remark}
\newtheorem{example}[theorem]{Example}

\newtheorem{conditions}[theorem]{Conditions}

\newcommand{\dd}[1]{\operatorname{d}\!#1}

\newcommand{\ee}[1]{\operatorname{e}^{#1}}
\newcommand{\R}{\mathds R}

\newcommand{\OF}{\mathcal{F}}

\newcommand{\OP}{\mathcal{P}}

\newcommand{\notiz}[1]{\relax}

\newcommand{\zitep}[1]{\relax}
\newcommand{\skr}{\rangle}
\newcommand{\1}{\mathds 1}            

\newcommand{\nn}{\mathds N}

\newcommand{\rr}{\mathds R}

\newcommand{\rrd}{\mathds{R}^d}

\newcommand{\cc}{\mathds C}
\newcommand{\skl}{\langle}
\newcommand{\sign}{\operatorname{sgn}}

\newcommand{\Price}[1][]{
		\ifthenelse{\equal{#1}{}}{\mathit{Price}}{\Price{}^{#1}}
	} 

\newcommand{\Call}{\mathit{Call}}

\newcommand{\icc}{i}

\newcommand{\tild}{~}

\newcommand*{\BS}{\text{Black{\&}Scholes}\xspace}

\newlength{\wordlength}

\newcommand{\olp}[1][]{\textcolor{white}{\underline{\textcolor{black}{\overline p}}_{\hspace{-.05em}\textcolor{black}{#1}}}}

\renewcommand{\cite}{\citet}

\numberwithin{equation}{section}
\numberwithin{figure}{section}
\numberwithin{table}{section}

\begin{document}
\title{\textbf{Chebyshev Interpolation for Parametric Option Pricing}\footnote{We like to thank Jonas Ballani, Behnam Hashemi, Daniel Kressner and Nick Trefethen for fruitful discussions on Chebyshev interpolation. Moreover, we gratefully acknowledge valuable feedback from Christian Bayer, Ernst Eberlein, Dilip Madan, Christian P\"{o}tz, Peter Tankov and Ralf Werner. For further we thank Paul Harrenstein and Pit Forster.
Additionally, we thank the participants of the conferences Advanced Modelling in Mathematical Finance, A conference in honour of Ernst Eberlein, Kiel. May 20--22, 2015, Stochastic Methods in Physics and Finance, 2015 in Heraklion, MoRePas–2015: Model Reduction of Parametrized Systems III, held in Trieste and the 12th German Probability and Statistic Days 2016 in Bochum. Furthermore we thank the participants of the research seminars Seminar Stochastische Analysis und Stochastik der Finanzm\"arkte, Technical University Berlin.  May 28, 2015 and the Groupe de Travail MathfiProNum: Finance math\'ematique, probabilit\'es num\'eriques et statistique des processus, Universit\'e Diderot, Paris. June 11, 2015.}
}

\bigskip
\author{\textbf{Maximilian Ga{\ss}$\vphantom{l}^{1,\dagger}$, Kathrin Glau$\vphantom{l}^{1}$,} \\\textbf{Mirco Mahlstedt$\vphantom{l}^{1,}$\footnote{We thank the KPMG Center of Excellence in Risk Management for their support.}\ , Maximilian Mair$\vphantom{l}^{1}$
}
\\\\$\vphantom{l}^{\text{1}}$Technical University of Munich, Germany
}

\maketitle
\begin{abstract}
Recurrent tasks such as pricing, calibration and risk assessment need to be executed accurately and in real-time. Simultaneously we observe an increase in model sophistication on the one hand and growing demands on the quality of risk management on the other. To address the resulting computational challenges, it is natural to exploit the recurrent nature of these tasks. We concentrate on Parametric Option Pricing (POP) and show that polynomial interpolation in the parameter space promises to reduce run-times while maintaining accuracy. The attractive properties of Chebyshev interpolation and its tensorized extension enable us to identify criteria for (sub)exponential convergence and explicit error bounds. We show that these results apply to a variety of European (basket) options and affine asset models. Numerical experiments confirm our findings. Exploring the potential of the method further, we empirically investigate the efficiency of the Chebyshev method for multivariate and path-dependent options.
\end{abstract}

	


\section{Introduction}

The development of fast and accurate computational methods for parametric models is one of the central issues in computational finance. Financial institutions dedicated to trading or assessment of financial derivatives have to cope with the daily tasks of computing numerous characteristic financial quantities. Examples of interest are prices, sensitivities and risk measures for products on different models and for varying parameter constellations. With regard to the ever growing market activities, more and more of these evaluations need to be delivered in real-time. 
In addition we face a constantly rising model sophistication since the original work of \cite{BlackScholes1973} and \cite{merton1973}. 
From the early nineties on, stochastic volatility and L\'evy models as well as models based on further classes of stochastic processes have been developed that reflect the observed market data in a more appropriate way. For asset models see e.g.\ \cite{Heston1993}, Eberlein, Keller and Prause~(1998)\nocite{EberleinKellerPrause98}, Duffie, Filipovi\'c and Schachermayer~(2003)\nocite{DuffieFilipovicSchachermayer2003}, Cuchiero, Keller-Ressel and Teichmann~(2015)\nocite{CuchieroKeller-ResselTeichmann2015}. In the case of fixed income models see e.g.\ \cite{EberleinOezkan2005}, Keller-Ressel, Papapantoleon and Teichmann~(2013)\nocite{Keller-ResselPapapantoleonTeichmann2013}, Filipovi\'c, Larsson and Trolle~(2014)\nocite{FilipovicLarssonTrolle2014}. The aftermath of the financial crisis 2007--2009, moreover, has lead to a new generation of more complex models, for instance by incorporating more risk factors. The usefulness of a pricing model critically depends on how well it captures the relevant aspects of market reality in its numerical implementation. 
Exploiting new ways to deal with the rising computational complexity therefore supports  the evolution of pricing models and touches a core concern of present mathematical finance.

A large body of computational tasks in finance needs to be repeatedly performed in real-time for a varying set of parameters. Prominent examples are option pricing and hedging of different option sensitivities, e.g. delta and vega, that also need to be calculated in real-time. In particular for optimization routines arising in model calibration, large parameter sets come into play. Further examples arise in the context of risk control and assessment, such as for quantification and monitoring of risk measures.
The following question serves as a starting point of our investigations:
\textit{How to systemically exploit the recurrent nature of parametric computational problems in finance with the approached objective to gain efficiency?} Looking for answers to this question, we focus on Parametric Option Pricing~(POP) in the sequel.

In the present literature on computational methods in finance, complexity reduction for parametric problems has largely been addressed by applying Fourier techniques following the seminal works of \cite{CarrMadan99} and \cite{Raible}. See also the monograph \cite{BoyarchenkoLevendorskii2002}. Research in this area concentrates on adopting fast Fourier transform (FFT) methods and variants for option pricing. \cite{Lee2004} accurately describes pricing European options with FFT. Further developments are for instance provided by Lord, Fang, Bervoets and Oosterlee~(2008)\nocite{LordFangBervoetsOosterlee2008} for early exercise options and by \cite{FengLinetsky2008} and \cite{KudryavtsevLevendorskiy2009} for barrier options. Another path to efficiently handle large parameter sets that has been pursued in finance relies on reduced basis methods. These are techniques to solve parametrized partial differential equations. \cite{SachsSchu2010}, Cont, Lantos and Pironneau~(2011)\nocite{ContLantosPironneau2011}, \cite{Pironneau2011} and Haasdonk, Salomon and Wohlmuth~(2012)\nocite{HaasdonkSalomonWohlmuth2012b} and Burkovska, Haasdonk, Salomon and Wohlmuth~(2015)\nocite{burkovska2015reduced} applied this approach to price European, American plain vanilla options and European baskets. FFT methods on the one hand can be highly beneficial when the prices are required in a large number of Fourier variables, e.g.\ for a large set of strikes of European plain vanillas. On the other hand numerical experiments have shown a promising gain in efficiency of reduced basis methods when an accurate PDE solver is readily available. In essence all these approaches reveal an immense potential of complexity reduction by targeting parameter dependence. Hereto, they exploit the functional architecture of the underlying pricing technique for varying parameters.

Financial institutions have to deal simultaneously with a diversity of models, a multitude of option types, and---as a consequence---a wide variety of underlying pricing techniques. It is therefore worthwhile to explore the possibility of a \emph{generic} complexity reduction method that is independent of the specific pricing technique.
To do so, we focus on the set of option prices and the set of parameters of interest, disregard on purpose the pricing technology and view the option price as a function of the parameters. 
The core idea is now to \emph{introduce interpolation of option prices in the parameter space as a complexity reduction technique for POP}. 

The resulting procedure naturally splits into two phases: Pre-computation and real-time evaluation. The first one is also called \emph{offline phase} while the second is also called \emph{online phase}.
  In the pre-computation phase the prices are computed for some fixed parameter configurations, namely the interpolation nodes. Here, any appropriate pricing method---for instance based on Fourier, PDE or even Monte Carlo techniques---can be chosen. The real-time evaluation phase then consists of the evaluation of the interpolation. Provided that the evaluation of the interpolation is faster than the benchmark tool, the scheme permits a gain in efficiency in all cases where accuracy can be maintained. Then, we distinguish two use cases:
\begin{itemize}
\item 
In comparison to the benchmark pricing routine, the fast evaluation of the interpolation will eventually outweigh the expensive pre-computation phase, if pricing is a task repeatedly employed.
Optimization procedures are an obvious instance where this feature becomes advantageous.
\item 
Even if the number of price computations is limited, we can still benefit from the separation of the procedure into its two phases. In this way, e.g., idle times in the financial industry can be put to good use by preparing the interpolation for whenever real-time pricing is needed during business activities.
\end{itemize}  
The question arising at this stage is: \textit{Under what circumstances can we hope to find an interpolation method that delivers both reliable results and a considerable gain in efficiency?} 
  
   One could now be tempted to proceed in a naive manner and first define an equidistant grid and then interpolate piecewise linearly in the parameter space. Numerical experiments for \BS call prices as function of the volatility, for instance, would then provide convincing evidence that the number of nodes needed for a given accuracy is considerably high. Increasing the polynomial degree might lead to better results. However, convergence might not be guaranteed. \cite{Runge1901} showed that polynomial interpolation on equispaced grids may diverge---even for analytic functions. Second, the evaluation of the polynomial interpolants  may be numerically problematic, as it is shown in \cite{Runge1901} that "the interpolation problem for polynomial interpolation on an equidistant grid is exponentially ill-conditioned", a formulation we borrow from \cite{TrefethenMythTalk}. For these reasons we abstain from polynomial interpolation with equidistant grids. Rather we take a step back and ask: \textit{Which methods for the interpolation of prices as functions of model and payoff parameters are numerically promising in terms of convergence, stability and implementational ease?}
  
Regarding this research question, we need to take into consideration both the set of interpolation methods as such and the specific features of the functions we investigate. 
It is well-known that the efficiency of interpolation methods critically depends on the degree of regularity of the approximated function. For the core problem of our study---European (basket) options---we investigate the regularity of the option prices as functions of the parameters. We find that these functions are indeed \textit{analytic} for a large set of option types, models and parameters. Taking the perspective of approximation theory, this inspires the hope to find suitable interpolation methods.

Empirically, we observe that parameters of interest often range within bounded intervals. One interpolation method that is highly effective for analytic functions on bounded intervals is Chebyshev interpolation. This intensively studied method enjoys excellent numerical properties---in stark contrast to polynomial interpolation on equally spaced nodal points. The interpolation nodes are known beforehand, implementation is straightforward and the method is numerically stable. For univariate functions that are several times differentiable, the method converges polynomially and for univariate analytic functions convergence is exponential. 
In a remarkable monograph, \cite{Trefethen2013} gives a comprehensive review of Chebyshev interpolation.
Its appealing theoretical properties are indeed of practical use as the software tool Chebfun\footnote{Chebfun is an open-source software system, see \url{http://www.chebfun.org}} demonstrates. In this implemention \cite{PlatteTrefethen2008} aim ``to combine the feel of symbolics with the speed of numerics". Therefore Chebyshev interpolation is our method of choice.\footnote{Chebyshev interpolation shares its good properties with for instance  Legendre transformation, for which we expect similarly positive results. We refer to \cite{Trefethen2013}, who states: "
It is the clustering near the ends of the interval that makes the difference, and other sets of points with similar clustering, like Legendre points [...] have similarly good behaviour."}

Exploring the potential of interpolation methods for more than one single free parameter, we choose a tensorized version of Chebyshev interpolation:
For parameters $p\in\rr^D$, where $D\in\nn$ denotes the dimensionality of the parameter space, the price $\Price^p$ is approximated by tensorized Chebyshev polynomials $T_j$ with pre-computed coefficients $c_j$, $j\in J$, as follows,
	\begin{equation*}
		\Price^p \approx \sum_{j\in J}c_j T_j(p).
	\end{equation*}

Chebyshev interpolation is a standard numerical method that has proven to be extremely useful for applications in such diverse fields as physics, engineering, statistics and economics. Nevertheless, for pricing tasks in mathematical finance Chebyshev interpolation still seems to be rarely used and its potential is yet to be unfolded. \cite{PistoriusStolte2012} use Chebyshev interpolation of \BS prices in the volatility as an intermediate step to derive a pricing methodology for a time-changed model. Independently from us, \cite{Pachon2016} recently proposed Chebyshev interpolation as a quadrature rule for the computation of option prices with a Fourier type representation, which is comparable to the cosine method.

Our main results are the following:
\begin{itemize}

\item Theorem \ref{ConvergenceFourierPrices} provides accessible sufficient conditions on options and models that guarantee an asymptotic error decay of order $O\big(\varrho^{-\sqrt[D]{N}}\big)$ in the total number $N$ of interpolation nodes where $\varrho>1$ is given by the domain of analyticity and $D$ is the number of varying parameters.

\item More specific conditions for parametric European options in L\'evy models are provided in Corollary~\ref{cor-EuroLevy}, while Corollary~\ref{cor-Basetaffine} provides the framework for parametric basket options in affine models. 
\end{itemize}
These results establish an error analysis that is based on the domain of analyticity of the prices as  functions of the parameters. Observing that typical payoff functions are not smooth, we cannot expect an exponential error decay for interpolation with arbitrarily small maturities. Small maturities thus serve as an example that domains of analyticity need to be carefully studied. 
\begin{itemize}
\item 
 The investigations in Sections \ref{sec-options}--\ref{sec:models} show that for a large set of relevant (basket) options, models and free parameters a domain of analyticity can indeed be identified. 
This gives examples of relevant financial applications where (sub)exponential error decay is guaranteed.
\end{itemize}
To numerically validate the theoretical results we compare prices obtained by Chebyshev interpolation to benchmark prices by Fourier techniques.
\begin{itemize}
\item 
Numerical experiments in affine models confirm the theoretical error decay empirically for European call options (Figure~\ref{fig:ChebyErrorDecayAllLog10scale_call}) and digital down{\&}out options (Figure~\ref{fig:ChebyErrorDecayAllLog10scale_digi}). For the considered model examples of \BS, Merton, CGMY and Heston we observe $L^\infty$-error levels of order $10^{-10}$ using not more than $N=(25+1)^2$ Chebyshev interpolation nodes when $D=2$ parameters are varied. 
\end{itemize}
Numerical results show that already a small number of nodes leads to high accuracy. This motivates us to further explore the potential of the Chebyshev method for multivariate options. Here we deliberately go beyond the scope of our theoretical results and consider additional features like path-dependency. 
\begin{itemize}
\item For multivariate basket and path-dependent options in the Black{\&}Scholes, Heston and Merton model we use Monte Carlo as reference method. In all of our settings in Section~\ref{accuracy_exotics} Chebyshev interpolation achieves an accuracy that is similar to the accuracy of the Monte Carlo simulation itself ($10^{-3}$) for $D=2$.
\end{itemize}
In  addition we present empirical results demonstrating the efficiency of the Chebyshev method.
\begin{itemize}
\item 
The gain in efficiency in comparison to Fourier techniques is first validated for bivariate options of European type in Section \ref{sec:ChebyshevEfficiency}.
\item Secondly, the explicit gains in efficiency in comparison to Monte Carlo methods are shown in Section \ref{sec:ChebyshevEfficiency_MC} taking multivariate lookback options in the Heston model as examples.
\end{itemize}

The remainder of the article is organized as follows. In Section \ref{sec-Cheby} we introduce Chebyshev interpolation in detail and present the general convergence results. Section \ref{sec-pop}  establishes a convergence analysis of Chebyshev interpolation for POP. We formulate analyticity conditions for the payoff profiles and models that guarantee (sub)exponential convergence of the method. These conditions are verified in Section \ref{sec-analyticity} for different option types, models and free parameters.
The numerical experiments in Section 5 confirm these findings using Fourier techniques. Pricing basket options, the gain in efficiency is numerically investigated. 
Experiments based on Monte Carlo and finite differences moreover suggest to further explore the potential of the approach beyond the scope of the theoretical investigations from the previous sections. The resulting conclusion and outlook are presented in Section 6. Finally, the appendix provides the proof of the multivariate convergence result.

\section{Chebyshev Polynomial Interpolation}\label{sec-Cheby}
In this section we introduce the notation for Chebyshev interpolation. Following \cite{Trefethen2013}, the one-dimensional version is shown. Then we present the multivariate extension and convergence results. Consider an option price with a single varying parameter 
\begin{align}\label{Price_p}
\Price^p, \qquad p\in [-1,1].
\end{align}
An interpolation of $\Price^p$ with Chebyshev polynomials of degree $N$ is of the form
 \begin{align}\label{eq-ChebInter1dim}
I_N(\Price^{(\cdot)})(p)&:= \sum_{j=0}^N c_j T_j(p), \\
\intertext{with coefficients}
c_j&:=\frac{2^{\1_{0<j<N}}}{N}\sum_{k=0}^{N}{}^{''} \Price^{p_k}\cos\Big(j\pi\frac{k}{N}\Big),\quad j\le N,
 \end{align}
and basis functions
  \begin{align}\label{eq-Tjp}
  T_j(p) := \cos\big(j \arccos(p)\big)\quad\text{for $p\in [-1,1]$ and $j\le N$}
   \end{align}
where $\sum{}^{''}$ indicates that the first and last summands are halved. The Chebyshev nodes ${p}_k = \cos\left(\pi\frac{k}{N}\right)$ may conveniently be displayed in a graph, see Figure \ref{fig:ChebyNodes1d}. For an arbitrary compact parameter interval $[\underline{p}, \overline{p}]$, interpolation \eqref{eq-ChebInter1dim} needs to be adjusted by the appropriate linear transformation.

\begin{figure}[htb!]
\includegraphics[scale=1, center]{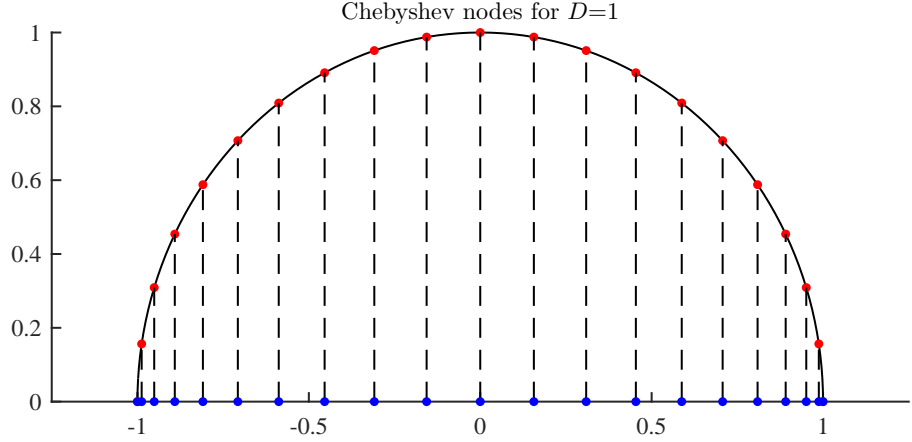}
\caption{A set of Chebyshev points $p_k\in [-1,1]$ (blue) for degree $N=20$ and equidistantly spaced auxiliary construction points (red) on the semi-circle.}
\label{fig:ChebyNodes1d}
\end{figure}

 \subsection{Multivariate Chebyshev Interpolation}\label{sec-multivariateChebyshev}
 \label{sec:Multi_Cheby}
The Chebyshev polynomial interpolation \eqref{eq-ChebInter1dim}--\eqref{eq-Tjp} has a tensor based extension to the multivariate case, see e.g. \cite{SauterSchwab2004}.
In order to obtain a nice notation, consider the interpolation of prices 
	\begin{equation}
		\Price^p,\qquad p\in [-1,1]^D.
	\end{equation}
	For a more general hyperrectangular parameter space $\OP=[\underline{p}_{1},\olp[1]]\times\ldots \times[\underline{p}_D,\olp[D]]$, the appropriate linear transformations need to be performed. Let $\overline{N}:=(N_1,\ldots,N_D)$ with $N_i \in\nn_0$ for $i=1,\ldots,D$. The interpolation with $\prod_{i=1}^D (N_{i}+1)$ summands is given by
	\begin{equation}
	I_{\overline{N}}(\Price^{(\cdot)})(p) := \sum_{j\in J} c_jT_j(p), 
	\end{equation}
where 
the summation index $j$ is a multiindex ranging over $J:=\{(j_1,\dots, j_D)\in\nn_0^D: j_i\le N_i\,\text{for }i=1,\ldots,D\}$, i.e.
\begin{equation}
	\label{eq:ChebySumd}
	I_{\overline{N}}(\Price^{(\cdot)})(p) = \sum_{j_1=0}^{N_1}\ldots \sum_{j_D=0}^{N_D} c_{(j_1,\ldots,j_D)} T_{(j_1,\ldots,j_D)}(p).
	\end{equation}
The basis functions $T_j$ for $j=(j_1,\dots, j_D)\in J$ are defined by
	\begin{equation}
	\label{def:ChebyTjd}
		T_{j}(p_1,\dots,p_D) = \prod_{i=1}^D T_{j_i}(p_i).
	\end{equation}	
The coefficients $c_j$ for $j=(j_1,\dots, j_D)\in J$ are given by
	\begin{equation}
	\label{def:Chebycj}
		c_j = \Big( \prod_{i=1}^D \frac{2^{\1_{\{0<j_i<N_i\}}}}{N_i}\Big)\sum_{k_1=0}^{N_1}{}^{''}\ldots\sum_{k_D=0}^{N_D}{}^{''} \Price^{p^{(k_1,\dots,k_D)}}\prod_{i=1}^D \cos\left(j_i\pi\frac{k_i}{N_i}\right),
	\end{equation}
where $\sum{}^{''}$ indicates that the first and last summand are halved and the Chebyshev nodes $p^k$ for multiindex $k=(k_1,\dots,k_D)\in J$ are given by
	\begin{equation}
	\label{eq:Chebynodesd}	
		p^k = (p_{k_1},\dots,p_{k_D})
	\end{equation}
	with the univariate Chebyshev nodes $p_{k_i}=\cos\left(\pi\frac{k_i}{N_i}\right)$ for $k_i=0,\ldots,N_i$ and $i=1,\ldots,D$. A set of $D$-variate Chebyshev nodes $p^{(k_1,\dots,k_D)}$ for $D=2$ and $N_1=N_2=20$ is depicted in Figure \ref{fig:ChebyNodes2d}.
\begin{figure}[htb!]
\includegraphics[scale=1, center]{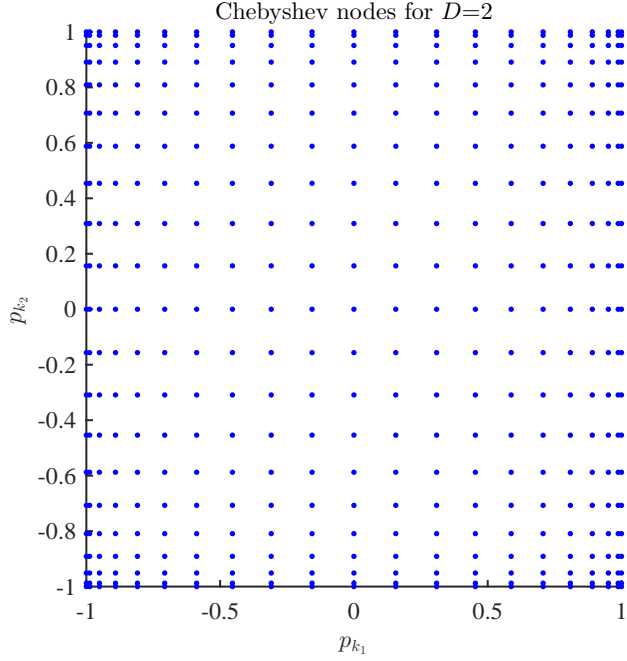}
\caption{A set of $D$-variate Chebyshev points $p^k\in [-1,1]^D$ for $D=2$ and $N_1=N_2=20$.
}
\label{fig:ChebyNodes2d}
\end{figure}

\newsavebox{\meinebox}
  \sbox{\meinebox}{\scalebox{0.925}[1]{\large\textbf{ Convergence of Multivariate Chebyshev Interpolation}}}
\subsection{\usebox{\meinebox}}
\label{sec-Cheby-convergence}
In the univariate case, it is well known that the error of approximation with Chebyshev polynomials decays polynomially for differentiable functions and exponentially for analytic functions. 
Let $f$ be analytic in $[-1,1]$ then it has an analytic extension to some \textit{Bernstein ellipse} $B([-1,1],\varrho)$ with parameter $\varrho>1$, defined as the open region in the complex plane bounded by the ellipse with foci $\pm 1$ and semiminor and semimajor axis lengths summing up to $\varrho$. 
This and the following result traces back to the seminal work of \cite{Bernstein1912}.
\begin{theorem}\label{Trefethen_1D_convergence}
\cite[Theorem 8.2]{Trefethen2013} Let a function f be analytic in the open Bernstein ellipse $B([-1,1],\varrho)$, with $\varrho>1$, where it satisfies $\vert f\vert\le V$ for some constant $V>0$. Then for each $N\ge 0$,
\begin{align*}
\Vert f-I_N(f)\Vert_{L^\infty([-1,1])}\le4 V\frac{\varrho^{-N}}{\varrho-1}.
\end{align*} 

\end{theorem}
In the multivariate case we will extend a convergence result from \cite{SauterSchwab2004}. 
We consider parametric option prices of form
\begin{align}
\Price^p \quad\text{for }p\in\OP
\end{align}
 with $\OP\subset\rr^D$ of hyperrectangular structure, i.e.\ $\OP=[\underline{p}_1,\olp[1]]\times\ldots\times[\underline{p}_D,\olp[D]]$ with real $\underline{p}_i\le\olp[i]$ for all $i=1,\ldots,D$.
We define the $D$-variate and transformed analogon of a Bernstein ellipse around the hyperrectangle $\OP$ with parameter vector $\varrho\in(1,\infty)^D$ as
\begin{align}\label{eq-genB}
B(\OP,\varrho):=B([\underline{p}_1,\olp[1]],\varrho_1)\times\ldots\times B([\underline{p}_D,\olp[D]],\varrho_D )
\end{align}
with $B([\underline{p},\olp],\varrho):=\tau_{[\underline{p},\olp]}\circ B([-1,1],\varrho)$, where for $p\in\cc$ we have the transform $\tau_{[\underline{p},\olp]}\big(\Re(p)\big):=\overline{p} + \frac{\underline{p}-\olp}{2}\big(1-\Re(p)\big)$ and $\tau_{[\underline{p},\olp]}\big(\Im(p)\big):= \frac{\olp-\underline{p}}{2}\Im(p)$. We call $B(\OP,\varrho)$ \textit{generalized Bernstein ellipse} if the sets $B([-1,1],\varrho_i)$ are Bernstein ellipses for $i=1,\ldots,D$.

%
  \begin{theorem}
 \label{Asymptotic_error_decay_multidim}
  Let $\OP\ni p\mapsto \Price^p$ be a real valued function that has an analytic extension to some generalized Bernstein ellipse $B(\OP,\varrho)$ for some parameter vector $\varrho\in (1,\infty)^D$ and $\sup_{p\in B(\OP,\varrho)}|\Price^p|\le V$. 
  Then
  \begin{align*}
 \max_{p\in\OP}\big|\Price^p - I_{\overline{N}}(\Price^{(\cdot)})(p)\big|
\le 2^{\frac{D}{2}+1}\cdot V \cdot\left(\sum_{i=1}^D\varrho_i^{-2N_i}\prod_{j=1}^D\frac{1}{1-\varrho_j^{-2}}\right)^{\frac{1}{2}}.
  \end{align*}
  \end{theorem}
  The proof of the theorem is provided in Appendix \ref{sec-proofAsymptotic}. 
  \begin{corollary}\label{cor-Asymptotic_error_Expo}
  Under the assumptions of Theorem \ref{Asymptotic_error_decay_multidim} there exists a constant $C>0$ such that
  \begin{align}\label{absch-TCanaOrder}
 \max_{p\in\OP}\big|\Price^p - I_{\overline{N}}(\Price^{(\cdot)})(p)\big|
 \le 
 C\underline{\varrho}^{-\underline{N}}, 
  \end{align}
  where $\underline{\varrho} = \min\limits_{1\le i\le D}\varrho_i$ and $\underline{N} = \min\limits_{1\le i\le D} N_i$. 
  \end{corollary}
  
  \begin{remark}
  In particular, under the assumptions required by Theorem \ref{Asymptotic_error_decay_multidim} with \mbox{$N=\prod_{i=1}^D (N_i+1)$} denoting the total number of nodes, Corollary \ref{cor-Asymptotic_error_Expo} shows that the error decay is of (sub)exponential order 
$O\big(\varrho^{-\sqrt[D]{N}}\big)$ for some $\varrho>1$.
  \end{remark}
  
In the setting of Theorem \ref{Asymptotic_error_decay_multidim} additionally the derivatives of $\Price^p$ are approximated by the according derivatives of the Chebyshev interpolation. The one-dimensional result is shown in \cite{Tadmor1986} and a multivariate result is derived in \cite{CanutoQuarteroni1982} for functions in Sobolev spaces. 
These results allow us to obtain the Chebyshev approximation of derivatives with no additional cost. 
To state the according convergence results we follow \cite{CanutoQuarteroni1982} and introduce the weighted Sobolev spaces for $\sigma\in \nn$ by
\begin{equation}
W_2^{\sigma,\omega}(\OP) = \left\{\phi\in L^2(\OP): \Vert \phi\Vert_{W_2^{\sigma,\omega}(\OP)}< \infty \right\}\label{Sobolev-Raum},
\end{equation}
with norm
\begin{equation}
\Vert \phi\Vert_{W_2^{\sigma,\omega}(\OP)}^2=\sum_{\vert\alpha\vert\le \sigma} \int_{\OP}\vert\partial^{\alpha}\phi(p)\vert^2 \omega(p) \dd p,\label{Sobolev-Norm}
\end{equation}
wherein $\alpha=(\alpha_1,\dots,\alpha_D)\in \nn_0^D$ is a multiindex and $\partial^\alpha = \partial^{\alpha_1} \cdots\partial^{\alpha_D} $ and the weight function $\omega$ on $\OP$ given by
\begin{equation*}
\omega(x):=\prod_{j=1}^D \omega(\tau^{-1}_{[\underline{p}_j,\overline{p}_j]}(x_j)),\qquad \omega(\tau^{-1}_{[\underline{p}_j,\overline{p}_j]}(x_j)):=(1-\tau^{-1}_{[\underline{p}_j,\overline{p}_j]}(x_j)^2)^{-\frac{1}{2}}
\end{equation*}
with $\tau_{[\underline{p}_j,\overline{p}_j]}(p)=\overline{p}_j + \frac{\underline{p}_j-\olp}{2}\big(1-p\big)$.
Then we apply the result of \cite[Theorem 3.1]{CanutoQuarteroni1982} in the following corollary.
  
\begin{corollary}\label{Corollary_Sensitivities_1}
Under the assumptions of Theorem \ref{Asymptotic_error_decay_multidim} for any $\frac{D}{2}<\sigma\in\nn$ and any $\sigma\geq\mu\in \nn_0$ there exists a constant $C>0$ such that
\begin{align*}
\Vert \Price^{(\cdot)} - I_{\overline{N}}(\Price^{(\cdot)})(\cdot)\Vert_{W_2^{\mu,\omega}(\OP)}\le C N^{2\mu-\sigma}\Vert\Price^{(\cdot)}\Vert_{W_2^{\sigma,\omega}(\OP)},
\end{align*}
\end{corollary}

\begin{proof}
In our setting we have $\OP\subset\rr^D$ of hyperrectangular structure. Under the assumptions of Theorem \ref{Asymptotic_error_decay_multidim} it follows that $p\mapsto\Price^p\in W_2^{\sigma}(\OP)$ and therewith $p\mapsto\Price^p\in W_2^{\sigma,\omega}(\OP)$. Before we apply \cite[Theorem 3.1]{CanutoQuarteroni1982}, which assumes $\OP=[-1,1]^D$, we investigate how the linear transformation $\tau_{\OP}$, as introduced in the proof of Theorem \ref{Asymptotic_error_decay_multidim}, influences the derivatives. 
Let $p\mapsto \Price^p$ be a function on $\OP$. We set $\widehat{h}(p)=\Price^p\circ\tau_{\OP}(p)$. Further, let $\widehat{I}_{\overline{N}}(\widehat{h})(p)$ be the Chebyshev interpolation of $\widehat{h}(p)$ on $[-1,1]^D$. Then, it directly follows
\begin{align*}
\Price^p - I_{\overline{N}}(\Price^{(\cdot)})(p)=\left(\widehat{h}(\cdot) -\widehat{I}_{\overline{N}}(\widehat{h})(\cdot)  \right)\circ \tau^{-1}_{\OP}(p).
\end{align*}
First, let us assume $D=1$, i.e. $\OP=[\underline{p},\overline{p}]$, and let $\alpha\in\mathbb{N}_0$. For the partial derivatives it holds
\begin{align*}
\partial^{\alpha}\Price^p - \partial^{\alpha}I_{\overline{N}}(\Price^{(\cdot)})(p)&=\partial^{\alpha}\left(\Price^p -I_{\overline{N}}(\Price^{(\cdot)})(p)\right)\\
&=\partial^{\alpha}\left(	\left(\widehat{h}(\cdot) -\widehat{I}_{\overline{N}}(\widehat{h})(\cdot)  \right)\circ \tau^{-1}_{\OP}(p)\right)\\
&=\partial^{\alpha-1}	\left(\partial^1\widehat{h}(\tau^{-1}_{\OP}(p)) -\partial^1\widehat{I}_{\overline{N}}(\widehat{h}^{(\cdot)})(\tau^{-1}_{\OP}(p))  \right)\\
&=\partial^{\alpha-1}\frac{2}{\overline{p}-\underline{p}}\left(\left[\partial^1\widehat{h}\right](\tau^{-1}_{\OP}(p)) -\left[\partial^1\widehat{I}_{\overline{N}}(\widehat{h}^{(\cdot)})\right](\tau^{-1}_{\OP}(p))  \right).
\end{align*}
Repeating this step iteratively yields
\begin{align*}
\partial^{\alpha}\Price^p - \partial^{\alpha}I_{\overline{N}}(\Price^{(\cdot)})(p)=\frac{2^{\alpha}}{(\overline{p}-\underline{p})^{\alpha}}\left(\left[\partial^{\alpha}\widehat{h}\right](\tau^{-1}_{\OP}(p)) -\left[\partial^{\alpha}\widehat{I}_{\overline{N}}(\widehat{h}^{(\cdot)})\right](\tau^{-1}_{\OP}(p))  \right).
\end{align*}
Therewith, the error on $[-1,1]$ is scaled with a factor $\frac{2^{\alpha}}{(\overline{p}-\underline{p})^{\alpha}}$. Extending this to the D-variate case with, this analogously results with $\alpha=(\alpha_1,\dots,\alpha_D)\in \nn_0^D$ is a multiindex and $\partial^\alpha = \partial^{\alpha_1} \cdots\partial^{\alpha_D} $
\begin{align*}
\partial^{\alpha}\Price^p - \partial^{\alpha}I_{\overline{N}}&(\Price^{(\cdot)})(p)=\\
&\prod_{i=1}^D\frac{2^{\vert\alpha_i\vert}}{(\overline{p}_i-\underline{p}_i)^{\vert\alpha_i\vert}}\left(\left[\partial^{\alpha}\widehat{h}\right](\tau^{-1}_{\OP}(p)) -\left[\partial^{\alpha}\widehat{I}_{\overline{N}}(\widehat{h}^{(\cdot)})\right](\tau^{-1}_{\OP}(p))  \right).
\end{align*}
 From Theorem\tild 3.1 in \cite{CanutoQuarteroni1982} the assertion follows directly for $\widehat{h}(\cdot)\text{ on }\OP=[-1,1]^D$, i.e.
for any $\frac{D}{2}<\sigma\in\nn$ and any $\sigma\geq\mu\in \nn_0$ there exists a constant $\tilde{C}>0$ such that
\begin{align}
\Vert \widehat{h}(\cdot) -\widehat{I}_{\overline{N}}(\widehat{h})(\cdot) \Vert_{W_2^{\mu,\omega}(\OP)}\le \tilde{C} N^{2\mu-\sigma}\Vert\widehat{h}(\cdot)\Vert_{W_2^{\sigma,\omega}(\OP)},\label{Canuto_1}
\end{align}  
For arbitrary $\OP$ the constant from \ref{Canuto_1} has to be multiplied with the according factor resulting from the linear transformation $\tau_{\OP}$. 
\end{proof}
The result in Corollary \ref{Corollary_Sensitivities_1} is given in terms of weighted Sobolev norms. In the following remark we connect the approximation error in the weighted Sobolev norm to the $C^l(\OP)$ norm, where $C^l(\OP)$ is the Banach space of all functions $u$ in $\OP$ such that $u$ and $\partial^\alpha u$ with $|\alpha|\leq l$ are uniformly continuous in $\OP$ and the norm
	\begin{equation*}
		\lVert u \rVert_{C^l(\OP)} = \max_{|\alpha|\leq l}\max\limits_{p\in\OP}|\partial^\alpha u(p)|
	\end{equation*}
is finite.
\begin{corollary}
Under the assumptions of Theorem \ref{Asymptotic_error_decay_multidim} for any $\frac{D}{2}<\sigma\in\nn$ and any $\sigma\geq \mu\in \nn_0$ and $l\in \nn_0$ with $\mu-l>\frac{D}{2}$, there exists a constant $\bar{C}(\sigma)>0$ depending on $\sigma$, such that
\begin{align*}
\Vert \Price^{(\cdot)} - I_{\overline{N}}(\Price^{(\cdot)})(\cdot)\Vert_{C^{l}(\OP)}\le \bar{C}(\sigma) N^{2\mu-\sigma}\max_{\vert\alpha\vert\le \sigma}\sup_{p\in\OP}\vert\partial^{\alpha}\Price^p\vert.
\end{align*}
\end{corollary}
\begin{proof}
In the setting of Corollary \ref{Corollary_Sensitivities_1}, we start with the estimation of the approximation error in the weighted Sobolev norms,
\begin{align}
\Vert \Price^{(\cdot)} - I_{\overline{N}}(\Price^{(\cdot)})(\cdot)\Vert_{W_2^{\mu,\omega}(\OP)}\le C N^{2\mu-\sigma}\Vert\Price^{(\cdot)}\Vert_{W_2^{\sigma,\omega}(\OP)}.\label{Norms_both_sides}
\end{align}
On $\OP$ it holds that $w(p)\ge1$ and therewith we can deduce for the Sobolev norm with a constant weight of 1,
\begin{align*}
\Vert \Price^{(\cdot)} - I_{\overline{N}}(\Price^{(\cdot)})(\cdot)\Vert_{W_2^{\mu,\omega}(\OP)}\ge\Vert \Price^{(\cdot)} - I_{\overline{N}}(\Price^{(\cdot)})(\cdot)\Vert_{W_2^{\mu,1}(\OP)}.
\end{align*}
With $W_2^{\mu}(\OP)$ the usual Sobolev space,
\begin{equation*}
W_2^{\mu}(\OP) = \left\{\phi\in L^2(\OP): \Vert \phi\Vert_{W_2^{\mu}(\OP)}< \infty \right\},\quad
\Vert \phi\Vert_{W_2^{\mu}(\OP)}^2=\sum_{\vert\alpha\vert\le \mu} \int_{\OP}\vert\partial^{\alpha}\phi(p)\vert^2 \dd p,
\end{equation*}
Corollary 6.2 from \cite{Wloka-english} directly yields that for any $l$ with $\mu-l>\frac{D}{2}$ there exists a constant $\tilde{C}$ such that
\begin{align}
\Vert \Price^{(\cdot)} - I_{\overline{N}}(\Price^{(\cdot)})(\cdot)\Vert_{C^{l}(\OP)}\le\tilde{C}\Vert \Price^{(\cdot)} - I_{\overline{N}}(\Price^{(\cdot)})(\cdot)\Vert_{W_2^{\mu}(\OP)}.\label{Norms_left_side}
\end{align}
In formula \eqref{Norms_left_side} we have derived a lower bound for the left hand side of expression \eqref{Norms_both_sides}. Next, we will find an upper bound for the right hand side of \eqref{Norms_both_sides}. From the definition of the weighted Sobolev norm, see \eqref{Sobolev-Norm}, it follows
\begin{align*}
\Vert \Price^{(\cdot)}\Vert_{W_2^{\sigma,\omega}(\OP)}&=\sqrt{\sum_{\vert\alpha\vert\le \sigma} \int_{\OP}\vert\partial^{\alpha}\Price^p\vert^2 \omega(p) \dd p}\\
&\le \sqrt{\sum_{\vert\alpha\vert\le \sigma} \sup_{p\in\OP}\vert\partial^{\alpha}\Price^p\vert^2 \int_{\OP} \omega(p) \dd p}.
\end{align*}
Here, we apply $\int_{\OP} \omega(p) \dd p=\pi^D$ and that there exists a constant $\alpha^2(\sigma)$ depending on $\sigma$ such that
\begin{align}
\Vert \Price^p\Vert_{W_2^{\sigma,\omega}(\OP)}&\le \sqrt{\alpha^2(\sigma)\max_{\vert\alpha\vert\le \sigma}\sup_{p\in\OP}\vert\partial^{\alpha}\Price^p\vert^2 \pi^D}\notag\\
&= \alpha(\sigma)\max_{\vert\alpha\vert\le \sigma}\sup_{p\in\OP}\vert\partial^{\alpha}\Price^p\vert \pi^\frac{D}{2}.\label{Norms_right_side}
\end{align}
Finally, using \eqref{Norms_left_side} and \eqref{Norms_right_side} in \eqref{Norms_both_sides} yields an estimate of the approximation error in the $C^l(\OP)$ norm,
\begin{align*}
\frac{1}{\tilde{C}}\Vert \Price^{(\cdot)} - I_{\overline{N}}(\Price^{(\cdot)})(\cdot)\Vert_{C^{l}(\OP)}\le C N^{2\mu-\sigma}\alpha(\sigma)\max_{\vert\alpha\vert\le \sigma}\sup_{p\in\OP}\vert\partial^{\alpha}\Price^p\vert \pi^\frac{D}{2}.
\end{align*}
Collecting all constants in $\bar{C}(\sigma)$, we achieve 
\begin{align*}
\Vert \Price^{(\cdot)} - I_{\overline{N}}(\Price^{(\cdot)})(\cdot)\Vert_{C^{l}(\OP)}\le \bar{C}(\sigma) N^{2\mu-\sigma}\max_{\vert\alpha\vert\le \sigma}\sup_{p\in\OP}\vert\partial^{\alpha}\Price^p\vert.
\end{align*}
\end{proof}

\section{Exponential Convergence of Chebyshev \\ Interpolation for POP}\label{sec-pop} 

In this section we embed the multivariate Chebyshev interpolation into the option pricing framework. We provide sufficient conditions under which option prices depend analytically on the parameters. 

Analytic properties of option prices can be conveniently studied in terms of Fourier transforms. First, Fourier representations of option prices are explicitly available for a large class of both option types and asset models. Second, Fourier transformation unveils the analytic properties of both the payoff structure and the distribution of the underlying stochastic quantity in a beautiful way. By contrast, if option prices are represented as expectations, their analyticity in the parameters is hidden.
For example the function $K\mapsto(S_T-K)^+$ is not even differentiable, whereas the Fourier transform of the dampened call payoff function evidently is analytic in the strike, compare Table \ref{table-exEuropayoff} on page~\pageref{table-exEuropayoff}.

\subsection{Conditions for Exponential Convergence}

Let us first introduce a general option pricing framework. We consider option prices of the form
\begin{align}\label{def-Price}
\Price^{p=(p^1,p^2)}=E\big(f^{p^1}(X^{p^2})\big)
\end{align}
 where $f^{p^1}$ is a parametrized family of 
measurable payoff functions $f^{p^1}:\R^d\rightarrow\R_+$ with payoff parameters $p^1\in\OP^1$ and $X^{p^2}$ is a family of $\R^d$-valued random variables with model parameters $p^2\in\OP^2$. The parameter set
\begin{align}\label{def-para}
p=(p^1,p^2)\in\OP=\OP^1\times\OP^2\subset \rr^D
\end{align}
is again of hyperrectangular structure, i.e.\ $\OP^1=[\underline{p}_1,\olp[1]]\times\ldots\times[\underline{p}_m,\olp[m]]$ and $\OP^2=[\underline{p}_{m+1},\olp[m+1]]\times\ldots\times[\underline{p}_D,\olp[D]]$ for some $1\le m \le D$ and real $\underline{p}_i\le\olp[i]$ for all $i=1,\ldots,D$.

Typically we are given a parametrized $\R^d$-valued driving
stochastic process $H^{p'}$ with $S^{p'}$ being the vector of asset price processes modeled as an exponential of $H^{p'}$, i.e.\ 
\begin{align}\label{asset-Rd}
S^{p',\icc}_{t}=S^{p',\icc}_0\exp(H^{p',\icc}_{t}), \qquad 0\leq t\leq T,\quad \;1\le \icc\le d,
\end{align}
and $X^{p^2}$ is an $\OF_T$-measurable $\rr^d$-valued
random variable, possibly depending on the history of the $d$
driving processes, i.e.\ $p^2=(T,p')$ and
\begin{align*}
X^{p^2}:= \Psi\big(H_t^{p'},\,0\le t\le T\big),
\end{align*}
where $\Psi$ is an $\R^d$-valued measurable functional. 

We now focus on the case that the price \eqref{def-Price} is given in terms of Fourier transforms. This enables us to provide sufficient conditions under which the parametrized prices have an analytic extension to an appropriate generalized Bernstein ellipse.
For most relevant options, the payoff profile $f^{p^1}$ is not integrable and its Fourier transform over the real axis is not well defined. 
Instead, there exists an exponential dampening factor $\eta\in\rrd$ such that $\ee{\langle \eta,\cdot\rangle}f^{p^1}\in L^1(\rr^d)$. We therefore introduce exponential weights in our set of conditions and denote the Fourier transform of $g\in L^1(\rrd)$ by
\begin{align*}
\hat{g}(z):= \int_{\rrd} \ee{i\skl z, x\skr} g(x) \dd x
\end{align*}
and we denote the Fourier transform of $\ee{\langle \eta,\cdot\rangle}f\in L^1(\rr^d)$ by $\hat{f}(\cdot-i\eta)$. The exponential weight of the payoff will be compensated by exponentially weighting the distribution of $X^{p^2}$ and that weight will reappear in the argument of $\varphi^{p^2}$, the characteristic function of $X^{p^2}$.

\begin{conditions}
\label{conditionsA1234}
\emph{ Let parameter set $\OP=\OP^1\times\OP^2\subset \rr^D$ of hyperrectangular structure as in \eqref{def-para}. 
Let $\varrho\in(1,\infty)^D$ and denote $\varrho^1:=(\varrho_1,\ldots,\varrho_m)$ and $\varrho^2:=(\varrho_{m+1},\ldots,\varrho_D)$ and let weight $\eta\in\rrd$.
\vspace{-1.5ex}
\begin{enumerate}[leftmargin=3em, label=(A\arabic{*}),widest=(A4)]\label{A2-A4}
\item
For every $p^1\in \OP^1$ the mapping $x\mapsto \ee{\skl \eta,x\skr}f^{p^1}(x)$ is in $L^1(\rrd$).\label{A1}
\item
For every $z\in\rrd$ the mapping $p^1\mapsto \widehat{f^{p^1}}(z-i\eta)$ is analytic in the generalized Bernstein ellipse $B(\OP^1,\varrho^1)$ and there are constants $c_1,c_2>0$ such that
$\sup_{p^1\in B(\OP^1,\varrho^1)}|\widehat{f^{p^1}}(-z-i\eta)|\le c_1e^{c_2|z|}$ for all $z\in\rrd$. \label{A2}
\item
For every $p^2\in \OP^2$ 
the exponential moment condition $E\big(\ee{-\langle \eta,X^{p^2}\rangle}\big)<\infty$ holds.\label{A3}
\item
For every $z\in\rrd$ the mapping $p^2\mapsto \varphi^{p^2}(z+i\eta)$ is analytic in the generalized Bernstein ellipse  $B(\OP^{2},\varrho^2)$ and there are constants $\alpha\in(1,2]$ and $c_1,c_2>0$ such that $\sup_{p^2 \in B(\OP^{2},\varrho^2)}|\varphi^{p^2}(z+i\eta)|\le c_1\ee{-c_2 |z|^\alpha}$ for all $z\in\rrd$.
 \label{A4}
\end{enumerate}
}
\end{conditions}
Conditions (A1)--(A4) are satisfied for a large class of payoff functions and asset models, see Sections \ref{sec-Payoffs} and \ref{sec-models}.

\begin{theorem}\label{ConvergenceFourierPrices}
Let $\varrho\in(1,\infty)^D$ and weight $\eta\in \rrd$. Under conditions (A1)--(A4) we have
\begin{align*}
 \max_{p\in\OP}\big|\Price^p - I_{\overline{N}}(\Price^{(\cdot)})(p)\big|
\le 2^{\frac{D}{2}+1}\cdot V \cdot\left(\sum_{i=1}^D\varrho_i^{-2N_i}\prod_{j=1}^D\frac{1}{1-\varrho_j^{-2}}\right)^{\frac{1}{2}}.
  \end{align*}
\end{theorem}
The proof of the theorem builds on the following proposition that can be derived from Eberlein, Glau and Papapantoleon (2010, Theorem 3.2)\nocite{EberleinGlauPapapantoleon2010a}, observing that the proof solely uses that $z\mapsto \widehat{f^{p^1}}(-z-i\eta)  \varphi^{p^2}(z+i\eta)$ belongs to $L^1(\rrd)$ instead of the slightly stronger assumption that $z\mapsto \varphi^{p^2}(z+i\eta)$ belongs to $L^1(\rrd)$ which is imposed there. 
\begin{proposition}\label{FourierPrices}
Assume conditions (A1), (A3) and that the mapping $z\mapsto \widehat{f^{p^1}}(-z-i\eta)  \varphi^{p^2}(z+i\eta)$ belongs to $L^1(\rrd)$ for every $p=(p^1,p^2)\in\OP$, then the option prices~\eqref{def-Price} have for every $p=(p^1,p^2)\in\OP$ the Fourier transform based representation 
\begin{align}\label{FT-Price}
\Price^p=\frac{1}{(2\pi)^d} \int_{\rrd + i\eta}  \widehat{f^{p^1}}(-z) \varphi^{p^2}(z) \dd z.
\end{align}
\end{proposition} 
We are now in a position to prove Theorem \ref{ConvergenceFourierPrices}.
\begin{proof}[Proof of Theorem \ref{ConvergenceFourierPrices}]
In view of Theorem \ref{Asymptotic_error_decay_multidim}, it suffices to show that the mapping $p\mapsto \Price^p$ has an analytic extension to the generalized Bernstein ellipse $B(\OP,\varrho)$. Thanks to Proposition \ref{FourierPrices}, we have the following Fourier representation of the option prices,
\begin{align*}
\Price^p=\frac{1}{(2\pi)^d} \int_{\rrd + i\eta}  \widehat{f^{p^1}}(-z) \varphi^{p^2}(z) \dd z.
\end{align*}
Due to assumptions (A2) and (A4) the mapping
\[
p=(p^1,p^2)\mapsto \widehat{f^{p^1}}(-z) \varphi^{p^2}(z)
\]
has an analytic extension to $B(\OP,\varrho)$.
Let $\gamma$ be a contour of a compact triangle in the interior of $B([\underline{p}_i,\olp[i]],\varrho_i)$ for arbitrary $i=1,\ldots,D$. Then thanks to assumption~(A2) and~(A4) we may apply Fubini's theorem and obtain
\begin{align*}
\int_\gamma \Price^{(p_1,\ldots,p_D)}(z)\dd p_i &= \frac{1}{(2\pi)^d}  \int_\gamma \int_{\rrd + i\eta}  \widehat{f^{p^1}}(-z) \varphi^{p^2}(z) \dd z \dd p_i\\
& =  \frac{1}{(2\pi)^d}  \int_{\rrd + i\eta}  \int_\gamma  \widehat{f^{p^1}}(-z) \varphi^{p^2}(z) \dd p_i  \dd z=0.
\end{align*}
Moreover, thanks to assumptions (A2) and (A4), dominated convergence shows continuity of $p\mapsto \Price^p$ in  $B(\OP,\varrho)$ which yields the analyticity of $p\mapsto \Price^p$ in  $B(\OP,\varrho)$ thanks to a version of Morera's theorem provided in \cite[Satz 8]{jaenich2004funktionentheorie}.
\end{proof}

Similar to Corollary \ref{Corollary_Sensitivities_1} in the setting of Theorem \ref{ConvergenceFourierPrices} additionally the according derivatives are approximated as well by the Chebyshev interpolation. A very interesting application of this result in finance is the computation of sensitivities like delta or vega of an option price for risk assessment purposes. 
Theorem \ref{ConvergenceFourierPrices} together with Corollary \ref{Corollary_Sensitivities_1} yield the following corollary.

\begin{corollary}\label{Corollary_Sensitivities_2}
Under the assumptions of Theorem \ref{ConvergenceFourierPrices}, for all $\ l\in\mathbb{N}$, $\mu$ and $\sigma$ with $\sigma>\frac{D}{2}$, $0\le\mu\le\sigma$ and $\mu-l>\frac{D}{2}$ there exist a constant $C$, such that
\begin{align*}
\Vert \Price^p - I_{\overline{N}}(\Price^{(\cdot)})(p)\Vert_{C^{l}(\OP)}\le C \underline{N}^{2\mu-\sigma}\Vert\Price^p\Vert_{W_2^{\sigma}(\OP)},
\end{align*}
where the spaces and norms are defined in Section \ref{sec-Cheby-convergence}.
\end{corollary}

\begin{remark}\label{rem-cond-para}
The upper bound in condition (A2) is tailored to the application to payoff profiles with varying option parameters, compare Section \ref{sec-Payoffs} and particularly Table \ref{table-exEuropayoff}, below.  Examples of (time-inhomogeneous) L\'evy processes satisfying the upper bound in (A4) for $\OP^2=[\underline{T},\overline{T}]\subset(0,\infty)$ are provided in \cite{Glau2015b}, where also its connection to the parabolicity of the Kolmogorov equation of the process is investigated. In a fix parameter setting, condition (A3) for $\alpha\in(1,2]$ in \cite{Glau2015b} translates to our upper bound in (A4). This condition is satisfied for a large variety of models.
\end{remark}

\subsection{\textbf{Convergence Results for Selected Option Prices}}\label{sec-options}

In the previous subsection, Conditions~\ref{conditionsA1234} and Theorem~\ref{ConvergenceFourierPrices} introduced a framework in which the Chebyshev approximation achieves (sub)exponential error decay. Now we relate this abstract framework to two specific settings.

\subsubsection{European Options in Univariate L\'evy Models}\label{sec:EO_Univariate_Levy}
Let $r$ be the deterministic and constant interest rate.
We consider the parametrized family of asset prices, 
\begin{align}
\label{eq:LevyAssetProcess}
S_t^\pi:=S_0\ee{L_t^\pi}
\end{align}
with $t\ge0$. For fixed $\pi=(\sigma,b)$ let $L^\pi$ be a L\'evy process with characteristics $(\sigma^2,b,F)$ and $\int_{|x|>1}|x|F(\dd x)<\infty$. Due to the Theorem of L\'evy-Khintchine we have the following representation of the characteristic function of the parametrized L\'evy process
 \begin{align}\label{LevyKhinchine}
E\big(\ee{izL_t^\pi}\big) =\ee{t\psi^\pi(z)},\quad \psi^\pi(z)= -\frac{\sigma^2 z^2}{2} + ibz + \int_\rr\big(\ee{izx} - 1 - izx\big) F(\dd x).
 \end{align}
We assume $L^\pi$ is defined under a risk neutral measure. Therefore, for every $\pi\in\Pi$ we assume $E(\ee{L_t^{\pi}})<\infty$ for some and equivalently all $t>0$ and the drift condition
\begin{align}
\label{eq:LevyDriftCond}
b=b(r,\sigma) = r -\frac{\sigma^2}{2} - \int_\rr\big(\ee{x} - 1 -x\big) F(\dd x),
\end{align}
to ensure that the discounted asset price process is a martingale. Denoting  
\begin{align}
\widetilde\psi(z):= \int_\rr\big(\ee{izx} - 1 - izx\big) F(\dd x),
\end{align}
 the pure jump part of the exponent of the characteristic function, this reads as
\begin{align}
b=b(r,\sigma) = r -\frac{\sigma^2}{2} - \widetilde\psi(-i).
\end{align}
The fair value at time $t=0$ of a European option with payoff function 
 $f^K$ for $K\in\OP^1:=[\underline{K},\overline{K}]\subset\rr$ with maturity $T\in[\underline{T},\overline{T}]\subset(0,\infty)$ is given by
\begin{align}
\Price^{(r, K, S_0,T,\pi)}= \ee{-rT} E\big(f^K(S_0 \ee{L^\pi_T}) \big) .
\end{align}
In order to guarantee (sub)exponential convergence of the Chebyshev interpolation, in the following corollary, we translate the exponential moment condition (A3), the analyticity condition and the upper bound in (A4) to conditions on the cumulant generating function $\psi^\pi$. 
\begin{corollary}\label{cor-EuroLevy}
Let Conditions (A1) and \ref{A2} be satisfied for weight $\eta\in\rr$ and $\varrho^1>1$ and set $\OP^1=[\underline{K},\overline{K}]$. Moreover, let $\OP^2=[\underline{S_0},\overline{S_0}]\times [\underline{T},\overline{T}]\subset\rr^2$ with $\underline{S_0},\underline{T}>0$. Assume 
\begin{align}\label{expoLevy}
\int_{|x|>1} (\ee{-\eta x}\vee \ee{x}) F(\dd x) < \infty
\end{align}
and there exist constants $C_1,C_2>0$ and $\beta<2$ such that
\begin{align}\label{Im-jump-Levy}
\Big|\Im\big(\widetilde\psi(z + i \eta)\big)\Big| \le C_1 + C_2|z|^\beta \qquad\text{for all }z\in\rr.
\end{align}
If additionally one of the following conditions is satisfied,
\begin{itemize}
\item[(i)]
$\sigma>0$,
\item[(ii)]
there exist $\alpha\in(1,2]$ with $\beta<\alpha$ and constants $C_1,C_2>0$ such that
\begin{align*}
\Re\big(\widetilde{\psi}\big)(z+i\eta) \le C_1 - C_2|z|^\alpha\qquad\text{for all }z\in\rr,
\end{align*}
\end{itemize}
\noindent then there exist constants $C>0$ and $\underline{\varrho}>1$ such that
  \begin{align}\label{absch-TCanaOrder2}
 \max_{p\in\OP^1\times\OP^2}\big|\Price^p - I_{\overline{N}}(\Price^{(\cdot)})(p)\big|
 \le 
	C\underline{\varrho}^{-\underline{N}}, 
  \end{align}
  where 
  $\underline{N} = \min\limits_{1\le i\le 3} N_3$. 
  
\end{corollary}
\begin{proof}
In view of Theorem \ref{ConvergenceFourierPrices} and Corollary \ref{cor-Asymptotic_error_Expo}, under the hypothesis of this corollary, it suffices to verify Conditions (A3) and (A4). 
Thanks to the exponential moment condition \eqref{expoLevy}, \cite[Theorem~25.17]{Sato} implies that the validity of the L\'evy-Khintchine formula \eqref{LevyKhinchine} extends to the complex strip $U:=\rr + i\left([0,1]\cup\big(\sign(\eta)[0,|\eta|]\big)\right)$, i.e.\ for every $\pi\in\Pi$ and every $t>0$,
\begin{align}\label{LevyKhinchineU}
E\big(\ee{izL_t^\pi}\big) =\ee{t\psi^\pi(z)}\qquad\text{for every }z\in U.
 \end{align}
In particular, $E(\ee{-\eta L_t^\pi}) =\ee{t\psi^\pi(i\eta)}<\infty$ for every $\pi\in \Pi$, which yields (A3).
Denote $s_0:=\log(S_0)$.
 The assumptions of the corollary yield for every $z\in U$ the analyticity of 
\[
(s_0,t)\mapsto \varphi^{p^2=(\ee{s_0},t)}(z)=S_0^{iz}\ee{t\psi^{\pi}(z)}
\] 
on the generalized Bernstein ellipse $B(\OP^2,\varrho^2)$ for some parameter $\varrho^2\in(1,\infty)^{2}$, i.e.\ the validity of the analyticity Condition in (A4). Next we consider
\begin{align*}
\big|\varphi^{p^2}(z+i\eta)\big|
&\le 
\big|\ee{i s_0 (z + i\eta)}\big|\big|\ee{t\psi^{\pi}  (z + i\eta)}\big|\\
&\le \exp\left\{- s_0 \eta - \Im(s_0) z + \Re(t)\Re\big(\psi^{\pi}(z+i\eta)\big) - \Im(t)\Im\big(\psi^{\pi}(z+i\eta)\big) \right\}.
\end{align*}
An elementary manipulation shows for every $z\in\rr$,
\begin{align*}
\Re\big(\psi^{\pi}(z+i\eta)\big)
&=
\psi^{\pi}(i\eta) -\frac{\sigma^2 z^2}{2} + \int_{\rr}\big(\cos(zx)-1\big) \ee{-\eta x} F(\dd x)
\end{align*}
which in combination with the assumption on the imaginary part of $\tilde{\psi}$ yields the upper bound in Condition (A4) under assumption~(ii). In view of
\begin{align*}
 \int_{\rr}\big(\cos(zx)-1\big) \ee{-\eta x} F(\dd x)\le 0,
\end{align*}
 condition (A4) also follows under assumption~(i).
\end{proof}
Let us notice that for the Merton model, Condition~\eqref{expoLevy} is satisfied for every $\eta\in\rr$ and Condition~\eqref{Im-jump-Levy} is satisfied with $\beta=1$.  Examples of L\'evy jump diffusion modes, i.e.\ examples satisfying condition (i) of Corollary~\ref{cor-EuroLevy} are e.g.\ the \BS and the Merton model.
Examples of pure jump L\'evy models satisfying condition (ii) of Corollary~\ref{cor-EuroLevy} are provided in \cite{Glau2015b}, compare also Remark~\ref{rem-cond-para} and Table~\ref{table-processes} below. 

For simplicity and in accordance with our numerical experiments, we fixed the model parameters $\pi$. The analysis can be extended to varying model parameters $\pi$.

\begin{remark}
\label{rem:LevyEuroDecay}
Assuming \eqref{expoLevy}, we observe that the mappings
\begin{align}
(r,\sigma)\mapsto b(r,\sigma) = r -\frac{\sigma^2}{2} - \int_\rr\big(\ee{x} - 1 -x\big) F(\dd x)
\end{align}
and 
$(r,\sigma)\mapsto \psi^{\pi=(b(r,\sigma),\sigma)}(z)$ for every $z\in U$
are holomorphic.
Moreover, $(r,\sigma)\mapsto b(r,\sigma)$ is 
bounded on every generalized Bernstein ellipse $B([\underline{r},\overline{r}]\times[\underline{\sigma},\overline{\sigma}],\varrho)$ with $\varrho\in (1,\infty)^2$.

\end{remark}

\begin{remark}\label{cor-callLevy}
Corollary \ref{cor-EuroLevy} allows us to determine explicit error bounds for call options in L\'evy models. The fair price at $t=0$ of a call option with strike $K$ and maturity $T$ with deterministic interest rate $r\ge0$ is 
\begin{align}
\Call^{S_0,K}_T=\ee{-rT}E\big(S_0\ee{L_T} - K\big)^+
\end{align}
under a risk-neutral probability measure. Noticing that
\begin{align}\label{Stoprice}
\Call^{S_0,K}_T=\ee{-rT}KE\big((S_0/K)\ee{L_t} -1\big)^+ ,
\end{align}
it suffices to interpolate the function
\begin{align}
(T,S_0)\mapsto \Call^{S_0,1}_T
\end{align}
on $[\underline{T},\overline{T}]\times[\underline{S_0}/\overline{K},\overline{S_0}/\underline{K}]$ in order to approximate the prices $\Call^{S_0,K}_T$ for values $(T,S_0,K)\in[\underline{T},\overline{T}]\times[\underline{S_0},\overline{S_0}]\times[\underline{K},\overline{K}]\subset(0,\infty)^3$. This effectively reduces the dimensionality $D$ of the interpolation problem by one.

Let us fix some $\eta<-1$ and let $\OP=[\underline{T},\overline{T}]\times[\underline{S_0}/\overline{K},\overline{S_0}/\underline{K}]$, $\zeta^1=\frac{\overline{S_0}\underline{K}+\underline{S_0}\overline{K}}{\overline{S_0}\underline{K}-\underline{S_0}\overline{K}}$ and  $\zeta^2=\frac{\overline{T}+\underline{T}}{\overline{T}-\underline{T}}$, 
then for every $\varrho^j\in(1,\zeta^j+\sqrt{(\zeta^j)^2-1})$ for $j=1,2$, let $\varrho = (\varrho^1, \varrho^2)$ and
	\begin{equation*}
		V:=\sup\limits_{(T,{S_0})\in B(\OP,\varrho)} \left|\Call^{{S_0},1}_T\right|,
	\end{equation*}
then
\begin{align*}
\max_{(T,S_0)\in\OP}\big|\Call^{S_0,1}_T - I_{N_1,N_2}(\Call^{(\cdot),1}_{(\cdot)})(T, S_0)\big| \le 4V\left(\frac{\varrho_1^{-2N_1}+\varrho_2^{-2N_2}}{(1-\varrho_1^{-2})\cdot(1-\varrho_2^{-2})}\right)^{\frac{1}{2}}.
\end{align*}
In particular, there exists a constant $C>0$ such that
  \begin{align}\label{absch-TCanaOrderCallLevy}
\max_{(T,S_0)\in\OP}\big|\Call^{S_0,1}_T - I_{N_1,N_2}(\Call^{(\cdot),1}_{(\cdot)})(T, S_0)\big| 
 \le 
	C\underline{\varrho}^{-\underline{N}}, 
  \end{align}
  where $\underline{\varrho} = \min\{\varrho_1, \varrho_2\}$ and $\underline{N} =  \min\{N_1, N_2\}$.\\ 
Additionally, when fixing the maturity $T$, letting 
$\zeta=\frac{\overline{S_0}\underline{K}+\underline{S_0}\overline{K}}{\overline{S_0}\underline{K}-\underline{S_0}\overline{K}}$, we obtain the exponential error decay
\begin{align}\label{absch-TCanaOrderCallLevy2}
\max_{\underline{S_0}/\overline{K}\le S_0\le\overline{S_0}/\underline{K}}\big|\Call^{S_0,1}_T - I_{N}(\Call^{(\cdot),1}_T)(S_0)\big| 
 \le 
 4V\frac{\varrho^{-N}}{\varrho-1},
  \end{align}
for some $\varrho\in(1,\zeta+\sqrt{\zeta^2-1})$ and $V=\sup\limits_{{S_0}\in B([\underline{S_0}/\overline{K},\overline{S_0}/\underline{K}],\varrho)} \left|\Call^{{S_0},1}_T\right|$.
\end{remark}

\subsubsection{Basket Options in Affine Models}
Let $X^{\pi'}$  be a parametric family of affine processes with state space $\mathfrak{D}\subset\rrd$ for $\pi'\in\Pi'$ such that for every $\pi'\in\Pi'$ there exists a complex-valued function $\nu^{\pi'}$ and a $\cc^d$-valued function $\phi^{\pi'}$ such that
\begin{align}\label{affineprop}
\varphi^{p^2=(t,x,\pi')}(z) = E\big(\ee{i\skl z,X_t^{\pi'}\skr}\big|X_0^{\pi'}=x\big) = \ee{\nu^{\pi'}(t,iz) + \skl \phi^{\pi'}(t,iz),x\skr},
\end{align}
for every $t\ge0$,  $z\in\rrd$ and $x\in\mathfrak{D}$. 
Under mild regularity conditions, the functions $\nu^{\pi'}$ and $\phi^{\pi'}$ are determined as solutions to generalized Riccati equations. 
We refer to Duffie, Filipovi\'c and Schachermayer (2003)
\nocite{DuffieFilipovicSchachermayer2003} for a detailed exposition. 
The rich class of affine processes comprises the class of L\'evy processes, for which $\nu^{\pi'}(t,iz)=t \psi^{\pi'}(z)$ with $\psi^{\pi'}$ given as some exponent in the L\'evy-Khintchine formula \eqref{LevyKhinchine} and $\phi^{\pi'}(t,iz)\equiv0$. Moreover, many popular stochastic volatility models such as the Heston model as well as stochastic volatility models with jumps, e.g.\ the model of~\cite{Barndorff-NielsenShephard2001} 
and time-changed L\'evy models, see Carr, Geman, Madan and Yor~(2003)\nocite{CarrGemanMadanYor2003} and \cite{Kallsen2006}, are driven by affine processes.

Consider option prices of the form
\begin{align}\label{def-Price1}
\Price^{(K,T,x,\pi')}=E\big(f^{K}(X^{\pi'}_T)|X_0^{\pi'}=x\big)
\end{align}
 where $f^{K}$ is a parametrized family of measurable payoff functions $f^{K}:\R^d\rightarrow\R_+$ for $K\in \OP^1$.
\begin{corollary}\label{cor-Basetaffine}
Under the conditions (A1)--(A3) for weight $\eta\in\rrd$, $\varrho\in(1,\infty)^D$ and $\OP=\OP^1\times\OP^2\subset\rr^D$ of hyperrectangular structure assume
\begin{itemize}
\item[(i)] for every parameter $p^2=(t,x,\pi')\in \OP^2\subset \rr^{D-m}$ that the validity of the affine property \eqref{affineprop} extends to $z=\rr + i\eta$, i.e.\ for every $z\in \rr + i\eta$,
\begin{align*}
\varphi^{p^2=(t,x,\pi')}(z) = E\big(\ee{i\skl z,X_t^{\pi'}\skr}\big|X_0^{\pi'}=x\big) = \ee{\nu^{\pi'}(t,iz) + \skl \phi^{\pi'}(t,iz),x\skr},
\end{align*}

\item[(ii)] for every $z\in\rrd$ that the mappings $(t,\pi')\mapsto \nu^{\pi'}(t,iz-\eta)$ and $(t,\pi')\mapsto \phi^{\pi'}(t,iz-\eta)$ have an analytic extension to the Bernstein ellipse $B(\Pi',\varrho')$ for some parameter $\varrho'\in (1,\infty)^{D-m-1}$,

\item[(iii)] there exist $\alpha\in(1,2]$ and constants $C_1,C_2>0$ such that uniformly in the parameters $p^2=(t,x,\pi')\in B(\OP^2,\widetilde{\varrho}^2)$ for some generalized Bernstein ellipse with $\widetilde{\varrho}^2\in(1,\infty)^{D-m}$
\begin{align*}
\Re\big(\nu^{\pi'}(t,iz-\eta) + \skl \phi^{\pi'}(t,iz-\eta),x\skr\big) \le C_1 - C_2|z|^\alpha\qquad\text{for all }z\in\rr.
\end{align*}
\end{itemize}
Then  there exist constants $C>0$, $\underline{\varrho}>1$ such that
  \begin{align*}
 \max_{p\in\OP^1\times\OP^2}\big|\Price^p - I_{\overline{N}}(\Price^{(\cdot)})(p)\big|
 \le 
 C\underline{\varrho}^{-\underline{N}}.
  \end{align*}
\end{corollary}
\begin{proof}
Thanks to Theorem \ref{ConvergenceFourierPrices} and Corollary \ref{cor-Asymptotic_error_Expo} and in view of the assumed validity of Conditions (A1)--(A3), it suffices to verify Condition (A4). While assumptions (i) and (ii) together yield the analyticity condition in (A4), part\tild(iii) provides the upper bound in (A4). 
\end{proof}

The existence of exponential moments of affine processes has been investigated by \cite{Keller-ResselMayerhofer2015}, where also criteria are provided under which formula \eqref{affineprop} and the related generalized Riccati system can be extended to complex exponential moments $z\in\cc^d$. The question has already been treated for important special cases, which allow for more explicit conditions. \cite{FilipovicMayerhofer2009} consider affine diffusions and \cite{CheriditoWugalter2012} investigate affine processes with killing when the jump measures possess exponential moments of all orders.

\section{Analyticity Properties}\label{sec-analyticity}

\subsection{\textbf{Analyticity Properties of Selected Payoff Profiles}}\label{sec-Payoffs}

We first list in Table \ref{table-exEuropayoff} a selection of payoff profiles $f^K$ for option parameter $K$ as function of the logarithm of the underlying.  As we have seen in Proposition~\ref{FourierPrices}, we can represent option prices under certain conditions by their Fourier transform. Therefore, the table provides the Fourier transform $\widehat{f^K}$ of the respective option payoff, as well. 
\begin{table}[h!]
\begin{center}
\begin{tabular}{lcccc}
\toprule\vspace{1ex}
\textbf{Type} & \textbf{Payoff} & \textbf{Weight} & \textbf{Fourier transform}      & $\widehat{f^K}$ \textbf{holomor-}  \\
     &   $f(x)$       & $\eta$ & $\widehat{f^K}(z - i\eta)$ & \textbf{phic in} $\log(K)$ \\[1ex]
     \hline\\
Call         & $(\ee x - K)^+$ & $<-1$ & $\frac{K^{iz + 1 + \eta}}{(iz+\eta)(iz + 1 + \eta)}$ & \checkmark\\\\
Put  		&  $(K - \ee x)^+$ & $>0$ & $\frac{K^{iz + 1 + \eta}}{(iz+\eta)(iz + 1 + \eta)}$ & \checkmark
\\\\
Digital & $\1_{x>\log(K)}$ & $<0$ & $-\frac{K^{iz +  \eta}}{iz + \eta}$ &\checkmark\\
\scalebox{.95}{down{\&}out}\\[1ex]
Asset-or-nothing & $\ee{x}\1_{x>\log(K)}$ & $<-1$ & $-\frac{K^{iz + 1 + \eta}}{iz + 1 + \eta}$ & \checkmark\\
\scalebox{.95}{down{\&}out}\\
\bottomrule
\end{tabular}
\caption{Examples of payout profiles of a single underlying.}\label{table-exEuropayoff}
\end{center}
\end{table}
All examples in Table  \ref{table-exEuropayoff} are special cases by $k:=\log(K)$ of the following lemma that has a straightforward proof.
\begin{proposition}\label{prop-fKholo}
Let $\rr\times\rrd\ni(k,x)\mapsto f^k(x)$ be of form $f^k(x)=h(k) f^0\big(\tau^k(x)\big)$ with a transform $\tau^k(x)_j= x_j +\alpha_j k $ with $\alpha_j\in\rr$ for every $j=1,\ldots,d$ and a holomorphic function $h$. Let $\eta\in\rrd$ such that $x\mapsto \ee{\skl \eta,x\skr}f^k(x)$ belongs to $L^1(\rrd)$ for every $k\in\rr$. Then 
\begin{itemize}
\item[(i)] the mapping $k\mapsto \widehat{f^k}(z-i\eta)$ has a holomorphic extension and

\item[(ii)]
$|\widehat{f^k}(z-i\eta)|\le |h(k)|\ee{|\eta||k|} \ee{|\Im(k)||z|}$ for every $k\in\cc$.
\end{itemize}

\end{proposition}
\begin{proof}
Both assertions are immediate consequences of 
\[
\widehat{f(\tau^k)}(z-i\eta) = \exp\bigg\{-k\sum_{j=1}^d\alpha_j(iz_j - \eta) \bigg\} \hat{f}(z-i\eta)
\] 
for all $z\in\rr^d$ and all $k\in\cc$.
\end{proof}


\begin{example}[Call on the minimum of several assets]
The payoff profile of a call option on the minimum of $d$ assets is defined as
	\begin{equation}
		f^K(x) = \left(e^{x_1} \wedge e^{x_2} \wedge \dots \wedge e^{x_d} - K  \right)^+,
	\end{equation}
for $x=(x_1,\dots,x_d)\in\rr^d$ and strike $K\in\rr^+$. With dampening constant $\eta\in(-\infty,-1)^d$ we have $x\mapsto \ee{\skl\eta,x\skr}f^K(x)\in L^1(\rrd)$. 
%
%
Proposition~\ref{prop-fKholo} shows that $K\mapsto \widehat{f^K}(z-i\eta)$ has a holomorphic extension.
\end{example}

For some payoff profiles it is not possible to transform them to integrable functions by a multiplication with an exponential dampening factor. Thus we split them into summands which can be appropriately dampened.
Therefore, we decompose the unity, $1 =\sum_{j=1}^{2^d}\1_{O^j}(x)$ a.e.\ with the distinct orthants $O^j$ of $\rrd$.  More precisely, for $j=1,\ldots,2^d$ let $\zeta^j:=(\zeta^j_1,\ldots,\zeta^j_d)$ with $\zeta^j_i\in\{-1,1\}$ for the $2^d$ different possible configurations and let 
\begin{equation}\label{def-Oj}
O^j\coloneqq \big\{(x_1,\ldots,x_d)\in\rrd \,\big|\, \zeta^j_ix_i\ge 0\,\quad\text{for all }i=1,\ldots,d\big\}.
\end{equation}
For each $j=1,\ldots,2^d$ we choose for the call, respectively put, profile
\begin{equation}\label{def-etaj}
\eta^j:= -(1+\varepsilon)\zeta^j,\,\quad \text{respectively }\eta^j:= \varepsilon \zeta^j
\end{equation}
for some $\varepsilon>0$ for each $j=1,\ldots,2^d$. 

The following proposition enables us to prove analytic dependence of prices of basket put and call options on the strike.
\begin{proposition}\label{lem-basket}
For every $j=1,\ldots,2^d$ let $O^j$ and $\eta^j$ as defined in \eqref{def-Oj} and \eqref{def-etaj}. Then for each $k\in\rr$ the payoff profile $f^k$ of a basket on a call, respectively put, is of the form
$f^k(x) = \big( \ee{x_1} + \ldots+\ee{x_d} - \ee k\big)^+$, respectively $f^k(x) = \big( \ee k - \ee{x_1} - \ldots-\ee{x_d} \big)^+$, and
\begin{align}
f^k(x) = 
 \sum_{j=1}^{2^d}f^k_{j}(x)
\end{align}
 with $f^k_j(x) := f^k(x)\1_{O^j}(x)$ for a.e.\ $x=(x_1,\ldots,x_d)\in\rrd$.
Moreover, for every $j=1,\ldots,2^d$ and every $k\in\rr$ we have that $x\mapsto \ee{\skl\eta^j,x\skr} f_j^k(x)\in L^1(\rrd)$  and for every $z\in\rrd$ the mapping
\begin{align}
k\mapsto \widehat{f^k_j}(z-i\eta^j)
\end{align}
has a holomorphic extension.
\end{proposition}
\begin{proof}
We show the claim for the call option. The put option case is proved analogously. For every $j=1,\ldots,2^d$, under the assumptions of the proposition it is elementary to show that $x\mapsto \ee{\skl\eta^j,x\skr} f_j^k(x)\in L^1(\rrd)$. Moreover, letting $u:=z -i\eta^j$ with $z\in\rrd$ we have
\begin{align}
\widehat{f^k_j}(u) &= \int_{O^j} \ee{i \skl u,x\skr} f^k(x) \dd x\nonumber\\
&= \ee{k} \ee{i\skl u,\vec{k}\skr }  \int_{A(k)} \ee{i\skl u, x\skr } \big( \ee{x_1} + \ldots+\ee{x_d} - 1\big) \dd x_1\ldots\dd x_d\label{intAk}
\end{align}
with $\vec{k}=(k,\ldots,k)\in\rrd$ and $A(k):= (O^j - \vec{k})\cap \{x\in\rrd:\,f^0(x)>0\}$. From the geometry of $A(k)$ and the holomorphicity of the integrand in \eqref{intAk} we can deduce that the integral in \eqref{intAk} is holomorphic in $k$ and thus obtain the assertion of the proposition. 
\end{proof}


\subsection{\textbf{Analyticity Properties of Asset Models}}\label{sec-models}
\label{sec:models}

In this section, we shortly introduce a range of asset models and present analyticity properties of their Fourier transforms. For some models and some parameters, the domain of analyticity is immediately observable. For some non-trivial cases, we state the domain briefly. Throughout the section, $T>0$ denotes the time to maturity of the option while $r>0$ refers to the constant risk-free interest rate.

\subsubsection{Multivariate \BS Model}
\label{sec:Levy1}
In the multivariate \BS model, the characteristic function  of $L_T^\pi$, the multivariate analogon of \eqref{eq:LevyAssetProcess}, with $\pi=(b,\sigma)$ is given by
	\begin{equation}
	\label{eq:dCFbs}
		\varphi^{p^2=(T, b,\sigma)}(z ) = \exp\left(T\left( i \langle b ,z\rangle - \frac{1}{2}\langle z , \sigma z\rangle\right)\right),
	\end{equation}
with covariance matrix $\sigma\in\rr^{d\times d}$ such that $\det(\sigma)>0$ and drift $ b \in\rr^d$ adhering to the drift condition
	\begin{equation}
	\label{eq:dCFbsdrift}
		 b _i = r-\frac{1}{2}\sigma_{ii},\quad i=1,\dots,d.
	\end{equation}

In the \BS model, analyticity in the parameters is immediately confirmed, i.e. $p^2\mapsto\varphi^{p^2}(z)$, given in \eqref{eq:dCFbs}, is holomorphic for every $z\in\rr^d$. The admissible parameter domain, however, is restricted to parameter constellations such that $\sigma$ is a covariance matrix. 

\begin{remark}
\label{rem:bsanalytic}
Let $\eta\in\R^d$ be the chosen weight in Conditions \ref{conditionsA1234} and let the open set $U$ be given by
	\begin{equation}
		U \subseteq (0,\infty) \times \R \times \left\{\vec{\sigma}\in\R^{d(d+1)/2}\ \big|\ \text{$\sigma(\vec{\sigma})$ positive definite} \right\},
	\end{equation}
where
$\sigma:\R^{d(d+1)/2}\rightarrow\R^{d\times d}$ is defined by $\sigma(\vec{\sigma})_{ij} = \sigma_{(\max\{i,j\}-1)\max\{i,j\}/2+\min\{i,j\}}$, $i,j\in\{1,\dots,d\}$, for $\vec{\sigma} \in \R^{d(d+1)/2}$. By construction, $\sigma(\vec{\sigma})$ is symmetric for any $\vec{\sigma} \in \R^{d(d+1)/2}$.

Then for every $z\in\rr^d$, $(T, b,\vec{\sigma}) \mapsto \varphi^{p^2=(T, b,\sigma(\vec{\sigma}))}(z+i\eta)$ given by \eqref{eq:dCFbs} is analytic on $U$. Note that $U$ does not depend on $\eta$.
\end{remark}

\subsubsection{Univariate Merton Jump Diffusion Model}
As second L\'evy model we state the Merton Jump Diffusion model by \cite{merton1976}. The logarithm of the asset price process follows a jump diffusion with volatility $\sigma^2$ enriched by jumps arriving at a rate $\lambda>0$ with normally $\mathcal{N}(\alpha,\beta^2)$ distributed jump sizes. Here, the characteristic function of $L_T^\pi$ in \eqref{eq:LevyAssetProcess} with $\pi=(b,\sigma,\alpha,\beta,\lambda)$ is given by
		\begin{equation}
		\varphi^{p^2=(T,b,\sigma,\alpha,\beta,\lambda)}(z) = \exp\left(T\left(ibz - \frac{\sigma^2}{2}z^2+\lambda\left(e^{iz\alpha-\frac{\beta^2}{2}z^2}-1\right)\right)\right)
	\end{equation}
with drift condition \eqref{eq:LevyDriftCond} turning into
	\begin{equation}
		b = r-\frac{\sigma^2}{2}-\lambda\left(e^{\alpha+\frac{\beta^2}{2}}-1\right),
	\end{equation}
where $\sigma>0$, $\alpha\in\R$ and $\beta\geq0$. Since the characteristic function in the Merton Jump Diffusion model is composed of analytic functions, it is itself analytic on the whole parameter domain.

\subsubsection{Univariate CGMY Model}
\label{sec:Levy4}
We consider the CGMY model by Carr, Geman, Madan and Yor (2002)\nocite{CarrGemanMadanYor2002}, which is a special case of the extended Koponen's family of \cite{BoyarchenkoLevendoskii2000}, with characteristic function \eqref{LevyKhinchine} of $L_T^\pi$ in \eqref{eq:LevyAssetProcess} with $\pi=(b,C,G,M,Y)$ 
	\begin{equation}
	\label{eq:CFcgmy}
	\begin{split}
		&\varphi^{p^2=(T,b,C,G,M,Y)}( z )\\
		&\ = \exp\left(T \left(i b  z + C\Gamma(-Y) \left[(M-i z )^Y - M^Y + (G+i z )^Y - G^Y\right]\right)\right),
	\end{split}
	\end{equation}
where $C>0$, $G> 0$, $M> 0$, $0<Y<2$ and $\Gamma(\cdot)$ denotes the Gamma function. The drift $ b \in\rr$ is set to
	\begin{equation}
	\label{eq:CFcgmydrift}
		 b  = r-C\Gamma(-Y) \left[(M-1)^Y - M^Y + (G+1)^Y - G^Y\right]
	\end{equation}
to comply with the drift condition \eqref{eq:LevyDriftCond} for martingale pricing.

\begin{remark}
\label{rem:cgmyanalytic}
For weight $\eta\in\R$ from Conditions \ref{conditionsA1234}, choose an open set $U(\eta)$ with
	\begin{equation}
		\begin{split}
			U(\eta) \subseteq&\ (0,\infty)\times \R \times (0,\infty)  \\
			&\quad \times\big\{(G,M) \in (0,\infty)^2\ \big|\ G-\eta>0,\ M+\eta>0\big\} \times (0,2).
		\end{split}
	\end{equation}
Then for every $z\in\rr$, $(T,b, C,G,M,Y) \mapsto \varphi^{p^2=(T,b,C,G,M,Y)}(z+i\eta)$ for the characteristic function $\varphi^{p^2}$ of the CGMY model \eqref{eq:CFcgmy} is analytic on $U(\eta)$.
\end{remark}
Table \ref{table-processes} displays for selected L\'evy models conditions on the weight $\eta\in\R$ and the index $\alpha\in(1,2]$ that guarantee (A3) and the upper bound in (A4) for a fixed model parameter constellation. 

\begin{table}[h!]
\begin{center}
\begin{tabular}{lcl}
\toprule
\multirow{2}{4cm}{\textbf{Class}} & \multicolumn{2}{l}{\textbf{(A3) and the upper bound in (A4) hold for} }\\
&$\eta$ such that &and $\alpha$ such that\\
\midrule  \\[-2ex]
Brownian Motion & &$\alpha=2$\\ 
with drift&\\
[2ex]
Merton Jump Diffusion & &$\alpha=2$\\ 
[2ex]
L\'evy jump diffusion with
&$\int_{|x|>1}\ee{|\eta||x|}F(\dd x)<\infty$&$\alpha=2$\\ 
characteristics $(b,\sigma,F)$&\\
[2ex]
univariate CGMY with & $\eta\in(-\min\{G,M\},\,\max\{G,M\} )$ & $\alpha=Y$\\
parameters $(C,G,M,Y)$ && \\
with $Y>1$ &&\\
\bottomrule
\end{tabular}
\caption{Conditions on $\eta$ and $\alpha$ for (A3) and the upper bound in (A4) to hold for a fixed model parameter constellation. The selected L\'evy models are described in more detail in Sections \ref{sec:Levy1}--\ref{sec:Levy4}.}
\label{table-processes}
\end{center}
\end{table}
%
For the CGMY model with parameters $C,G,M>0$ and $Y<2$ it is obvious that~\eqref{expoLevy} is statisfied for every $\eta\in\rr$ with $G>\eta$ and $M>-\eta$. Moreover, equation~(4.12) in~\cite{Glau2015a} shows that for $Y\le 1$, \eqref{Im-jump-Levy} is satisfied with $\beta=1$.

\subsubsection{Heston Model for Two Assets}
\label{sec:Heston_Model_for_Two_Assets}
Here we state the two asset version of the multivariate Heston model in the special case of having a single, univariate driving volatility process $\{v_t\}_{t\geq0}$. The two asset price processes are modeled as
	\begin{equation}
		S^1_t = S_0^1e^{H^1_t}\quad\text{ and }\quad S^2_t = S_0^2e^{H^2_t},\quad\text{for }t\geq 0,
	\end{equation}
where $H=(H^1,H^2)$ solves the following system of SDEs,
	\begin{align*}
		\dd H_t^1 =&\ \left(r-\frac{1}{2}\sigma_1^2\right) \dd t + \sigma_1 \sqrt{v_t}\dd W_t^1,\\
		\dd H_t^2 =&\ \left(r-\frac{1}{2}\sigma_2^2\right) \dd t + \sigma_2 \sqrt{v_t}\dd W_t^2,\\
		\dd v_t =&\ \kappa(\theta-v_t)\dd t + \sigma_3\sqrt{v_t}\dd W_t^3,
	\end{align*}
where the Brownian motions $W_i$, $i=1,2,3$, are correlated according to $\langle W^1,W^2\rangle = \rho_{12}$, $\langle W^1,W^3\rangle = \rho_{13}$, $\langle W^2,W^3\rangle = \rho_{23}$. Following Eberlein, Glau, Papapantoleon~(2010)\nocite{EberleinGlauPapapantoleon2010a}, the characteristic function of $H_T$ in this framework is
	\begin{equation}
	\label{eq:Heston2dChar}
	\begin{split}
		 &\varphi^{p^2=(T,v_0,\kappa,\theta,\sigma_1,\sigma_2,\sigma_3,\rho_{12},\rho_{13},\rho_{23})}(z)\\
		 &\ \qquad = \exp\left(Ti\left\langle \begin{pmatrix}r\\r\end{pmatrix},z\right\rangle\right)\exp\bigg(\frac{v_0}{\sigma_3^2}\frac{(a-c)(1-\exp(-cT))}{1-g\exp(-cT)} \\
		&\ \qquad\qquad\qquad +\frac{\kappa\theta}{\sigma_3^2}\left[(a-c)T - 2\log\left(\frac{1-g\exp(-cT)}{1-g}\right)\right]\bigg),
	\end{split}
	\end{equation}
with auxiliary functions
	\begin{equation}
	\begin{split}
		\zeta = \zeta(z) =&\ -\left(\left\langle z, \begin{pmatrix}
												\sigma_1^2 & \rho_{12}\sigma_{1}\sigma_{2}\\
												\rho_{12}\sigma_{1}\sigma_{2} & \sigma_2^2
											\end{pmatrix}z\right\rangle + \left\langle\begin{pmatrix} \sigma_1\\ \sigma_2\end{pmatrix}, i z \right\rangle\right)\\
		&\ -\left(\sigma_1^2 z _1^2 + \sigma_2^2 z_2^2 +2\rho_{12}\sigma_1\sigma_2 z_1 z_2 + i\sigma_1^2 z_1 + i\sigma_2^2 z_2\right),\\
		a = a(z) =&\ \kappa-i\rho_{13}\sigma_1\sigma_3 z_1-i\rho_{23}\sigma_2\sigma_3 z_2,\\
		c = c(z) =&\ \sqrt{a(z)^2-\sigma_3^2\zeta(z)},\\
		g = g(z) =&\ \frac{a(z)-c(z)}{a(z)+c(z)},
	\end{split}
	\end{equation}
and positive parameters $v_0,\ \kappa,\ \theta$ and $\sigma_3$ fulfilling the Feller condition
	\begin{equation}
		\sigma_3^2 \leq 2\kappa\theta,
	\end{equation}
ensuring an almost surely non-negative volatility process $(v_t)_{t\geq 0}$. Obviously, for each $z\in\rr^d$, the characteristic function $\varphi^{p^2=(T,v_0,\kappa,\theta,\sigma_1,\sigma_2,\sigma_3,\rho_{12},\rho_{13},\rho_{23})}(z)$ of \eqref{eq:Heston2dChar} is analytic in $v_0$ and $\theta$. Let us additionally mention that the analiticity in $z$ has already been investigated by \cite[Section 2.3, Lemma 2.1]{Levendorskiy2012} for the univariate Heston model.

\section{Numerical Experiments}\label{sec-numerics}

We apply the Chebyshev interpolation method to parametric option pricing considering a variety of option types in different well known option pricing models. Moreover, we conduct both an error analysis as well as a convergence study. The first focuses on the accuracy that can be achieved with a reasonable number of Chebyshev interpolation points. The latter confirms the theoretical order of convergence derived in Section \ref{sec-pop}, when the number of Chebyshev points increases. Finally, we study the gain in efficiency for selected multivariate options.

We measure the numerical accuracy of the Chebyshev method by comparing derived prices with prices coming from a reference method. We employ the reference method not only for computing reference prices but also for computing prices at Chebyshev nodes $\Price^{p^{(k_1,\dots,k_D)}}$ with $(k_1,\dots,k_D)\in J$ during the precomputation phase of the Chebyshev coefficients $c_j$, $j\in J$, in \eqref{def:Chebycj}. Thereby, a comparability between Chebyshev prices and reference prices is maintained. 

We implemented the Chebyshev method for applications with two parameters. To that extent we pick two free parameters $p_{i_1},\ p_{i_2}$ out of \eqref{def-para}, $1\leq i_1<i_2\leq D$, in each model setup and fix all other parameters at reasonable constant values. We then evaluate option prices for different products on a discrete parameter grid $\overline{\mathcal{P}} \subseteq[\underline{p}_{i_1},\overline{p}_{i_1}]\times [\underline{p}_{i_2},\overline{p}_{i_2}]$ defined by
	\begin{equation}
	\label{eq:defevalgrid}
		\begin{split}
			\overline{\mathcal{P}} =&\ \left\{\left(p_{i_1}^{k_{i_1}},p_{i_2}^{k_{i_2}}\right),\ {k_{i_1}},{k_{i_2}}\in\{0,\dots,100\}\right\},\\
			p_{i_j}^{k_{i_j}} =&\ \underline{p}_{i_j} + \frac{k_{i_j}}{100}\left( \overline{p}_{i_j}- \underline{p}_{i_j}\right),\ k_{i_j}\in\{0,\dots,100\},\ j\in\{1,2\}.
		\end{split}
	\end{equation}	
Once the prices have been derived on $\overline{\mathcal{P}}$, we compute the discrete $L^\infty(\overline{\mathcal{P}})$ and $L^2(\overline{\mathcal{P}})$ error measures,
	\begin{equation}
	\label{eq:deferrormeas}
	\begin{split}
		\varepsilon_{L^\infty}(\overline{N}) =&\ \max\limits_{p\in\overline{\mathcal{P}}}\left|\Price^p - I_{\overline{N}}(\Price^{(\cdot)})(p)\right|,\\
		\varepsilon_{L^2}(\overline{N}) =&\ \sqrt{\Delta_{\overline{\mathcal{P}}}\sum_{p\in\overline{\mathcal{P}}}\left|\Price^p - I_{\overline{N}}(\Price^{(\cdot)})(p)\right|^2},
	\end{split}
	\end{equation}
where $\Delta_{\overline{\mathcal{P}}} = \frac{(\overline{p}_{i_1}- \underline{p}_{i_1})}{100}\frac{(\overline{p}_{i_2}- \underline{p}_{i_2})}{100}$, to interpret the accuracy of our implementation and of the Chebyshev method as such.

\subsection{European Options}
We consider a plain vanilla European call option on one asset as well as a European digital down{\&}out option, first. Both derivatives have been introduced in Table~\ref{table-exEuropayoff}. For these products we investigate the performance of the Chebyshev interpolation method for the Heston model and the L\'evy models of Black{\&}Scholes, Merton and the CGMY model. 
\begin{table}[htb!]
\begin{center}
\begin{tabular}{@{}lllllll@{}}
\toprule
\textbf{Model}  & \phantom{a} & \multicolumn{2}{c}{\textbf{fixed parameters}}                  & \phantom{a} & \multicolumn{2}{c}{\textbf{free parameters}}                \\ 
       &                 & \multicolumn{1}{c}{$p^1$} & \multicolumn{1}{c}{$p^2$} &          & \multicolumn{1}{c}{$p^1$} & \multicolumn{1}{c}{$p^2$} \\ \midrule
BS     &                 & $K=1$                     & $\sigma=0.2$             &          & $S_0/K\in[0.8,\ 1.2]$    & $T\in[0.5,2]$  \\ \midrule
Merton    &                 & $K=1$                     & $\sigma=0.15$,             &          & $S_0/K\in[0.8,\ 1.2]$    &  $T\in[0.5,2]$                          \\
       &                 &                           & $\alpha = -0.04$,             &          &  &                           \\
       &                 &                           & $\beta=0.02$,              &          &                           &         \\
       & & &  $\lambda=3$ \\ 
\midrule
CGMY   &                 & $K=1$                     & $C=0.6$,                  &          & $S_0/K\in[0.8,\ 1.2]$    &  $T\in[0.5,2]$    \\
       &                 &                           & $G=10$,                   &          &  &                           \\
       &                 &                           & $M=28$,                   &          &                           &                           \\
       &                 &                           & $Y=1.1$                   &          &                           &                           \\ \midrule
Heston &                 & $K=1$                    & $T=2$,             &          & $S_0/K\in[0.8,\ 1.2]$     & $v_0\in[0.1^2,\ 0.4^2]$   \\
       &                 &                      & $\kappa=1.5$,           &          &                           &                           \\
       &                 &                           & $\theta=0.2^2$,            &          &                           &                           \\
       &                 &                           & $\sigma=0.25$,                &          &                           &                           \\
       & & & $\rho=0.1$ \\\bottomrule
\end{tabular}
\caption{Parametrization of models and European call option.}
\label{tab:ChebyD2EuroCallparam}
\end{center}
\end{table}
 We keep the strike parameter constant, $p^1=k=\log(K)$ for $K=1$. As previously discussed in Remark~\ref{cor-callLevy}, this is no restriction of generality. We also disregard interest rates, setting $r=0$. For the three L\'evy models we vary the maturity $T$ (in years) as well as the option moneyness $S_0/K$. All three models fall within the scope of Corollary~\ref{cor-callLevy}, where the error is analyzed explicitly. Thus we expect (sub)exponential convergence for the three L\'evy models. For the Heston model we vary $S_0/K$ and let $v_0$ as one of the model parameters float. Due to the analyticity of the Fourier transform of the payoff in $S_0/K$ and of the characteristic function of the process in $v_0$, compare Section \ref{sec:Heston_Model_for_Two_Assets}, we expect this convergence for the Heston model as well.
  A detailed overview of the realistically chosen parametrization is given by Table \ref{tab:ChebyD2EuroCallparam}. 
 For numerical integration in Fourier pricing we use Matlab's \texttt{quadgk} routine over the interval $[0,\infty)$ with absolute precision bound of $\varepsilon<10^{-14}$.

The first question we address concerns the achievable accuracy with a fixed number of Chebyshev polynomials. We set $N_1=N_2=10$ and precompute the Chebyshev coefficients as defined in \eqref{def:Chebycj} with $D=2$ using the parametrization of Table \ref{tab:ChebyD2EuroCallparam} for the models therein. We evaluate the resulting polynomial over a parameter grid of dimension $D=2$ and compute the approximate European option prices in each node. As a comparison, we also compute the respective Fourier price via numerical integration of the accordingly parametrized integrand in \eqref{FT-Price}. Figure~\ref{fig:Cheby2Dhowgoodwith10_call} shows the results for the European call option. The accuracy achieved by $N=N_1=N_2=10$ shows a significant spread over the four different models but reaches very satisfying error levels from $10^{-7}$ to $10^{-10}$. Increasing the number of Chebyshev points further improves the accuracy. Since at its core the implementation of the Chebyshev method consists of summing up matrices, this refinement comes at virtually no additional cost.

\begin{figure}[htcp!]
	\centering
	\begin{subfigure}[t]{.48\linewidth}
		\centering
		\includegraphics[scale=1]{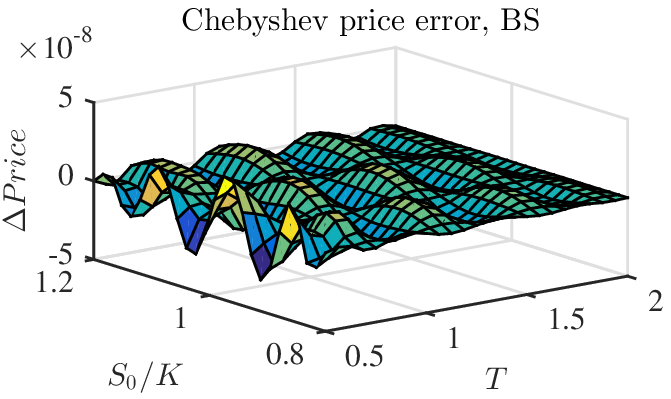}
	\end{subfigure}
	\begin{subfigure}[t]{.48\linewidth}
		\centering
		\includegraphics[scale=1]{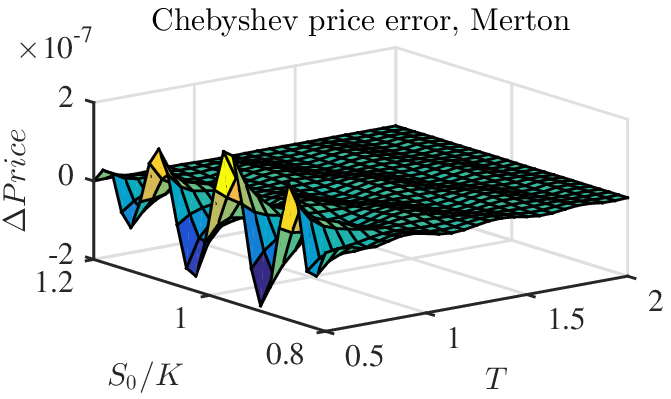}
	\end{subfigure}
	\begin{subfigure}[t]{.48\linewidth}
		\centering
		\includegraphics[scale=1]{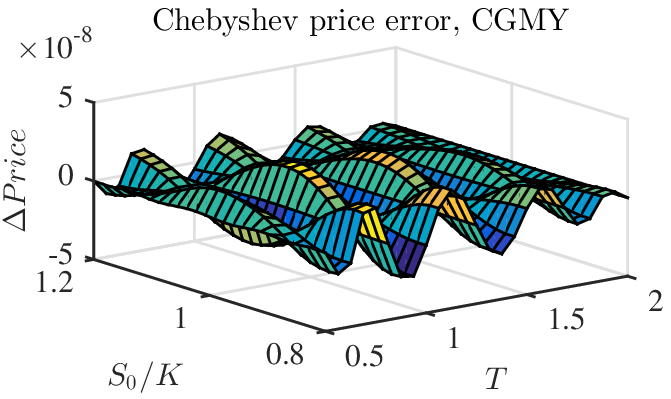}
	\end{subfigure}
	\begin{subfigure}[t]{.48\linewidth}
		\centering
		\includegraphics[scale=1]{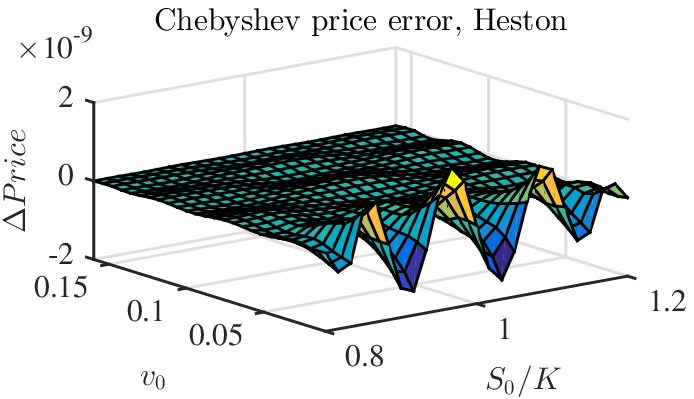}
	\end{subfigure}
	\caption{Absolute pricing error for a European call option with strike $K=1$ in various models. We compare the Chebyshev interpolation with $N=N_1=N_2=10$ to classic Fourier pricing by numerical integration. The parametrization of the models and the option has been chosen according to Table \ref{tab:ChebyD2EuroCallparam}. We observe that the achieved accuracy varies significantly between the models. The order of accuracy of \BS prices is identical to the accuracy for the CGMY case. The Merton model falls behind by one order of magnitude, while the Heston model shows a very strong fit between prices and their approximation by the Chebyshev method.}\label{fig:Cheby2Dhowgoodwith10_call}
\end{figure}
In the same parametrization setting we price a European digital down{\&}out option. While a call payoff profile is not differentiable but at least continuous, the digital payoff function is not even continuous, compare Table~\ref{table-exEuropayoff}. This reduction in smoothness of the payoff function reduces the accuracy of the interpolation $p\mapsto\Price^{(p)}$ as well. We compare the Chebyshev interpolation again with $N_1=N_2=10$ to classic Fourier pricing by numerical integration. The parametrization of the models and the option has been chosen again according to Table~\ref{tab:ChebyD2EuroCallparam}, where $K=1$ now denotes the parameter of the digital option. Figure~\ref{fig:Cheby2Dhowgoodwith10_digi} shows the results. Comparing the results of the call option pricing in Figure~\ref{fig:Cheby2Dhowgoodwith10_call} to the results of the digital option pricing in Figure~\ref{fig:Cheby2Dhowgoodwith10_digi} we observe that the accuracy achieved by $N_1=N_2=10$ is reduced by a factor $10^1$ to $10^2$. This demonstrates how little the reduced smoothness of the payoff profile effects the accuracy of the Chebyshev method in these cases.
\begin{figure}[htcp!]
	\centering
	\begin{subfigure}[t]{.48\linewidth}
		\centering
		\includegraphics[scale=1]{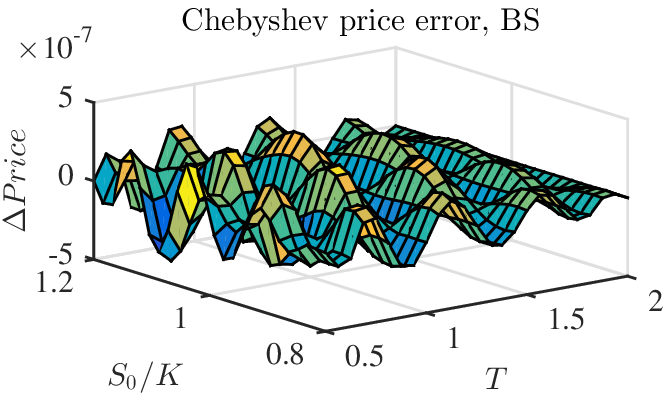}
	\end{subfigure}
	\begin{subfigure}[t]{.48\linewidth}
		\centering
		\includegraphics[scale=1]{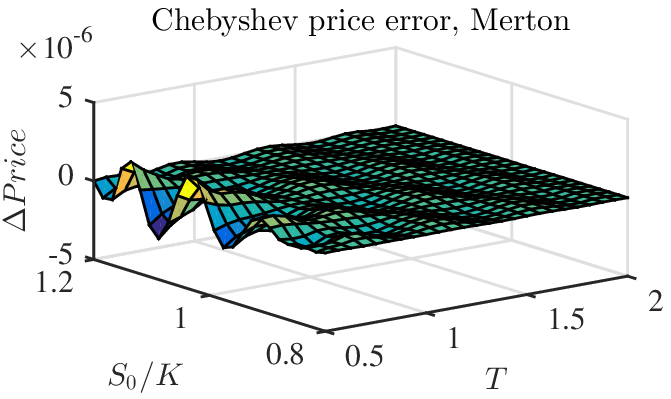}
	\end{subfigure}
	\begin{subfigure}[t]{.48\linewidth}
		\centering
		\includegraphics[scale=1]{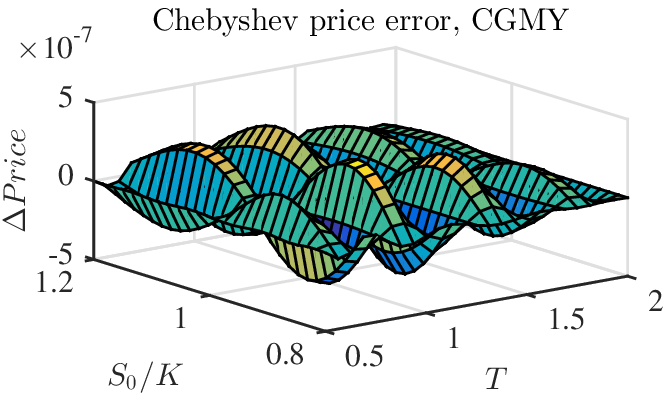}
	\end{subfigure}
	\begin{subfigure}[t]{.48\linewidth}
		\centering
		\includegraphics[scale=1]{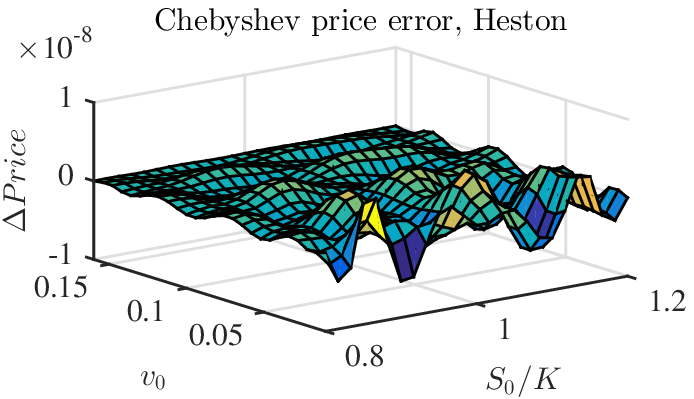}
	\end{subfigure}
	\caption{Absolute pricing error for a European digital down{\&}out option with $K=1$ in various models. The plots of this figure correspond to the plots of Figure~\ref{fig:Cheby2Dhowgoodwith10_call} where a Europan call option was priced instead. We compare again the performance of the Chebyshev method for $N=N_1=N_2=10$. All four price surfaces lose between one to two orders of magnitude in accuracy compared to their call option counterparts.}\label{fig:Cheby2Dhowgoodwith10_digi}
\end{figure}
Furthermore we conduct an empirical convergence study for this very same setting of option and model parametrization. For an increasing degree $N=N_1=N_2$, the Chebyshev polynomial is set up and prices over a parameter grid of structure \eqref{eq:defevalgrid} are computed. Again, Fourier pricing serves as a comparison. For each $N\in\{1,\dots,30\}$, the error measures $\varepsilon_{L^\infty}$ and $\varepsilon_{L^2}$, defined by \eqref{eq:deferrormeas} on the discrete parameter grid $\overline{\mathcal{P}}$ defined in \eqref{eq:defevalgrid}, are evaluated. We observe exponential convergence for all four models. Figure~\ref{fig:ChebyErrorDecayAllLog10scale_call} shows the decay for the European call option while Figure~\ref{fig:ChebyErrorDecayAllLog10scale_digi} does so for the European digital down{\&}out option.\\

\begin{center}
\begin{figure}[htcp!]
	\includegraphics[scale=0.92]{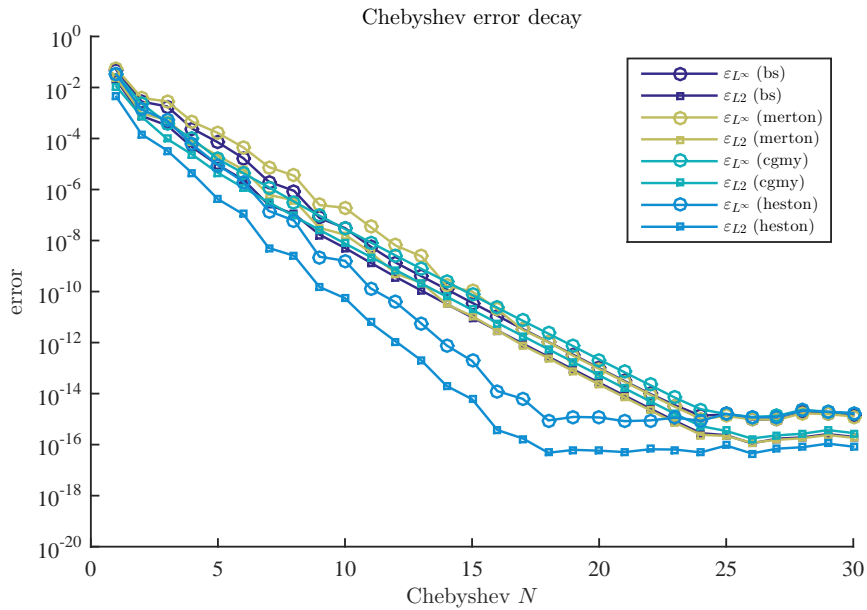}
 	\caption{Convergence study for the Black{\&}Scholes, Merton, CGMY and the Heston model for prices of a European call option parametrized as stated in Table \ref{tab:ChebyD2EuroCallparam}. The reference price is derived by Fourier pricing and numerical integration with an absolute accuracy of $10^{-14}$, which is reached by all models for $N=N_1=N_2\approx 25$ the latest. The error decay for the three L\'evy models of Black{\&}Scholes, Merton and the CGMY model roughly coincide, extending the findings from Figure \ref{fig:Cheby2Dhowgoodwith10_call}. }\label{fig:ChebyErrorDecayAllLog10scale_call}
\end{figure}
\end{center}
\begin{center}
\begin{figure}[htcp!]
	\includegraphics[scale=0.92]{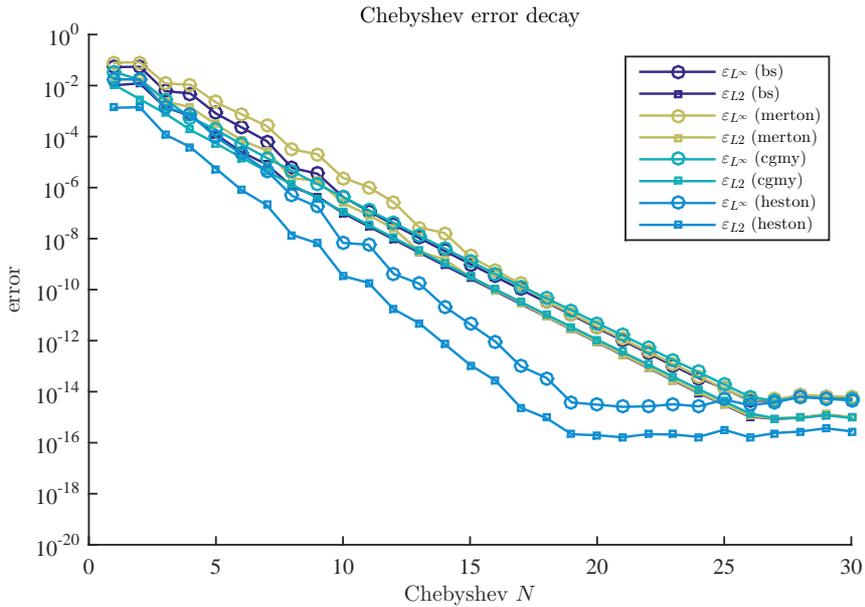}
 	\caption{Convergence study for prices of a European digital down{\&}out option in the Black{\&}Scholes, Merton, CGMY and the Heston model parametrized as stated in Table \ref{tab:ChebyD2EuroCallparam}. The results of this figure correspond to the results shown in Figure~\ref{fig:ChebyErrorDecayAllLog10scale_call} where the error decay of prices of a European call option in the same model setup was analyzed. Now higher $N$ and thus more Chebyshev nodes are needed to reach the same levels of accuracy as before.} \label{fig:ChebyErrorDecayAllLog10scale_digi}
\end{figure}
\end{center}
Defining $\zeta = \frac{\overline{T} + \underline{T}}{\overline{T} - \underline{T}}$ and setting $\varrho=\zeta +\sqrt{\zeta^2-1}$, the theoretical convergence analysis predicts a slope of the error decays in both Figure~\ref{fig:ChebyErrorDecayAllLog10scale_call} as well as Figure~\ref{fig:ChebyErrorDecayAllLog10scale_digi} of at least
\begin{equation*}
	\mathcal{S} = \log_{10}\left(\varrho\right) \approx -0.48
\end{equation*}
or steeper. Note that $\zeta_2 = \frac{\overline{S_0}\underline{K}+\underline{S_0}\overline{K}}{\overline{S_0}\underline{K}-\underline{S_0}\overline{K}}$ leads to a higher $\varrho$ and for this analysis we compare the slope to the minimal $\underline{\varrho}$ value as in Corollary \ref{cor-Asymptotic_error_Expo}. Empirically, we observe a slope for the \BS model of about $\mathcal{S}_\text{BS} = -0.64$, for the Merton model of $\mathcal{S}_\text{Merton} = -0.61$ and for the CGMY model of $\mathcal{S}_\text{CGMY} = -0.61$. Thus, the error in each L\'evy model empirically confirms the theoretical claim of Remark~\ref{cor-callLevy}. 

\FloatBarrier

\subsection{Basket and Path-dependent Options}\label{accuracy_exotics}
In this section we use the Chebyshev method to price basket and path-dependent options. First, we apply the method to interpolate Monte-Carlo estimates of prices of financial products and check the resulting accuracy. To this aim we exemplarily choose basket, barrier and lookback options in 5-dimensional \BS, Heston and Merton models. Second, we combine the Chebyshev method with a Crank-Nicolson finite difference solver with Brennan Schwartz approximation, see \cite{BrenanSchwartz1977}, for pricing a univariate American put option in the \BS model.

In our Monte-Carlo simulation we use $10^6$ sample paths, antithetic variates as variance reduction technique and 400 time steps per year.
The error of the Monte-Carlo method cannot be computed directly. We thus turn to statistical error analysis and use the well-known 95\% confidence bounds to determine the accuracy. Following the assumption of a normally distributed Monte-Carlo estimator with mean equal to the estimator's value and variance equal to the empirical variance of the payoff on the Monte-Carlo samples these bounds are derived. The confidence bounds then yield a range around the mean that includes the true price with 95\% probability. We pick two free parameters $p_{i_1},\ p_{i_2}$ out of \eqref{def-para}, $1\leq i_1<i_2\leq D$, in each model setup and fix all other parameters at reasonable constant values. In this section we define the discrete parameter grid $\overline{\mathcal{P}} \subseteq[\underline{p}_{i_1},\overline{p}_{i_1}]\times [\underline{p}_{i_2},\overline{p}_{i_2}]$  by
	\begin{equation}
	\label{MC_Testgrid}
		\begin{split}
			\overline{\mathcal{P}} =&\ \left\{\left(p_{i_1}^{k_{i_1}},p_{i_2}^{k_{i_2}}\right),\ {k_{i_1}},{k_{i_2}}\in\{0,\dots,40\}\right\},\\
			p_{i_j}^{k_{i_j}} =&\ \underline{p}_{i_j} + \frac{k_{i_j}}{40}\left( \overline{p}_{i_j}- \underline{p}_{i_j}\right),\ k_{i_j}\in\{0,\dots,40\},\ j\in\{1,2\},
		\end{split}
	\end{equation}	
and call $\overline{\mathcal{P}}$ test grid. On this test grid the largest confidence bound is 0.025 an on average lees than 0.013. For the finite difference method we find that the absolute error between numerical approximation and option price is below 0.005 on all computed parameter tuples in $\overline{\mathcal{P}}$. This error bound was computed by comparing each approximation to the limit of the sequence of finite difference approximations with increasing grid size. In our calculations we work with a grid size in time as well as in space (log-moneyness) of $50\cdot\max\{1,T\}$ and compared the result to the resulting prices using grid sizes of $1000\cdot\max\{1,T\}$. This grid size has been determined as sufficient for the limit due to  hardly any changes compared to grid sizes of $500\cdot\max\{1,T\}$.

Here, our main concern is the accuracy of the Chebyshev interpolation when we vary for each option the parameters strike and maturity analogously to the previous section. For $N\in\{5,10,30\}$, we precompute the Chebyshev coefficients as defined in~\eqref{def:Chebycj} with $D=2$ where always $N_1=N_2=N$. An overview of fixed and free parameters in our model selection is given in Table~\ref{tab:ChebyD2Exoticsparam}. For computational simplicity in the Monte-Carlo simulation, we assume uncorrelated underlyings.

\begin{table}[h]
\begin{center}
\begin{tabular}{@{}lllllll@{}}
\toprule
\textbf{Model}  & \phantom{a} & \multicolumn{2}{c}{\textbf{fixed parameters}}                  & \phantom{a} & \multicolumn{2}{c}{\textbf{free parameters}}                \\ 
       &                 & \multicolumn{1}{c}{$p^1$} & \multicolumn{1}{c}{$p^2$} &          & \multicolumn{1}{c}{$p^1$} & \multicolumn{1}{c}{$p^2$} \\ \midrule
BS     &                 & $S^j_{0}=100$,                     & $\sigma_j=0.2$             &          & $K\in[83.33,\ 125]$    &    $T\in[0.5,2]$                       \\
       &                 &   $r=0.005$                        &                           &          &              &                           \\ \midrule

Heston &                 & $S^j_{0}=100$,                    & $\kappa_j=2$,             &          & $K\in[83.33,\ 125]$     &  $T\in[0.5,2]$   \\
       &                 &    $r=0.005$                  & $\theta_j=0.2^2$,           &         &                           &                           \\
       &                 &                           & $\sigma_j=0.3$,            &          &                           &                           \\
       &                 &                           & $\rho_j=-0.5$,                &          &                           &                           \\
       &                 &                           & $v_{j,0}=0.2^2$                &          &                           &                           \\
        \midrule

Merton &                 & $S^j_{0}=100$,                    & $\sigma_j=0.2$,             &          & $K\in[83.33,\ 125]$     & $T\in[0.5,2]$    \\
       &                 &       $r=0.005$               & $\alpha_j=-0.1$,           &         &                           &                           \\
       &                 &                           & $\beta_j=0.45$,            &          &                           &                           \\
       &                 &                           & $\lambda_j=0.1$                &          &                           &                           \\ \bottomrule
\end{tabular}
\end{center}
\caption{Parametrization of models, basket and path-dependent options. The model parameters are given for $j=1,\ldots,d$ to reflect the multivariate setting with free parameters strike $K$ and maturity $T$. Note that in contrast to the two dimensional Heston model described in Section \ref{sec:Heston_Model_for_Two_Assets} we use here in the numerical experiments a multivariate Heston model in which the volatility of each underlying is driven by its own volatility process.}
\label{tab:ChebyD2Exoticsparam}
\end{table}

Let us briefly define the payoffs of the multivariate basket and path-dependent options. The payoff profile of a basket option for $d$ underlyings is given as
$$f^K\big(S^1_T,\ldots,S^d_T\big)=\left(\left(\frac{1}{d}\sum_{j=1}^d S^j_T \right)-K\right)^{+}. $$
We denote $S_t=(S^1_t,\ldots,S^d_t)$, $\underline{S}^j_T:=\min_{0\le t\le T}S^j_t$ and $\overline{S}^j_T:=\max_{0\le t\le T}S^j_t$. A lookback option for $d$ underlyings is defined as
$$f^K\left(\overline{S}^1_T,\ldots,\overline{S}^d_T\right)=\left(\left(\frac{1}{d}\sum_{j=1}^d \overline{S}^j_T\right)-K\right)^{+}.$$ As a multivariate barrier option on $d$ underlyings we define the payoff 
$$f^K\big(\{S(t)\}_{0\le t\le T}\big)=\left(\left(\frac{1}{d}\sum_{j=1}^d S^j_T\right)-K\right)^{+}\cdot \mathbbm{1}_{\{\underline{S}^j_T\ge 80,\ j=1,\ldots,d\}}. $$
For an American put option the payoff is the same as for a European put,
$$f^K\big(S_t\big)=\left(K-S_t\right)^{+},$$
but the option holder has the right to exercise the option at any time $t$ up to maturity $T$.\\

\begin{table}[h!]
\begin{center}
\begin{tabular}{llcccc}
\toprule
\textbf{Model}&\textbf{Option}&$\varepsilon_{L^\infty}$&\textbf{MC price}&\textbf{MC conf. bound}& \textbf{CI price}\\
\hline
BS&Basket&$1.338\cdot10^{-1}$&8.6073&$1.171\cdot10^{-2}$   &8.4735\\
Heston&Basket&$9.238\cdot10^{-2}$&0.0009&$1.036\cdot10^{-4}$&0.0933\\
Merton&Basket&$9.815\cdot10^{-2}$&8.8491&$1.552\cdot10^{-2}$&8.7510\\
\hline
BS&Lookback&$2.409\cdot10^{-1}$&9.4623&$9.861\cdot10^{-3}$&9.2213\\
Heston&Lookback&$5.134\cdot10^{-1}$&0.0314&$6.472\cdot10^{-4}$&-0.4820\\
Merton&Lookback&$2.074\cdot10^{-1}$&1.0919&$9.568\cdot10^{-3}$&0.8844\\
\hline
BS&Barrier&$1.299\cdot10^{-1}$&1.0587&$5.092\cdot10^{-3}$&1.1887\\
Heston&Barrier&$1.073\cdot10^{-1}$&2.7670&$9.137\cdot10^{-3}$&2.6597\\
Merton&Barrier&$9.916\cdot10^{-2}$&1.3810&$1.102\cdot10^{-2}$&1.4802\\
\bottomrule
\end{tabular}
\caption{Interpolation of exotic options with Chebyshev interpolation. $N=5$ and $d=5$ in all cases. In addition to the $L^{\infty}$ errors the table displays the Monte-Carlo (MC) prices, the Monte-Carlo confidence bounds and the Chebyshev Interpolation (CI) prices for those parameters at which the $L^{\infty}$ error is realized.}\label{table-exotics_N5}
\end{center}
\end{table}

\begin{table}[h!]
\begin{center}
\begin{tabular}{llcccc}
\toprule
\textbf{Model}&\textbf{Option}&$\varepsilon_{L^\infty}$&\textbf{MC price}&\textbf{MC conf. bound}& \textbf{CI price}\\
\hline
BS&Basket&$2.368\cdot10^{-3}$&2.4543&$7.493\cdot10^{-3}$&2.4566\\
Heston&Basket&$2.134\cdot10^{-3}$&3.1946&$1.073\cdot10^{-2}$&3.1925\\
Merton&Basket&$3.521\cdot10^{-3}$&6.1929&$2.231\cdot10^{-2}$&6.1894\\
\hline
BS&Lookback&$2.861\cdot10^{-2}$&0.9827&$4.197\cdot10^{-3}$&0.9541\\
Heston&Lookback&$1.098\cdot10^{-1}$&2.0559&$4.826\cdot10^{-3}$&2.1656\\
Merton&Lookback&$3.221\cdot10^{-2}$&4.7072&$1.264\cdot10^{-2}$&4.7394\\
\hline
BS&Barrier&$4.414\cdot10^{-3}$&5.3173&$1.725\cdot10^{-2}$&5.3129\\
Heston&Barrier&$5.393\cdot10^{-3}$&0.7158&$5.879\cdot10^{-3}$&0.7212\\
Merton&Barrier&$3.376\cdot10^{-3}$&9.2688&$2.302\cdot10^{-2}$&9.2722\\
\bottomrule
\end{tabular}
\caption{Interpolation of exotic options with Chebyshev interpolation. $N=10$ and $d=5$ in all cases. In addition to the $L^{\infty}$ errors the table displays the Monte-Carlo (MC) prices, the Monte-Carlo confidence bounds and the Chebyshev Interpolation (CI) prices for those parameters at which the $L^{\infty}$ error is realized.}\label{table-exotics_N10}
\end{center}
\end{table}

\begin{table}[h!]
\begin{center}
\begin{tabular}{llcccc}
\toprule
\textbf{Model}&\textbf{Option}&$\varepsilon_{L^\infty}$&\textbf{MC price}&\textbf{MC conf. bound}& \textbf{CI price}\\
\hline
BS&Basket&$1.452\cdot10^{-3}$&5.1149&$1.200\cdot10^{-2}$&5.1163\\
Heston&Basket&$1.047\cdot10^{-3}$&7.6555&$1.371\cdot10^{-2}$&7.6545\\
Merton&Basket&$3.765\cdot10^{-3}$&7.2449&$2.359\cdot10^{-2}$&7.2412\\
\hline
BS&Lookback&$3.766\cdot10^{-3}$&25.9007&$1.032\cdot10^{-2}$&25.9045\\
Heston&Lookback&$1.914\cdot10^{-3}$&16.4972&$9.754\cdot10^{-3}$&16.4991\\
Merton&Lookback&$3.646\cdot10^{-3}$&27.1018&$1.623\cdot10^{-2}$&27.1054\\
\hline
BS&Barrier&$5.331\cdot10^{-3}$&5.6029&$1.730\cdot10^{-2}$&5.6082\\
Heston&Barrier&$2.486\cdot10^{-3}$&3.6997&$1.353\cdot10^{-2}$&3.6972\\
Merton&Barrier&$4.298\cdot10^{-3}$&6.6358&$2.309\cdot10^{-2}$&6.6315\\
\bottomrule
\end{tabular}
\caption{Interpolation of exotic options with Chebyshev interpolation. $N=30$ and $d=5$ in all cases. In addition to the $L^{\infty}$ errors the table displays the Monte-Carlo (MC) prices, the Monte-Carlo confidence bounds and the Chebyshev Interpolation (CI) prices for those parameters at which the $L^{\infty}$ error is realized.}\label{table-exotics_N30}
\end{center}
\end{table}

We now turn to the results of our numerical experiments. In order to evaluate the accuracy of the Chebyshev interpolation we look for the worst case error $\varepsilon_{L^{\infty}}$. The absolute error of the Chebychev interpolation method can be directly computed by comparing the interpolated option prices with those obtained by the reference numerical algorithm i.e. either the Monte-Carlo or the Finite Difference method. Since the Chebychev interpolation matches the reference method on the Chebychev nodes, we will use the out-of-sample test grid as in \eqref{MC_Testgrid}.  Table \ref{table-exotics_N5} shows the numerical results for the basket and path-dependent options for $N=5$, Table \ref{table-exotics_N10} for $N=10$ and Table \ref{table-exotics_N30} for $N=30$. In addition to the $L^{\infty}$ errors the tables display the Monte-Carlo (MC) prices, the Monte-Carlo confidence bounds and the Chebyshev Interpolation (CI) prices for those parameters at which the $L^{\infty}$ error is realized. 

The results show that for $N=30$ the accuracy is for all selected options at a level of $10^{-3}$. We see that the Chebyshev interpolation error is dominated by the Monte-Carlo confidence bounds to a degree which renders it negligible in a comparison between the two. For basket and barrier options the $L^{\infty}$ error already reaches satisfying levels of order $10^{-3}$ at $N=10$ already. Again, the Chebyshev approximation falls within the confidence bounds of the Monte-Carlo approximation. Thus, Chebyshev interpolation with only $121=(10+1)^2$ nodes suffices for mimicking the Monte Carlo pricing results. This statement does not hold for lookback options, where the $L^{\infty}$ error still differs noticeably when comparing $N=10$ to $N=30$. As can be seen from Table\tild \ref{table-exotics_N5} Chebyshev interpolation with $N=5$ may yield unreliable pricing results. For lookback options in the Heston model we even observe negative prices in individual cases.
\begin{table}[h!]
\begin{center}
\begin{tabular}{lccc}
\toprule
$N$&$\varepsilon_{L^\infty}$& \textbf{FD price}&\textbf{CI price}\\
\hline
5&$3.731\cdot10^{-3}$&1.9261&1.9224\\
10&$1.636\cdot10^{-3}$&12.0730&12.0746\\
30&$3.075\cdot10^{-3}$&6.3317&6.3286\\
\bottomrule
\end{tabular}
\caption{Interpolation of one-dimensional American puts with Chebyshev interpolation in the \BS model. In addition to the $L^{\infty}$ errors the table displays the Finite Differences (FD) prices and the Chebyshev Interpolation (CI) prices for those parameters at which the $L^{\infty}$ error is realized.}\label{table-exotics_American}
\end{center}
\end{table}
Chebyshev pricing of American options in the \BS model is even more accurate as illustrated in Table \ref{table-exotics_American}. Here, already for $N=5$ the accuracy of the reference method is achieved. We conclude that the Chebyshev interpolation is highly promising for the evaluation of multivariate basket and path-dependent options. Yet the accuracy of the interpolation critically depends on the accuracy of the reference method at the nodal points which motivates further analysis that we perform in the subsequent subsection.

\subsubsection{Interaction of Approximation Errors at Nodal Points and Interpolation Errors}
\label{accuracy_exotics_bound}
The Chebyshev method is most promising for use cases, where computationally intensive pricing methods are required. Then, for computing the prices at the Chebyshev nodes in order to set up the interpolation, the issue of distorted prices at the Chebyshev nodes and their consequences rises naturally. The observed noisy prices at the Chebyshev nodes are
\begin{align*}
\Price_{\varepsilon}^{p^{(k_1,\dots,k_D)}}=\Price^{p^{(k_1,\dots,k_D)}} + \varepsilon^{p^{(k_1,\dots,k_D)}},
\end{align*}
where $\varepsilon^{p^{(k_1,\dots,k_D)}}$ is the approximation error introduced by the underlying numerical technique at the Chebyshev nodes. Due to linearity, the resulting interpolation is of the form
	\begin{equation}
		I_{\overline{N}}(\Price_{\varepsilon}^{(\cdot)})(p) = I_{\overline{N}}(\Price^{(\cdot)})(p)+I_{\overline{N}}(\varepsilon^{(\cdot)})(\cdot)
	\end{equation}
with the error function
	\begin{equation}
		\varepsilon^{(p)}=\sum_{j_D=0}^{N_D}\ldots\sum_{j_1=0}^{N_1}c_{j_1,\ldots,j_D}^\varepsilon T_{j_1,\ldots,j_D}(p),
	\end{equation}
with the coefficients $c^\varepsilon_j$ for $j=(j_1,\dots, j_D)\in J$ given by
	\begin{equation}
		c^\varepsilon_j = \Big( \prod_{i=1}^D \frac{2^{\1_{\{0<j_i<N_i\}}}}{N_i}\Big)\sum_{k_1=0}^{N_1}{}^{''}\ldots\sum_{k_D=0}^{N_D}{}^{''} \varepsilon^{p^{(k_1,\dots,k_D)}}\prod_{i=1}^D \cos\left(j_i\pi\frac{k_i}{N_i}\right).
	\end{equation}
If $\varepsilon^{p^{(k_1,\dots,k_D)}}\le \overline{\varepsilon}$ for all Chebyshev nodes $p^{(k_1,\dots,k_D)}$, we obtain
\begin{equation}
|\varepsilon^{(p)}|\le 2^D \bar{\varepsilon}\prod_{i=1}^D (N_i+1),
\end{equation}
since the Chebyshev polynomials are bounded by $1$. 
This yields the following remark.
\begin{remark}\label{Remark_Error}
Let $\OP\ni p\mapsto \Price^p$ be given as in Theorem \ref{Asymptotic_error_decay_multidim} and assume that $\varepsilon^{p^{(k_1,\dots,k_D)}}\le \overline{\varepsilon}$ for all Chebyshev nodes $p^{(k_1,\dots,k_D)}$. Then
  \begin{equation}
  \label{absch-TCana}
  \begin{split}
 \max_{p\in\OP}\big|\Price^p& - I_{\overline{N}}(\Price^{(\cdot)}_\varepsilon)(p)\big|
  \\
 &\le2^{\frac{D}{2}+1}\cdot V \cdot\left(\sum_{i=1}^D\varrho_i^{-2N_i}\prod_{j=1}^D\frac{1}{1-\varrho_j^{-2}}\right)^{\frac{1}{2}} + 2^D \bar{\varepsilon}\prod_{i=1}^D (N_i+1).
  \end{split}
  \end{equation} 
\end{remark}
The following example shall illustrate the practical consequences of Remark \ref{Remark_Error}. In the setting of Corollary \ref{cor-callLevy} we set  $[\underline{S_0}/\overline{K},\overline{S_0}/\underline{K}]=[0.8,1.2]$, $[\underline{T},\overline{T}]=[0.5,2]$. This results in $\zeta^1=\frac{2.5}{1.5}=\frac{5}{3}$ and $\zeta^2=\frac{2}{0.4}=5$. Thus, for $\varrho_1=2.9\in(1,3)$ and $\varrho_2=9.8\in(1,5+\sqrt{24})$, Remark \ref{Remark_Error} yields with $N_1=N_2=6$,
     \begin{align*}
   \max_{p\in\OP}\big|\Price^p - I_{\overline{N}}(\Price^{(\cdot)})(p)\big|&\le 0.0072 + 196\cdot\bar{\epsilon}.
  \end{align*}
  In this example, the accuracy of the reference method has to reach a level of $10^{-5}$ to guarantee an overall error of order $10^{-3}$. This demonstrates a trade-off between increasing $N_1$ and $N_2$ compared to the accuracy of the reference method. The error bound above is rather conservative. Our experiments from the previous section suggest that this bound highly overestimates the errors empirically observed. However, the presented error bound from Remark \ref{Remark_Error} can guarantee a desired accuracy by determining an adequate number of Chebyshev nodes and the corresponding accuracy of the reference method used at the Chebyshev nodes. For practical implementation we suggest the following procedure. For a prescribed accuracy the $N_i$, $i=1,\dots,D$, can be determined from the first term in \eqref{absch-TCana} by choosing $N_i$, $i=1,\dots,D$, as small as possible such that the prescribed accuracy is attained. Accordingly, the accuracy that the reference method needs to achieve is bounded by the second term. A very accurate reference method in combination with small $N_i$, $i=1,\dots,D$, promise best results. With this rule of thumb in mind the experiments of Section \ref{sec:ChebyshevEfficiency_MC} below have been conducted.

\subsection{Study of the gain of efficiency}
In the previous section we investigated the accuracy of the Chebyshev polynomial interpolation method using Fourier, Monte-Carlo and Finite Difference as reference pricing methods. Finally, we investigate the gain in efficiency achieved by the method in comparison with Fourier pricing as well as in comparison to Monte Carlo pricing. In Section \ref{sec:ChebyshevEfficiency} we compute the results on a standard PC with an Intel i5 CPU, 2.50~GHz with cache size 3~MB. In Section \ref{sec:ChebyshevEfficiency_MC} we used a PC with Intel Xeon~CPU with 3.10~GHz with 20~MB SmartCache. All codes are written in Matlab R2014a.

\subsubsection{Comparison to Fourier pricing}
\label{sec:ChebyshevEfficiency}
Here, we compare the method to Fourier pricing. We choose the pricing problem of a call option on the minimum of two assets as an example. Today's values of the underlying two assets are fixed at
	\begin{equation}
		S_0^{1} = 1,\qquad S_0^{2} = 1.2.
	\end{equation}
Modeling the future development of the underlyings, $(S_t^{j})_{t\geq 0}$, $j\in\{1,2\}$, we consider two bivariate models, separately. First, the two assets will be driven by the bivariate \BS model of Section~\ref{sec:Levy1}. The bivariate \BS model is parametrized by a covariance matrix $\sigma\in\IR^{2\times 2}$ that we choose to be given by
	\begin{equation}
		\sigma_{11} = 0.2^2,\qquad\sigma_{12} = 0.01,\qquad\sigma_{22} = 0.25^2.\label{eq:ChebyEffBSParams}
	\end{equation}
In a second efficiency study, asset movements follow the more involved bivariate Heston model in the version of Section~\ref{sec:Heston_Model_for_Two_Assets} above for which we choose the parametrization	
	\begin{equation}
	\begin{aligned}
		v_0 &= 0.05,\qquad  &\sigma_1 &= 0.15,\qquad  &\rho_{13} &= 0.01,\\
		\kappa &= 0.4963,\qquad &\sigma_2 &= 0.2,\qquad &\rho_{12} &= 0,\label{eq:ChebyEffHestonParams}\\
		\theta &= 0.2286,\qquad &\sigma_3 &= 0.1,\qquad &\rho_{23} &= 0.02.
	\end{aligned}
	\end{equation}
In both cases we neglect interest rates, thus setting $r=0$. The benchmark method, that is Fourier pricing, is evaluated using Matlab's \texttt{quad2d} routine. We prescribe an absolute and relative accuracy of at least $10^{-8}$ from the integration result and integrate the Fourier integrand over the domain \mbox{$\Omega = [-50,\ 50] \times [0,50]$}, with a maximum number of $4000$ function evaluations.
We use Fourier integration with the same accuracy specifications at the Chebyshev nodes to set up the Chebyshev method for pricing based on strike $K$ and maturity $T$ as the two free parameters taking values in the intervals
	\begin{equation}
	\begin{aligned}
		K&\in [K_\text{min},\ K_\text{max}],&\qquad K_\text{min}&=0.8,& K_\text{max}&=1.2,\\ 
		T&\in [T_\text{min},\ T_\text{max}],&\qquad T_\text{min}&=0.5,& T_\text{max}&=2.
	\end{aligned}
	\end{equation}
For a fair comparison, the number of Chebyshev polynomials is chosen such that Chebyshev interpolation prices yield an accuracy that matches the accuracy of the benchmark method resulting in
	\begin{equation}
		N_\text{Cheby}^\text{BS} = 11\quad\text{and}\quad N_\text{Cheby}^\text{Heston} = 23,
	\end{equation}
for the bivariate \BS model and the bivariate Heston model, respectively. \mbox{Figure~\ref{fig:ChebyshevEfficiencySurfaces}} illustrates the absolute errors over the whole $K\times T$ domain of interest between Fourier pricing and the Chebyshev method for both models, with the Chebyshev interpolator being based on $N_\text{Cheby}^\text{BS}+1$ polynomials in the \BS model case and $N_\text{Cheby}^\text{Heston}+1$ polynomials in the Heston model case.
	
\begin{figure}
	\centering
	\begin{subfigure}[t]{.48\linewidth}
		\includegraphics{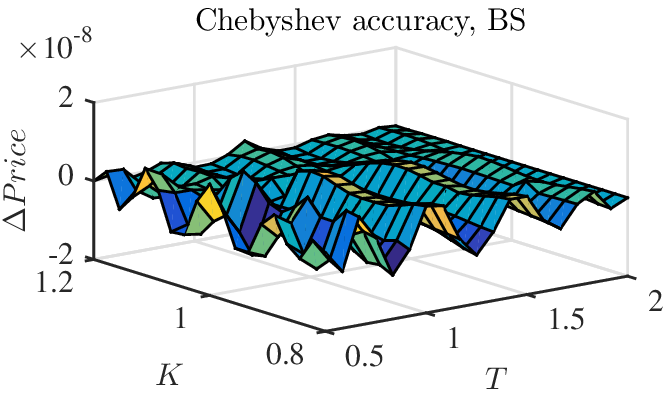}
	\end{subfigure}
	\begin{subfigure}[t]{.48\linewidth}
		\includegraphics{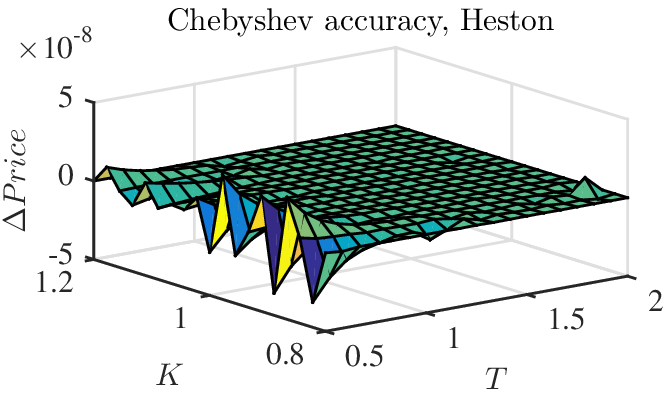}
	\end{subfigure}
\caption{Left: Difference between prices from the Fourier method and Chebyshev interpolation in the bivariate \BS model over the whole parameter domain of interest. The model is parametrized as indicated by~\eqref{eq:ChebyEffBSParams}. Chebyshev interpolation is based on $N_\text{Cheby}^\text{BS}+1 = 12$ Chebyshev polynomials. Right: The respective plot for the Heston model parametrized as in~\eqref{eq:ChebyEffHestonParams}. Here, Chebyshev interpolation is based on $N_\text{Cheby}^\text{Heston}+1 = 24$ Chebyshev polynomials. We achieve an absolute accuracy of order $10^{-8}$ in both cases, thus matching the accuracy that the benchmark method Fourier pricing reaches.} \label{fig:ChebyshevEfficiencySurfaces}
\end{figure}

In order to set up an efficiency study, we will compute sets of prices with increasing number of parameter tupels. To this end, when the offline phase of the Chebyshev method has been completed we compute~$98$ pricing surfaces, that is for each $M\in\{3,\dots,100\}$ we compute prices for all parameter tuples from $\Theta_M$ defined by
	\begin{align}
		\Theta_M = \big\{(K_i^{M}, T_j^{M})\ \big|\ &K_i^{M} = K_\text{min} + \frac{i-1}{M-1}(K_\text{max}-K_\text{min}),\notag\\
		&T_j^{M} = T_\text{min} + \frac{j-1}{M-1}(T_\text{max}-T_\text{min}),\text{ for } 1\leq i,j\leq M\big\}.\label{eq:defThetaMFourier}
	\end{align}
	
The computation time consumed by the Chebyshev offline phase is stored. Also, for each $M\in\{3,\dots,100\}$, run-times for deriving all $|\Theta_M|=M^2$ prices are measured and stored for both routines, the Fourier pricing method and the Chebyshev interpolation algorithm. Figure~\ref{fig:ChebyshevEfficiencyRuntimes} depicts these run-time measurements and Table~\ref{tab:ChebyshevEfficiencyRuntimes} provides a second perspective.
\begin{figure}
	\centering
	\makebox[0pt]{\includegraphics[width=.55\linewidth]{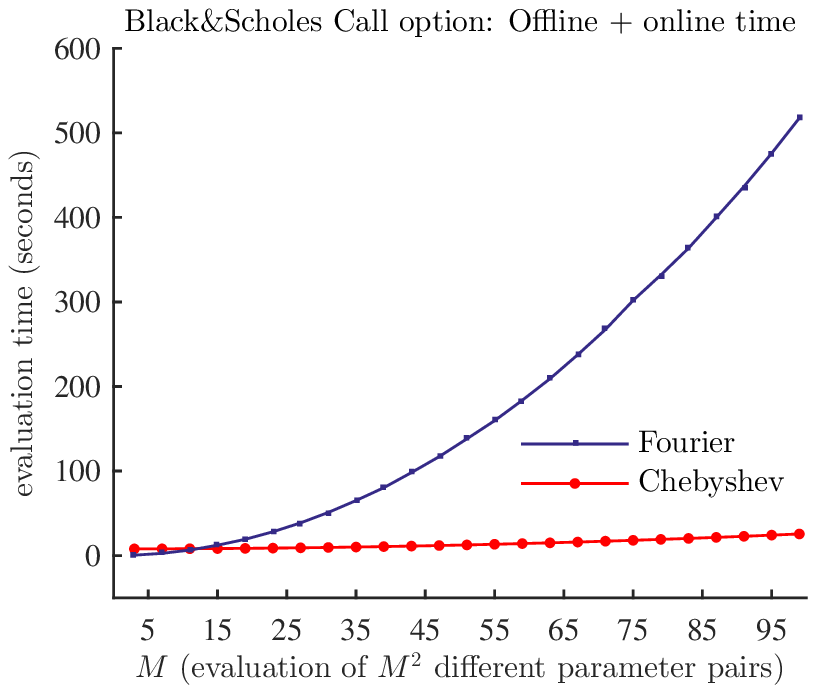}
	\includegraphics[width=.55\linewidth]{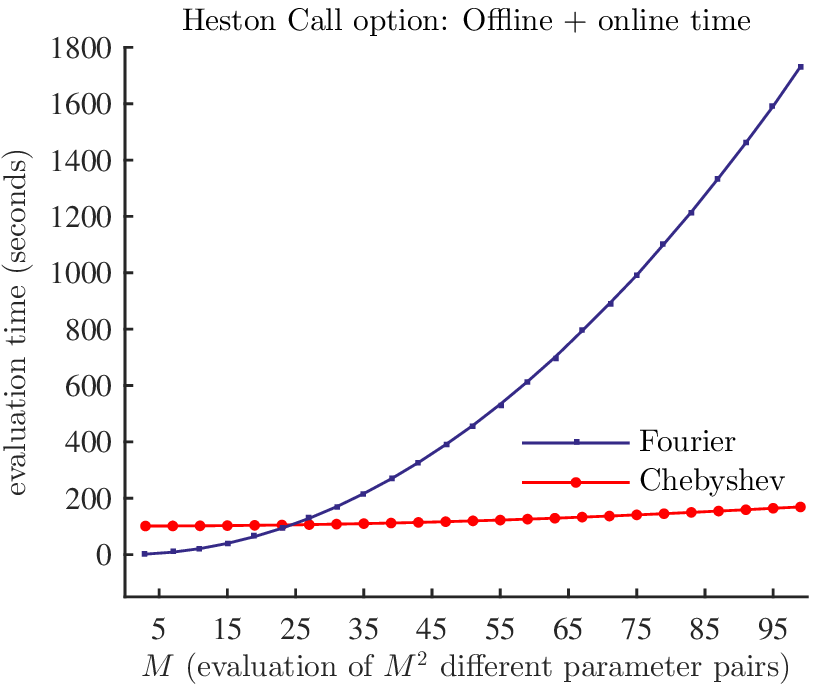}}
 \caption{Comparison of pricing times between Fourier pricing and the Chebyshev method for a call option on the minimum of two assets in the \BS model (left) and the Heston model (right). For each $M\in\{3,\dots,100\}$, run-times for deriving option prices for all $M^2$ parameter tupels from $\Theta_M$ defined by~\eqref{eq:defThetaMFourier} are depicted. In both model cases, computation times for the Chebyshev method contain the duration of the offline phase that has to be conducted once in the beginning. The Fourier and the Chebyshev curves roughly intersect when $M=N^\text{BS}_\text{Cheby}+1=12$ for the \BS model and when $M=N^\text{Heston}_\text{Cheby}+1=24$ for the Heston model, respectively. }\label{fig:ChebyshevEfficiencyRuntimes}
\end{figure}

In the \BS model case, the offline phase required $T_\text{offline}^\text{BS} = 8$ seconds for deriving option prices at all $(N_\text{Cheby}^\text{BS}+1)^2=144$ Chebyshev nodes. The Heston model required $T_\text{offline}^\text{Heston}= 101$ seconds for the $(N_\text{Cheby}^\text{Heston}+1)^2=576$ supporting prices. 
Taking this initial investment into account deems pricing with the Chebyshev method rather costly when only few option prices are derived after the offline phase has been completed. Yet, as Figure~\ref{fig:ChebyshevEfficiencyRuntimes} shows, the increase in pricing speed that is achieved once the Chebyshev algorithm has been set up eventually outpaces Fourier pricing. From our experiments we conclude that the Chebyshev method already outruns its benchmark Fourier pricing in terms of total run-times when the number of prices to be computed exceeds $(N^\text{BS}_\text{Cheby}+1)^2$ or $(N^\text{Heston}_\text{Cheby}+1)^2$, respectively. Additionally, Table \ref{tab:ChebyshevEfficiencyRuntimes} highlights that in both models for a total number of $50^2$ parameter tuples, the Chebyshev method exhibits a significant decrease in (total) pricing run-times. For the maximal number investigated, i.e. $100^2$ parameter tuples, pricing in the \BS model results in 95\% of run-time savings and 90\% run-time savings in the Heston model case in our implementation. This results from the fact the online phase in the Chebyshev method consist of computationally cheap polynomial evaluations and elementary assembling.

\begin{table}[]
\centering
\resizebox{\textwidth}{!}{  \begin{tabular}{lllllcllll}
\toprule
 & \multicolumn{4}{c}{\textbf{BS}}      &           & \multicolumn{4}{c}{\textbf{Heston}}                \\
$M$    & 10   & 50     & 75     & 100   &  \phantom{a}   & 10     & 50     & 75     & 100     \\
\midrule
\parbox[0pt][2.1em][c]{0cm}{}$T_\text{online}^\text{Cheby}$ ($s$)    & 0.18 & 4.54   & 10.20  & 18.11 &  & 0.70   & 17.58  & 39.66  & 69.82   \\ 
\parbox[0pt][2.1em][c]{0cm}{}$T_\text{offline+online}^\text{Cheby}$ ($s$)    & 8.06 & 12.42   & 18.07  & 25.98 &  & 101.96 & 118.85  & 140.92  & 171.08   \\ 
\parbox[0pt][2.1em][c]{0cm}{}$T^\text{Fourier}$ ($s$)  & 5.34 & 131.96 & 301.82 & 528.74 & & 17.60  & 442.62 & 991.33 & 1788.08 \\
\parbox[0pt][2.1em][c]{0cm}{}$\dfrac{T_\text{offline+online}^\text{Cheby}}{T^\text{Fourier}}$   & 151\%  & 9.41\%   & 5.99\%   & 4.91\%  &  & 579.27\% & 26.85\%  & 14.22\%  & 9.57\% \\
\bottomrule 
\end{tabular}}
\caption{
Efficiency study for the bivariate \BS model and the bivariate Heston model: A selection of the results fully depicted in Figure~\ref{fig:ChebyshevEfficiencyRuntimes}. With increasing number of computed prices, the Chebyshev algorithm increasingly profits from the initial investment of the offline phase.}
\label{tab:ChebyshevEfficiencyRuntimes}
\end{table}

\FloatBarrier
\subsubsection{Comparison to Monte-Carlo pricing}
\label{sec:ChebyshevEfficiency_MC}
In this section we choose a multivariate lookback option in the Heston model, based on 5 underlyings, as example. For the efficiency study we first vary one parameter, then we vary two.\\
\textbf{Variation of one model parameter}\\
For the multivariate lookback option in the Heston model the following parameters are fixed with $j=1,\dots,5$ as
	\begin{equation}
	\begin{aligned}
		S^j_{0} &= 100,\qquad  &r &= 0.005,\qquad  &K &= 100,\qquad &T&=1,\\
		\kappa_j &= 2,\qquad &\theta_j &= 0.2^2,\qquad &\rho_j &=-0.5,\qquad &v_{j,0}&=0.2^2.\label{eq:ChebyEffMCHestonParams}
	\end{aligned}
	\end{equation}
As free parameter in the Chebyshev interpolation we pick the volatility of the volatility coefficient $\sigma=\sigma_j$, $j=1,\ldots,5$,
	\begin{equation}	
		\sigma\in [\sigma_\text{min},\ \sigma_\text{max}],\qquad \sigma_\text{min}=0.1,\quad \sigma_\text{max}=0.5.	
	\end{equation}
The benchmark method is the Monte-Carlo pricing, again with $10^6$ sample paths, antithetic variates as variance reduction technique and 400 time steps per year. We refer to this setting as benchmark setting.

Following the discussion from Section \ref{accuracy_exotics_bound}, for the evaluation of the prices at the nodal points we guarantee a small $\bar{\varepsilon}$ by the Monte-Carlo method we enrich the Monte-Carlo setting to $5\cdot 10^6$ sample paths, antithetic variates and 400 time steps per year. In Table~\ref{table-MC_High_Accuracy} we present the accuracy results for the Chebyshev interpolation with $N^{\text{Heston}}_{\text{Cheby}}=6$ based on the enriched Monte-Carlo setting. To this aim, we compare the absolute differences of the Chebyshev prices and the enriched Monte-Carlo prices on the test grid $\overline{\mathcal{P}} \subseteq[\underline{p},\overline{p}]$,
	\begin{equation}	
		\begin{split}
			\overline{\mathcal{P}} =&\ \left\{\left(\sigma^k\right),\ k\in\{0,\dots,20\}\right\},\\
			\sigma^k =&\ \sigma_\text{min} + \frac{k}{20}\left( \sigma_\text{max}- \sigma_\text{min}\right),\ k\in\{0,\dots,20\}.
		\end{split}
	\end{equation}

	The highest observed error on the test grid is at a level of $10^{-2}$. On the same test grid the benchmark Monte-Carlo setting has a worst case confidence bound of $1.644\cdot10^{-2}$ and by comparing the benchmark Monte-Carlo prices to the enriched Monte-Carlo prices on this test grid, the maximal absolute error is $7.361\cdot10^{-3}$. Therefore, we conclude that the Monte-Carlo benchmark setting and the presented Chebyshev interpolation method have a roughly comparable accuracy and on the basis of this accuracy study we now turn to the comparison of run-times.
	
We compare run-times of the Chebyshev interpolation with $N^{\text{Heston}}_{\text{Cheby}}=6$, in which the offline phase is based on the enriched Monte-Carlo setting, to the run-times of the Monte-Carlo benchmark setting described above.  
\begin{table}[]
\centering
\resizebox{0.7\textwidth}{!}{
\begin{tabular}{ccccc}
\toprule
\textbf{Varying}&$\varepsilon_{L^\infty}$&\textbf{MC price}&\textbf{MC conf. bound}& \textbf{CI price}\\
\hline
$\sigma$&9.970$\cdot10^{-3}$&18.6607&$4.592\cdot10^{-3}$&18.6707\\
\bottomrule
\end{tabular}}
\caption{Interpolation of multivariate lookback options with Chebyshev interpolation for $N=6$ based on an enriched Monte-Carlo setting with $5\cdot 10^6$ sample paths, antithetic variates and 400 time steps per year. In addition to the $L^{\infty}$ error on the test grid we also report the Monte-Carlo (MC) price, the Monte-Carlo confidence bound and the Chebyshev Interpolation (CI) price for those parameters at which the $L^{\infty}$ error is realized. We observe that the accuracy of the Chebyshev interpolation for $N=6$ is roughly in the same range as the accuracy of the benchmark Monte-Carlo setting (worst case confidence bound of $1.644\cdot10^{-2}$ and worst case error of $7.361\cdot10^{-3}$).}\label{table-MC_High_Accuracy}
\end{table}

Table~\ref{tab:ChebyshevEfficiencyRuntimes_MC_1D} provides the respective results. The results for $M=1$ have been empirically measured, all others have been extrapolated from that since for each parameter set, the same amount of computation time would have to be invested. The table indicates that from $M=50$ onwards interpolation by Chebyshev is faster. In Figure~\ref{fig:Cheby_Efficiency_MC_Heston_1D} we present additionally for each $M=1,\ldots,100$ the run-times of the Chebyshev interpolation method, including the offline phase, compared to the Monte-Carlo method. Here we observe that for $M=35$ both lines intersect and for $M>35$ the Chebyshev interpolation method is faster. Contrary to Section \ref{sec:ChebyshevEfficiency}, the intersection of the two lines does not occur at $M=N^{\text{Heston}}_{\text{Cheby}}+1$. This reflects the fact that in the offline phase for the Chebyshev interpolation a Monte-Carlo method with more sample paths has been used.

\begin{figure}[htcp!]
	\centering
		\includegraphics[width=1\textwidth]{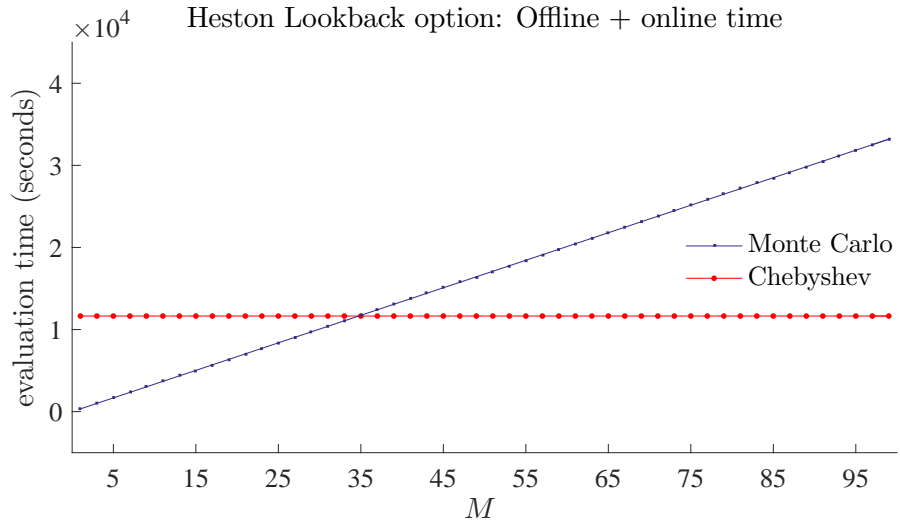}

	\caption{Effiency study for a multidimensional lookback option in the Heston model with 5 underlyings varying one model parameter $\sigma$. Comparison of run-times of Monte-Carlo pricing with Chebyshev pricing including the offline phase. Both methods have been set up to deliver comparable accuracies. We observe that both curves roughly intersect at $M=35$.}\label{fig:Cheby_Efficiency_MC_Heston_1D}
\end{figure}

\begin{table}[]
\centering
\resizebox{0.7\textwidth}{!}{  \begin{tabular}{lcccc}
\toprule
$M$      & 1     & 10     & 50     & 100     \\
\midrule
\parbox[0pt][2.1em][c]{0cm}{}$T_\text{online}^\text{Cheby}$ ($s$)    &2.7$\cdot10^{-5}$&2.7$\cdot10^{-4}$&1.4$\cdot10^{-3}$&2.7$\cdot10^{-3}$  \\ 
\parbox[0pt][2.1em][c]{0cm}{}$T_\text{offline+online}^\text{Cheby}$ ($s$)   &1.2$\cdot10^{4}$&1.2$\cdot10^{4}$&1.2$\cdot10^{4}$&1.2$\cdot10^{4}$   \\ 
\parbox[0pt][2.1em][c]{0cm}{}$T^\text{Monte-Carlo}$ ($s$)  &3.4$\cdot10^{2}$&3.4$\cdot10^{3}$&1.7$\cdot10^{4}$&3.4$\cdot10^{4}$ \\
\parbox[0pt][2.1em][c]{0cm}{}$\dfrac{T_\text{offline+online}^\text{Cheby}}{T^\text{Monte-Carlo}}$   &3473.4\%&347.3\%&69.5\%&34.73\% \\
\bottomrule 
\end{tabular}}
\caption{Efficiency study for a multivariate lookback option in the Heston model based on 5 underlyings. Here, we vary one model parameter and compare the Chebyshev results to Monte-Carlo. Both methods have been set up to deliver comparable accuracies. With increasing number of computed prices, the Chebyshev algorithm increasingly profits from the initial investment of the offline phase.}
\label{tab:ChebyshevEfficiencyRuntimes_MC_1D}
\end{table}

\noindent\textbf{Variation of two model parameters}\\
Now, we vary two parameters, we choose $\rho_j=\rho$, $j=1,\ldots,5$, and vary
	\begin{equation}
	\begin{aligned}
		\rho&\in [\rho_\text{min},\ \rho_\text{max}],&\qquad \rho_\text{min}&=-1,& \rho_\text{max}&=1,\\ 
		\sigma&\in [\sigma_\text{min},\ \sigma_\text{max}],&\qquad \sigma_\text{min}&=0.1,& \sigma_\text{max}&=0.5,
	\end{aligned}
	\end{equation}
fixing all other parameters to the values of setting \eqref{eq:ChebyEffMCHestonParams}.
Guaranteeing a roughly comparable accuracy between the Chebyshev interpolation method and the benchmark Monte-Carlo pricing, we use the following test grid $\overline{\mathcal{P}} \subseteq[\sigma_\text{min},\sigma_\text{max}]\times [\rho_\text{min},\rho_\text{max}]$,
	\begin{equation}
	\label{eq:defevalgrid}
		\begin{split}
			\overline{\mathcal{P}} =&\ \left\{\left(\sigma^{k_{1}},\rho^{k_{2}}\right),\ {k_{1}},{k_{2}}\in\{0,\dots,20\}\right\},\\
			\sigma^{k_{1}} =&\ \sigma_\text{min} + \frac{k_{1}}{20}\left( \sigma_\text{max}- \sigma_\text{min}\right),\ k_{1}\in\{0,\dots,20\},\\
			\rho^{k_{2}} =&\ \rho_\text{min} + \frac{k_{2}}{20}\left( \rho_\text{max}- \rho_\text{min}\right),\ k_{2}\in\{0,\dots,20\}.
		\end{split}
	\end{equation}	
In Table~\ref{table-MC_High_Accuracy2} we present the accuracy results for the Chebyshev interpolation with $N^{\text{Heston}}_{\text{Cheby}}=6$ based on the enriched Monte-Carlo setting.
Comparing the benchmark Monte-Carlo setting and the enriched Monte-Carlo setting on this test grid, we observe that the maximal absolute error is $2.791\cdot10^{-2}$ and the confidence bounds of the benchmark Monte-Carlo setting do not exceed $6.783\cdot10^{-2}$.

\begin{table}[]
\centering
\resizebox{0.7\textwidth}{!}{
\begin{tabular}{ccccc}
\toprule
\textbf{Varying}&$\varepsilon_{L^\infty}$&\textbf{MC price}&\textbf{MC conf. bound}& \textbf{CI price}\\
\hline
$\sigma$, $\rho$&$5.260\cdot10^{-2}$&5.239&$1.428\cdot10^{-2}$&5.292\\
\bottomrule
\end{tabular}}
\caption{Interpolation of multivariate lookback options with Chebyshev interpolation for $N=6$ based on an enriched Monte-Carlo setting with $5\cdot 10^6$ sample paths, antithetic variates and 400 time steps per year. In addition to the $L^{\infty}$ error on the test grid we also report the Monte-Carlo (MC) price, the Monte-Carlo confidence bound and the Chebyshev Interpolation (CI) price for those parameters at which the $L^{\infty}$ error is realized. We observe that the accuracy of the Chebyshev interpolation $N=6$ is roughly in the same range as the accuracy of the benchmark Monte-Carlo setting (worst case confidence bound of $6.783\cdot10^{-2}$ and worst case error of $2.791\cdot10^{-2}$).}\label{table-MC_High_Accuracy2}
\end{table}

For a run-time comparison, we show for different values of $M$ the run-times necessary to compute the prices for $M^2$ parameter tupels. Again, the run-times are measured for $M=1$ and extrapolated for other values of $M$. Table~ \ref{tab:ChebyshevEfficiencyRuntimes_MC_2D} presents the results. In Figure~\ref{fig:Cheby_Efficiency_MC_Heston_2D} for each $M=1,\ldots,100$ the run-times of the Chebyshev interpolation method, including the offline phase, compared to the Monte-Carlo method are presented. We observe that for $M=15$ both lines intersect and for $M>15$ the Chebyshev method outperforms its benchmark. Contrary to the case of varying only one parameter, the intersection of both lines occurs at a significantly lower value of $M$ due to the fact that for each $M$ pricing for $M^2$ parameter tupels is performed.
\begin{table}[]
\centering
\resizebox{0.7\textwidth}{!}{  \begin{tabular}{lcccc}
\toprule
          & \multicolumn{4}{c}{\textbf{Heston}}                \\
$M$    & 1     & 10     & 50     & 100     \\
\midrule
\parbox[0pt][2.1em][c]{0cm}{}$T_\text{online}^\text{Cheby}$ ($s$)    &7.1$\cdot10^{-4}$&7.1$\cdot10^{-2}$&1.8&7.1  \\ 
\parbox[0pt][2.1em][c]{0cm}{}$T_\text{offline+online}^\text{Cheby}$ ($s$)   &8.2$\cdot10^{4}$&8.2$\cdot10^{4}$&8.2$\cdot10^{4}$& 8.2$\cdot10^{4}$  \\ 
\parbox[0pt][2.1em][c]{0cm}{}$T^\text{Monte-Carlo}$ ($s$)  &3.4$\cdot10^{2}$&3.4$\cdot10^{4}$&8.4$\cdot10^{5}$&3.4$\cdot10^{6}$ \\
\parbox[0pt][2.1em][c]{0cm}{}$\dfrac{T_\text{offline+online}^\text{Cheby}}{T^\text{Monte-Carlo}}$   &24313.9\%&243.1\%&9.7\%&2.4\% \\
\bottomrule 
\end{tabular}}
\caption{Efficiency study for a multivariate lookback option in the Heston model based on 5 underlyings. Here, we vary two model parameters and compare the Chebyshev results to Monte-Carlo. Both methods have been set up to deliver comparable accuracies. With increasing number of computed prices, the Chebyshev algorithm increasingly profits from the initial investment of the offline phase.}
\label{tab:ChebyshevEfficiencyRuntimes_MC_2D}
\end{table}
	
\begin{figure}[htcp!]
	\centering
		\includegraphics[width=1\textwidth]{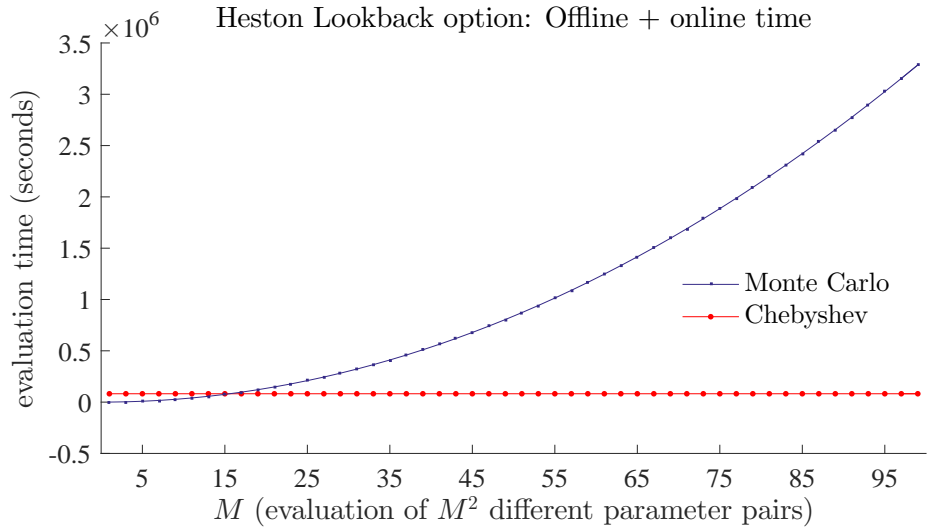}
	\caption{Effiency study for a multivariate lookback option in the Heston model based on 5 underlyings, varying the two model parameters $\sigma$ and $\rho$. Comparison of run-times for Monte-Carlo pricing and Chebyshev pricing including the offline phase. Both methods have been set up to deliver comparable accuracies. We observe that the Monte-Carlo and the Chebyshev curves roughly intersect at $M=15$.}\label{fig:Cheby_Efficiency_MC_Heston_2D}
\end{figure}	

Additionally, Table \ref{tab:ChebyshevEfficiencyRuntimes_MC_2D} highlights that for a total number of $50^2$ parameter tuples, the Chebyshev method exhibits a significant decrease in (total) pricing run-times. For the maximal number of $100^2$ parameter tuples that we investigated, pricing in either model resulted in more than 97\% of run-time savings in our implementation.  While the computation of $100^2$ Heston prices via the Monte-Carlo method consumes up to 39 days, the Chebyshev method computes the very same prices in 23 hours only. Note that only 7 seconds of this time span are consumed by actual pricing during the online phase.

	
\section{Conclusion and Outlook}

This article focuses on applying the Chebyshev method to European option pricing. Within this scope, Sections \ref{sec-Cheby}--\ref{sec-analyticity} establish theoretical convergence results and numerical case studies in Section \ref{sec-numerics} confirm these findings. Moreover, in experiments with Fourier pricing an accuracy of $10^{-5}$ is observed with less than ten Chebyshev nodes in each parameter, see Figure \ref{fig:ChebyErrorDecayAllLog10scale_call}. 
The financial implications of the high precision achieved with such a small number of interpolation nodes are twofold. First, it shows that there are interesting cases in which we observe an accuracy in the range of machine precision. In this comfortable situation without methodological risk we can ignore the fact that approximations are implemented. Second, compared with other sources of risk, already errors from a much lesser accuracy level can be ignored. If we agree that an accuracy of $10^{-4}$ is satisfactory, already $36$--$49$ interpolation nodes for the approximation of call option prices as reported in Figure \ref{fig:ChebyErrorDecayAllLog10scale_call} are sufficient. 

Additionally, also the numerical experiments for American, barrier and lookback options display promising results. A theoretical error analysis for nonlinear pricing problems is beyond the scope of the present article, while we are convinced that further investigations in this direction are valuable. 
For instance our analytic approach based on examinations of the Fourier representations can be adopted to barrier options in L\'evy models leading to the involvement of Wiener-Hopf factorizations, see \cite{EberleinGlauPapapantoleon2010b}. In general we expect the regularity analysis to become more challenging in the presence of nonlinearities. For American options the current work of \cite{Teichmann2015} may lead to regularity assertions for American options that are inherited from their corresponding European counterparts.

The theoretical and experimental results of our case studies show that the method can perform considerably well when few parameters are varied. As a consequence, we recommend the interpolation method for this case and also when solely the strike of a plain vanilla option is varied and fast Fourier methods are available. For calibration purposes for example, strikes are not given in a discrete logarithmic scale, which makes an additional interpolation necessary in order to apply FFT. Here, Chebyshev polynomials offer an attractive alternative. In particular, the maturity can be used as supplementary free variable. Moreover, for models with a low number of parameters, another path could be beneficial: Interpolating the objective function of the parameters directly. Then the optimization would boil down to a minimization of a tensorized polynomial, which could be exploited in further research. As may be seen from \cite{ArmentiCrepeyDrapeauPapapantoleon2015}, where the present article is applied for the first time, this advantage can also be exploited for other optimization procedures in finance for example in risk allocation.

The multivariate construction of the interpolation and the theoretical error analysis suggest that the empirically observed error behaviour extends to three and more varying parameters, as well, as long as analyticity is provided. More precisely, in case of analyticity, the rate is of order $\rho^{-\sqrt[D]{N}}$ for some constant $\rho$ depending on the domain of analyticity and $N$ the total number of interpolation nodes. For multivariate polynomial interpolation, the introduction of sparsity techniques promise higher efficiency, as for instance by compression techniques for tensors as reviewed by \cite{KoldaBader2009}. 
We address the issue of curse of dimensionality further in Ga{\ss}, Glau and Mair (2015)\nocite{GassGlauMair2015}, where we take a different route by replacing the Chebyshev interpolation with an empirical interpolation for Fourier pricing methods.

\appendix
\section{Proof of Theorem \ref{Asymptotic_error_decay_multidim}}\label{sec-proofAsymptotic}
\begin{proof}
In \cite[Proof of Lemma 7.3.3]{SauterSchwab2004} the proof is given for the following error bound:
 \begin{align*}
 \max_{p\in\OP}\big|f - I_{\overline{N}}(f)\big|
\le \sqrt{D}2^{\frac{D}{2}+1}V\varrho_{\min}^{-N}(1-\varrho_{\min}^{-2})^{-\frac{D}{2}},
  \end{align*}
where $N$ is the number of interpolation points in each of the $D$ dimensions, $\varrho_{\min}:=\min_{i=1}^D\varrho_i$ and $V$ the bound of $f$ on $B(\OP,\varrho)$ with $\OP = [-1,1]^D$. Here, we extend the proof by incorporating the different values of $N_i$, $i=1,\ldots,D$, as well as expressing the error bound with the different $\varrho_i$, $i=1,\ldots,D$. 

In general we work with a parameter space $\OP$ of hyperrectangular structure, $\OP=[\underline{p}_{1},\olp[1]]\times\ldots \times[\underline{p}_D,\olp[D]]$. With the introduced linear transformation in Section \ref{sec:Multi_Cheby} we have a transformation $\tau_{\OP}:[-1,1]^D\rightarrow\OP$ with
\begin{align*}
\tau_{\OP}(p)=\left(\overline{p}_i+\frac{\underline{p}_i-\overline{p}_i}{2}(1-p)\right)_{i=1}^D.
\end{align*}
Let $p\mapsto \Price^p$ be a function on $\OP$. We set $\widehat{\Price^p}=\Price^p\circ\tau_{\OP}(p)$. Further, let $\widehat{I}_{\overline{N}}(\widehat{\Price}^{(\cdot)})(p)$ be the Chebyshev interpolation of $\widehat{\Price^p}$ on $[-1,1]^D$. Then it holds
\begin{align*}
I_{\overline{N}}(\Price^{(\cdot)})(p)=\hat{I}_{\overline{N}}(\widehat{\Price}^{(\cdot)})(\cdot)\circ\tau^{-1}_{\OP}(p).
\end{align*}
Therewith, it directly follows
\begin{align*}
\Price^p - I_{\overline{N}}(\Price^{(\cdot)})(p)=\left(\widehat{\Price} -\widehat{I}_{\overline{N}}(\widehat{\Price}^{(\cdot)})(\cdot)  \right)\circ \tau^{-1}_{\OP}(p).
\end{align*}
Applying the error estimation from \cite[Lemma 7.3.3]{SauterSchwab2004} results
\begin{align*}
\big|\Price - I_{\overline{N}}&(\Price^{(\cdot)})(\cdot)\big|_{C^0(\OP)}=\big|\Price - I_{\overline{N}}(\Price^{(\cdot)})(\cdot)\big|_{C^0([-1,1]^D)}\\
&\le \sqrt{D}2^{\frac{D}{2}+1}\widehat{V}\varrho_{\min}^{-N}(1-\varrho_{\min}^{-2})^{-\frac{D}{2}}\\
&= \sqrt{D}2^{\frac{D}{2}+1}V\varrho_{\min}^{-N}(1-\varrho_{\min}^{-2})^{-\frac{D}{2}},
\end{align*}
where $\widehat{V} =\sup_{p\in B([-1,1]^D,\varrho)}\widehat{\Price}^p$, $V=\sup_{p\in B(\OP,\varrho)}\Price^p$. Summarizing, the transformation $\tau_{\OP}:[-1,1]^D\rightarrow\OP$ does not affect the error analysis, only by applying the transformation as described in Section \ref{sec:Multi_Cheby}, 
\begin{align}\label{eq-genB}
B(\OP,\varrho):=B([\underline{p}_1,\olp[1]],\varrho_1)\times\ldots\times B([\underline{p}_D,\olp[D]],\varrho_D ),
\end{align}
with $B([\underline{p},\olp],\varrho):=\tau_{[\underline{p},\olp]}\circ B([-1,1],\varrho)$. Note that $\varrho_i$ is not the radius of the ellipse $B([\underline{p}_i,\overline{p}_i],\varrho_i)$ but of the normed ellipse $B([-1,1],\varrho_i)$.
Therefore, in the following it suffices to show the proof for $\OP=[-1,1]^D$.

 As in \cite[Proof of Lemma 7.3.3]{SauterSchwab2004} we introduce the scalar product
\begin{align*}
\langle f,g\rangle_{\varrho}:=\int_{B(\OP,\varrho)}\frac{f(z)\overline{g(z)}}{\prod_{i=1}^D\sqrt{|1-z_i^2|}}\dd z
\end{align*}
and the Hilbert space
\begin{align*}
L^2(B(\OP,\varrho)):=\{f:\ f\text{ is analytic in }B(\OP,\varrho)\text{ and }\vert\vert f\vert\vert^2_{\varrho}:=\langle f,f\rangle_{\varrho}<\infty\}.
\end{align*}
Following \cite[Proof of Lemma 7.3.3]{SauterSchwab2004}, we define a complete orthonormal system for $L^2(B(\OP,\varrho))$ w.r.t. the scalar product  $\langle \cdot,\cdot\rangle_{\varrho}$ by the scaled Chebyshev polynomials
\begin{align*}
\tilde{T}_{\mu}(z):=c_{\mu}T_{\mu}(z)\text{ with } c_{\mu}:=\left(\frac{2}{\pi}\right)^{\frac{D}{2}}\prod_{i=1}^D(\varrho_i^{2\mu_i}+\varrho_i^{-2\mu_i})^{-\frac{1}{2}},\quad\text{ for all } \mu\in\mathbb{N}_0^D.
\end{align*}
Then, for any arbitrary bounded functional $E$ on $L^2(B(\OP,\varrho))$ we have
\begin{align}
\vert E(f)\vert\le \vert\vert E\vert\vert_{\varrho}\vert\vert f\vert\vert_{\varrho},\label{norm_of_operator}
\end{align}
where $\vert\vert E\vert\vert_{\varrho}$ denotes the operator norm. Due to the orthonormality of $\left(\tilde{T}_{\mu}\right)_{\mu\in\mathbb{N}_0^D}$ it follows that
\begin{align*}
\vert\vert E\vert\vert_{\varrho}=\sup_{f\in L^2(B(\OP,\varrho))\setminus \{0\}}\frac{\vert E(f)\vert}{\Vert f\Vert_{\varrho}}=\sqrt{\sum_{\mu\in\mathbb{N}_0^D}\vert E(\tilde{T}_{\mu})\vert^2}.
\end{align*}
In the following let $E$ be the error of the Chebyshev polynomial interpolation at a fix $p\in\OP$,
\begin{align*}
E(f):=|f(p)- I_{\overline{N}}(f{(\cdot)})(p)|.
\end{align*}
Starting with \eqref{norm_of_operator}, we first focus on $\vert\vert E\vert\vert_{\varrho}$,
\begin{align*}
\Vert E \Vert^2_{\varrho}=\sum_{\mu\in\mathbb{N}_0^D}\vert E(\tilde{T}_{\mu})\vert^2=\sum_{\mu\in\mathbb{N}_0^D} c^2_{\mu}\vert E(T_{\mu})\vert^2.
\end{align*}
At this step we apply Lemma \ref{Lemma_Interpol_Tj} and obtain
\begin{align*}
\sum_{\mu\in\mathbb{N}_0^D} c^2_{\mu}\vert E(T_{\mu})\vert^2=\sum_{\mu\in\mathbb{N}_0^D, \exists i:\mu_i>N_i} c^2_{\mu}\vert E(T_{\mu})\vert^2\le \sum_{\mu\in\mathbb{N}_0^D, \exists i:\mu_i>N_i} 4c^2_{\mu}.
\end{align*}
Overall, using $\left(\prod_{j=1}^D\varrho_j^{2\mu_j}+x\right)^{-1}\le \left(\prod_{j=1}^D\varrho_j^{2\mu_j}\right)^{-1}=\prod_{j=1}^D\varrho_j^{-2\mu_j}$ for $x>0$, $\mu_j\in\mathbb{N}_0$ and $j=1,\ldots,D$ and  this leads to
\begin{align*}
\Vert E \Vert^2_{\varrho}&\le 4 \sum_{\mu\in\mathbb{N}_0^D, \exists i:\mu_i>N_i}c_{\mu}^2\le 4\left(\frac{2}{\pi}\right)^D\sum_{i=1}^D\left( \sum_{\mu\in\mathbb{N}_0^D, \mu_i>N_i}\prod_{j=1}^D\varrho^{-2\mu_j}_j\right)\\
&\le4\left(\frac{2}{\pi}\right)^D \sum_{i=1}^D\varrho_i^{-2N_i}\left( \sum_{\mu\in\mathbb{N}_0^D,\mu_i>N_i}\varrho_i^{-2(\mu_i-N_i)}\prod_{j=1,j\neq i}^D\varrho^{-2\mu_j}_j\right)\\
&\le 4\left(\frac{2}{\pi}\right)^D \sum_{i=1}^D\varrho_i^{-2N_i}\left( \sum_{\mu\in\mathbb{N}_0^D}\prod_{j=1}^D\varrho^{-2\mu_j}_j\right).
\end{align*}
From this point on we use the convergence of the geometric series since $\vert\varrho_j^{-2}\vert<1,\ j=1,\ldots,D$,
\begin{align*}
\Vert E \Vert^2_{\varrho}&\le 4\left(\frac{2}{\pi}\right)^D \sum_{i=1}^D\varrho_i^{-2N_i}\left( \sum_{\mu_1=0}^\infty\ldots\sum_{\mu_D=0}^\infty\prod_{j=1}^D\varrho^{-2\mu_j}_j\right)\\
&= 4\left(\frac{2}{\pi}\right)^D \sum_{i=1}^D\varrho_i^{-2N_i}\left( \sum_{\mu_1=0}^\infty\ldots\sum_{\mu_{D-1}=0}^\infty\prod_{j=1}^{D-1}\varrho^{-2\mu_j}_j\sum_{\mu_D=0}^{\infty}\varrho^{-2\mu_D}_D\right)\\
&= 4\left(\frac{2}{\pi}\right)^D \sum_{i=1}^D\varrho_i^{-2N_i}\left( \sum_{\mu_1=0}^\infty\ldots\sum_{\mu_{D-1}=0}^\infty\prod_{j=1}^{D-1}\varrho^{-2\mu_j}_j\frac{1}{1-\varrho^{-2}_D} \right)\\
&=\ldots=4\left(\frac{2}{\pi}\right)^D \sum_{i=1}^D\varrho_i^{-2N_i}\prod_{j=1}^D\frac{1}{1-\varrho^{-2}_j}.
\end{align*}
Recalling \eqref{norm_of_operator}, we have to estimate $\Vert f\Vert_\varrho$,
\begin{align*}
\Vert f\Vert_\varrho^2=\int_{B(\OP,\varrho)}\frac{f(z)\overline{f(z)}}{\prod_{i=1}^D\sqrt{|1-z_i^2|}}\dd z\le \left(\sup_{z\in B(\OP,\varrho)}\vert f(z)\vert\right)^2\Vert 1\Vert_\varrho^2.
\end{align*}
From $\pi^{\frac{D}{2}}\tilde{T}_0=1$ it directly follows that $\Vert 1\Vert_\varrho^2=\left(\pi^{\frac{D}{2}}\right)^2\Vert \tilde{T}_0\Vert_\varrho^2=\pi^D$ and therewith
\begin{align*}
\Vert f\Vert_\varrho^2\le \pi^D\cdot V^2.
\end{align*}
Combining the results leads to
\begin{align*}
\vert E(f)\vert=|f(p)- I_{\overline{N}}(f{(\cdot)})(p)\big|&\le \left(\pi^D\cdot V^2\cdot 4\left(\frac{2}{\pi}\right)^D \sum_{i=1}^D\varrho_i^{-2N_i}\prod_{j=1}^D\frac{1}{1-\varrho^{-2}_j}\right)^{\frac{1}{2}}\\
&= 2^{\frac{D}{2}+1}V\left( \sum_{i=1}^D\varrho_i^{-2N_i}\prod_{j=1}^D\frac{1}{1-\varrho^{-2}_j}\right)^{\frac{1}{2}}.
\end{align*}
\end{proof}
The following lemma shows that the Chebyshev interpolation of a polynomial with a degree as most as high as the degree of the interpolating Chebyshev polynomial is exact and furthermore determines an upper bound for interpolating Chebyshev polynomials with a higher degree.
\begin{lemma}\label{Lemma_Interpol_Tj}
For $x\in[-1,1]^D$ it holds
\begin{align}
&|T_{\mu}(x)-I_{\overline{N}}(T_{\mu}{(\cdot)})(x)|=0\quad\forall\mu\in\mathbb{N}_0^D:\mu_i\le N_i,i=1,\ldots,D,\label{Lemma_Interpol_T_0}\\
&|T_{\mu}(x)-I_{\overline{N}}(T_{\mu}{(\cdot)})(x)|\le2\quad\forall\mu\in\mathbb{N}_0^D:\exists i\in\{1,\ldots,D\}:\mu_i>N_i\label{Lemma_Interpol_T_2}.
\end{align}
\end{lemma}
\begin{proof}
Uniqueness properties of the Chebyshev interpolation directly imply \eqref{Lemma_Interpol_T_0}. 
The proof of \eqref{Lemma_Interpol_T_2} is similar to \cite[Proof of Hilfssatz 7.3.1]{SauterSchwab2004}. They use the zeros of the Chebyshev polynomial as interpolation points, whereas we use the extreme points and therefore, we use a different orthogonality property in this proof. We first focus on the one-dimensional case. Recalling \eqref{eq-ChebInter1dim}, the Chebyshev interpolation of $T_{\mu}$, $\mu>N$, is given as
\begin{align*}
I_N(T_{\mu})(x)=\sum_{j=0}^Nc_jT_j(x)\quad\text{with}\quad c_j=\frac{2^{\mathbbm{1}_{0<j<N}}}{N}\sum_{k=0}^N{}^{''}T_{\mu}(x_k)T_{j}(x_k),\quad j\le N,
\end{align*}
where $x_k$ denotes the $k$-th extremum of $T_N$. Here, we can apply the following orthogonality (\cite[p.54]{Rivlin1990}),
\begin{equation}
\label{eq:Ortho}
\begin{split}
\sum_{k=0}^N{}^{''}T_{\mu}(x_k)T_{j}(x_k)=\begin{cases}
0, & \mu+j\neq0\text{ mod } (2N) \text{ and } \vert \mu-j\vert\neq0\text{ mod } (2N),\\
N,& \mu+j=0\text{ mod } (2N) \text{ and } \vert \mu-j\vert=0\text{ mod } (2N), \\
\frac{N}{2}, &\mu+j=0\text{ mod } (2N) \text{ and } \vert \mu-j\vert\neq0\text{ mod } (2N),\\
\frac{N}{2}, &\mu+j\neq0\text{ mod } (2N)\text{ and }\vert \mu-j\vert=0\text{ mod } (2N).
\end{cases}
\end{split}
\end{equation}
For $j\le N$ and $\mu>N$ this yields the existence of $\gamma \le N$ such that
\begin{align}\label{eq:Aux}
I_N(T_{\mu})=T_{\gamma}.
\end{align}
\eqref{eq:Aux} follows elementarily from the case that for any $\mu>N$ only for one $0\le j\le N$ the orthogonality can lead to a coefficient $c_j>0$.

Proving the claim, we distinguish several cases. In all of these cases we assume that there exists $0\le j\le N$ such that $\sum_{k=0}^N{}^{''}T_{\mu}(x_k)T_{j}(x_k)\neq0$. We will then show that for all other $0\le i\le N,i\neq j$ it follows $\sum_{k=0}^N{}^{''}T_{\mu}(x_k)T_{j}(x_k)=0$.

Firstly, assume there exists $j$ such that $\mu+j=0\text{ mod } (2N)$ and $\mu-j=0\text{ mod } (2N)$. Then it directly follows for all $0\le i\le N$, $i\neq j$ that $\mu+i\neq 0\text{ mod } (2N)$ and $\mu-i\neq0\text{ mod } (2N)$.

Secondly, assume there exists $j$ such that $\mu+j=0\text{ mod } (2N)$ and $\mu-j\neq0\text{ mod } (2N)$. Analogously, for all $0\le i\le N$, $i\neq j$ we have $\mu+i\neq 0\text{ mod } (2N)$ and additionally from $\mu+j=0\text{ mod } (2N)$ it follows that $\mu+j-2N=0\text{ mod } (2N)$ and therewith for all $0\le i\le N$, $i\neq j$ we have $\mu-i>\mu+j-2N$ which is equivalent to $ \mu-i\neq0\text{ mod } (2N)$.

Similar argumentation holds for the third case $\mu+j\neq 0\text{ mod }(2N)\text{ and }\vert \mu-j\vert=0\text{ mod } (2N)$.

Therewith, \eqref{eq:Aux} holds and it directly follows that $\vert T_{\mu}-I_N(T_{\mu})\vert\le\vert T_{\mu}\vert+ \vert I_N(T_{\mu})\vert\le 1+1=2$. Thus \eqref{Lemma_Interpol_T_2} holds in the one-dimensional case. The extension to the multivariate case follows analogously by applying the triangle inequality and inserting the one-dimensional result to each tensor component.
\end{proof}

\bibliographystyle{chicago}
  \bibliography{LiteraturFourierRB}

\begin{thebibliography}{}

\bibitem[\protect\citeauthoryear{Armenti, Cr\'epey, Drapeau, and
  Papapantoleon}{Armenti et~al.}{2015}]{ArmentiCrepeyDrapeauPapapantoleon2015}
Armenti, Y., S.~Cr\'epey, S.~Drapeau, and A.~Papapantoleon (2015).
\newblock Multivariate shortfall risk allocation.
\newblock available on \url{http://arxiv.org/abs/1507.05351}.

\bibitem[\protect\citeauthoryear{Barndorff-Nielsen and
  Shephard}{Barndorff-Nielsen and
  Shephard}{2001}]{Barndorff-NielsenShephard2001}
Barndorff-Nielsen, O.~E. and N.~Shephard (2001).
\newblock {Non-Gaussian Ornstein--Uhlenbeck-based models and some of their uses
  in financial economics}.
\newblock {\em Journal of the Royal Statistical Society, Series B\/}~{\em 63},
  167--241.

\bibitem[\protect\citeauthoryear{Bernstein}{Bernstein}{1912}]{Bernstein1912}
Bernstein, S.~N. (1912).
\newblock Sur l'ordre de la meilleure approximation des fonctions continues par
  des polynomes de degr{\'e} donn{\'e}, mimoires acad.
\newblock {\em Acad{\'e}mie Royale de Belgique. Classe des Sciences.
  M{\'e}moires\/}~{\em 4}.

\bibitem[\protect\citeauthoryear{Black and Scholes}{Black and
  Scholes}{1973}]{BlackScholes1973}
Black, F. and M.~Scholes (1973).
\newblock The pricing of options and corporate liabilities.
\newblock {\em Journal of Political Economy\/}~{\em 81\/}(3), 637--654.

\bibitem[\protect\citeauthoryear{Boyarchenko and Levendorski\u{i}}{Boyarchenko
  and Levendorski\u{i}}{2000}]{BoyarchenkoLevendoskii2000}
Boyarchenko, S.~I. and S.~Z. Levendorski\u{i} (2000).
\newblock Option pricing for truncated l{\'e}vy processes.
\newblock {\em International Journal of Theoretical and Applied Finance\/}~{\em
  3\/}(03), 549--552.

\bibitem[\protect\citeauthoryear{Boyarchenko and Levendorski\u{i}}{Boyarchenko
  and Levendorski\u{i}}{2002}]{BoyarchenkoLevendorskii2002}
Boyarchenko, S.~I. and S.~Z. Levendorski\u{i} (2002).
\newblock {\em Non-Gaussian Merton-Black-Scholes Theory}, Volume~9.
\newblock World Scientific.

\bibitem[\protect\citeauthoryear{Brennan and Schwartz}{Brennan and
  Schwartz}{1977}]{BrenanSchwartz1977}
Brennan, M.~J. and E.~S. Schwartz (1977).
\newblock {The valuation of American put options}.
\newblock {\em The Journal of Finance\/}~{\em 2\/}(32), 449--462.

\bibitem[\protect\citeauthoryear{Burkovska, Haasdonk, Salomon, and
  Wohlmuth}{Burkovska et~al.}{2015}]{burkovska2015reduced}
Burkovska, O., B.~Haasdonk, J.~Salomon, and B.~Wohlmuth (2015).
\newblock {Reduced basis methods for pricing options with the Black-Scholes and
  Heston models}.
\newblock {\em SIAM Journal on Financial Mathematics\/}~{\em 6\/}(1), 685--712.

\bibitem[\protect\citeauthoryear{Canuto and Quarteroni}{Canuto and
  Quarteroni}{1982}]{CanutoQuarteroni1982}
Canuto, C. and A.~Quarteroni (1982).
\newblock Approximation results for orthogonal polynomials in sobolev spaces.
\newblock {\em Mathematics of Computation\/}~{\em 38\/}(157), 67--86.

\bibitem[\protect\citeauthoryear{Carr, Geman, Madan, and Yor}{Carr
  et~al.}{2002}]{CarrGemanMadanYor2002}
Carr, P., H.~Geman, D.~B. Madan, and M.~Yor (2002).
\newblock The fine structure of asset returns: An empirical investigation.
\newblock {\em Journal of Business\/}~{\em 75\/}(2), 305--332.

\bibitem[\protect\citeauthoryear{Carr, Geman, Madan, and Yor}{Carr
  et~al.}{2003}]{CarrGemanMadanYor2003}
Carr, P., H.~Geman, D.~B. Madan, and M.~Yor (2003).
\newblock {Stochastic volatility for L\'evy processes}.
\newblock {\em Mathematical Finance\/}~{\em 13}, 345--382.

\bibitem[\protect\citeauthoryear{Carr and Madan}{Carr and
  Madan}{1999}]{CarrMadan99}
Carr, P. and D.~B. Madan (1999).
\newblock Option valuation and the fast {F}ourier transform.
\newblock {\em Journal of Computional Finance\/}~{\em 2\/}(4), 61--73.

\bibitem[\protect\citeauthoryear{Cheridito and Wugalter}{Cheridito and
  Wugalter}{2012}]{CheriditoWugalter2012}
Cheridito, P. and A.~Wugalter (2012).
\newblock Pricing and hedging in affine models with possibility of default.
\newblock {\em SIAM Journal on Financial Mathematics\/}~{\em 3\/}(1), 328--350.

\bibitem[\protect\citeauthoryear{Cont, Lantos, and Pironneau}{Cont
  et~al.}{2011}]{ContLantosPironneau2011}
Cont, R., N.~Lantos, and O.~Pironneau (2011).
\newblock A reduced basis for option pricing.
\newblock {\em SIAM Journal on Financial Mathematics\/}~{\em 2\/}(1), 287--316.

\bibitem[\protect\citeauthoryear{Cuchiero, Keller-Ressel, and
  Teichmann}{Cuchiero et~al.}{2012}]{CuchieroKeller-ResselTeichmann2015}
Cuchiero, C., M.~Keller-Ressel, and J.~Teichmann (2012).
\newblock Polynomial processes and their applications to mathematical finance.
\newblock {\em Finance and Stochastics\/}~{\em 4\/}(16).

\bibitem[\protect\citeauthoryear{Duffie, Filipovi{\'c}, and
  Schachermayer}{Duffie et~al.}{2003}]{DuffieFilipovicSchachermayer2003}
Duffie, D., D.~Filipovi{\'c}, and W.~Schachermayer (2003).
\newblock Affine processes and applications in finance.
\newblock {\em Annals of Applied Probability\/}~{\em 13\/}(3), 984--1053.

\bibitem[\protect\citeauthoryear{Eberlein, Glau, and Papapantoleon}{Eberlein
  et~al.}{2010}]{EberleinGlauPapapantoleon2010a}
Eberlein, E., K.~Glau, and A.~Papapantoleon (2010).
\newblock Analysis of {F}ourier transform valuation formulas and applications.
\newblock {\em Applied Mathematical Finance\/}~{\em 17\/}(3), 211--240.

\bibitem[\protect\citeauthoryear{Eberlein, Glau, and Papapantoleon}{Eberlein
  et~al.}{2011}]{EberleinGlauPapapantoleon2010b}
Eberlein, E., K.~Glau, and A.~Papapantoleon (2011).
\newblock Analyticity of the {W}iener--{H}opf factorization and valuation of
  exotic options in {L}\'evy models.
\newblock In G.~Di~Nunno and B.~{\O}ksendal (Eds.), {\em Advanced Mathematical
  Methods for Finance}, pp.\  223 --� 245. Springer.

\bibitem[\protect\citeauthoryear{Eberlein, Keller, and Prause}{Eberlein
  et~al.}{1998}]{EberleinKellerPrause98}
Eberlein, E., U.~Keller, and K.~Prause (1998).
\newblock New insights into smile, mispricing and value at risk: the hyperbolic
  model.
\newblock {\em Journal of Business\/}~{\em 71\/}(3), 371--405.

\bibitem[\protect\citeauthoryear{Eberlein and {\"O}zkan}{Eberlein and
  {\"O}zkan}{2005}]{EberleinOezkan2005}
Eberlein, E. and F.~{\"O}zkan (2005).
\newblock The {L}\'evy {LIBOR} model.
\newblock {\em Finance and Stochastics\/}~{\em 9\/}(3), 327--348.

\bibitem[\protect\citeauthoryear{Feng and Linetsky}{Feng and
  Linetsky}{2008}]{FengLinetsky2008}
Feng, L. and V.~Linetsky (2008).
\newblock Pricing discretely monitored barrier options and defaultable bonds in
  {L}{\'e}vy process models: A fast {H}ilbert transform approach.
\newblock {\em Mathematical Finance\/}~{\em 18\/}(3), 337--384.

\bibitem[\protect\citeauthoryear{Filipovi\'c, Larsson, and Trolle}{Filipovi\'c
  et~al.}{2014}]{FilipovicLarssonTrolle2014}
Filipovi\'c, D., M.~Larsson, and A.~Trolle (2014).
\newblock Linear-rational term structure models.
\newblock Preprint.

\bibitem[\protect\citeauthoryear{Filipovi\'c and Mayerhofer}{Filipovi\'c and
  Mayerhofer}{2009}]{FilipovicMayerhofer2009}
Filipovi\'c, D. and E.~Mayerhofer (2009).
\newblock {Affine diffusion processes: Theory and applications}.
\newblock In {\em {Advanced Financial Modelling}}, Volume~8 of {\em Radon
  Series on Computational and Applied Mathematics}, pp.\  125--164. de Gruyter.

\bibitem[\protect\citeauthoryear{Ga{\ss}, Glau, and Mair}{Ga{\ss}
  et~al.}{2015}]{GassGlauMair2015}
Ga{\ss}, M., K.~Glau, and M.~Mair (2015).
\newblock {Magic Points Fourier Transform Option Pricing -- Polynomial and
  Empirical Interpolation}.
\newblock Working paper.

\bibitem[\protect\citeauthoryear{Glau}{Glau}{2015}]{Glau2015a}
Glau, K. (2015).
\newblock Classification of {L}\'evy processes with parabolic {K}olmogorov
  backward equations.
\newblock forthcoming in {S}{I}{A}{M} journal {T}heory of {P}robability.

\bibitem[\protect\citeauthoryear{Glau}{Glau}{2016}]{Glau2015b}
Glau, K. (2016).
\newblock {Feynman-Kac formula for {L}\'evy processes with discontionuous
  killing rate}.
\newblock forthcoming in Finance and Stochastics.

\bibitem[\protect\citeauthoryear{Haasdonk, Salomon, and Wohlmuth}{Haasdonk
  et~al.}{2013}]{HaasdonkSalomonWohlmuth2012b}
Haasdonk, B., J.~Salomon, and B.~Wohlmuth (2013).
\newblock {A reduced basis method for the simulation of American options}.
\newblock In {\em Numerical Mathematics and Advanced Applications 2011}, pp.\
  821--829. Springer.

\bibitem[\protect\citeauthoryear{Heston}{Heston}{1993}]{Heston1993}
Heston, S.~L. (1993).
\newblock A closed-form solution for options with stochastic volatility with
  applications to bond and currency options.
\newblock {\em Review of Financial Studies\/}~{\em 6\/}(2), 327--343.

\bibitem[\protect\citeauthoryear{J{\"a}nich}{J{\"a}nich}{2004}]{jaenich2004funktionentheorie}
J{\"a}nich, K. (2004).
\newblock {\em Funktionentheorie: eine Einf{\"u}hrung}.
\newblock Springer-Lehrbuch. Springer.

\bibitem[\protect\citeauthoryear{Kallsen}{Kallsen}{2006}]{Kallsen2006}
Kallsen, J. (2006).
\newblock A didactic note on affine stochastic volatility models.
\newblock In Y.~Kabanov, R.~Lipster, and J.~Stoyanov (Eds.), {\em From
  Stochastic Calculus to Mathematical Finance: The Shiryaev Festschrift}, pp.\
  343--368. Springer.

\bibitem[\protect\citeauthoryear{Keller-Ressel and Mayerhofer}{Keller-Ressel
  and Mayerhofer}{2015}]{Keller-ResselMayerhofer2015}
Keller-Ressel, M. and E.~Mayerhofer (2015).
\newblock {Exponential moments of affine processes}.
\newblock {\em The Annals of Applied Probability\/}~{\em 25}, 714--–752.

\bibitem[\protect\citeauthoryear{Keller-Ressel, Papapantoleon, and
  Teichmann}{Keller-Ressel
  et~al.}{2013}]{Keller-ResselPapapantoleonTeichmann2013}
Keller-Ressel, M., A.~Papapantoleon, and J.~Teichmann (2013).
\newblock The affine libor models.
\newblock {\em Mathematical Finance\/}~{\em 23\/}(4), 627--658.

\bibitem[\protect\citeauthoryear{Kolda and Bader}{Kolda and
  Bader}{2009}]{KoldaBader2009}
Kolda, T.~G. and B.~W. Bader (2009).
\newblock Tensor decompositions and applications.
\newblock {\em SIAM Review\/}~{\em 51\/}(3), 455--500.

\bibitem[\protect\citeauthoryear{Kudryavtsev and Levendorski\u{i}}{Kudryavtsev
  and Levendorski\u{i}}{2009}]{KudryavtsevLevendorskiy2009}
Kudryavtsev, O. and S.~Z. Levendorski\u{i} (2009).
\newblock Fast and accurate pricing of barrier options under {L}\'evy
  processes.
\newblock {\em Finance and Stochastics\/}~{\em 13\/}(4), 531--562.

\bibitem[\protect\citeauthoryear{Lee}{Lee}{2004}]{Lee2004}
Lee, R.~W. (2004).
\newblock {Option pricing by transform methods: Extensions, unification, and
  error control}.
\newblock {\em Journal of Computational Finance\/}~{\em 7\/}(3), 51--–86.

\bibitem[\protect\citeauthoryear{Levendorski\u{i}}{Levendorski\u{i}}{2012}]{Levendorskiy2012}
Levendorski\u{i}, S.~Z. (2012).
\newblock Efficient pricing and reliable calibration in the {H}eston model.

\bibitem[\protect\citeauthoryear{Lord, Fang, Bervoets, and Oosterlee}{Lord
  et~al.}{2008}]{LordFangBervoetsOosterlee2008}
Lord, R., F.~Fang, F.~Bervoets, and C.~W. Oosterlee (2008).
\newblock {A fast and accurate FFT-based method for pricing early-exercise
  options under {L}\'evy processes}.
\newblock {\em SIAM Journal on Scientific Computing\/}~{\em 30\/}(4),
  1678--1705.

\bibitem[\protect\citeauthoryear{Merton}{Merton}{1973}]{merton1973}
Merton, R.~C. (1973).
\newblock Theory of rational option pricing.
\newblock {\em The Bell Journal of Economics and Management Science\/}~{\em
  4\/}(1), 141--183.

\bibitem[\protect\citeauthoryear{Merton}{Merton}{1976}]{merton1976}
Merton, R.~C. (1976).
\newblock Option pricing when underlying stock returns are discontinuous.
\newblock {\em Journal of Financial Economics\/}~{\em 3}, 125--144.

\bibitem[\protect\citeauthoryear{Pachon}{Pachon}{2016}]{Pachon2016}
Pachon, R. (2016).
\newblock {Numerical pricing of European options with arbitrary payoffs}.
\newblock Available at SSRN: \url{http://ssrn.com/abstract=2712402}.

\bibitem[\protect\citeauthoryear{Pironneau}{Pironneau}{2011}]{Pironneau2011}
Pironneau, O. (2011).
\newblock Reduced basis for vanilla and basket options.
\newblock {\em Risk and Decision Analysis\/}~{\em 2\/}(4), 185--194.

\bibitem[\protect\citeauthoryear{Pistorius and Stolte}{Pistorius and
  Stolte}{2012}]{PistoriusStolte2012}
Pistorius, M. and J.~Stolte (2012).
\newblock {Fast computation of vanilla prices in time-changed models and
  implied volatilities}.
\newblock {\em International Journal of Theoretical and Applied Finance\/}~{\em
  15\/}(4), 1250031--1250031.

\bibitem[\protect\citeauthoryear{Platte and Trefethen}{Platte and
  Trefethen}{2008}]{PlatteTrefethen2008}
Platte, R.~B. and N.~L. Trefethen (2008).
\newblock {\em {Chebfun: A new kind of numerical computing}}, pp.\  69--86.
\newblock Springer.

\bibitem[\protect\citeauthoryear{Raible}{Raible}{2000}]{Raible}
Raible, S. (2000).
\newblock {\em L\'evy {P}rocesses in {F}inance: {T}heory, {N}umerics, and
  {E}mpirical {F}acts}.
\newblock Ph.\ D. thesis, Universit\"at Freiburg.

\bibitem[\protect\citeauthoryear{Rivlin}{Rivlin}{1990}]{Rivlin1990}
Rivlin, T.-J. (1990).
\newblock {\em Chebyshev polynomials.}
\newblock John Wiley \& Sons, Inc., copublished in the United States with John
  Wiley \&.

\bibitem[\protect\citeauthoryear{Runge}{Runge}{1901}]{Runge1901}
Runge, C. (1901).
\newblock {{\"U}ber empirische Funktionen und die Interpolation zwischen
  {\"a}quidistanten Ordinaten}.
\newblock {\em Zeitschrift f\"ur Mathematik und Physik\/}~{\em 46}, 224--243.

\bibitem[\protect\citeauthoryear{Sachs and Schu}{Sachs and
  Schu}{2010}]{SachsSchu2010}
Sachs, E.~W. and M.~Schu (2010).
\newblock Reduced order models in {PIDE} constrained optimization.
\newblock {\em Control and Cybernetics\/}~{\em 39\/}(3), 661--675.

\bibitem[\protect\citeauthoryear{Sato}{Sato}{1999}]{Sato}
Sato, K.-I. (1999).
\newblock {\em L{\'e}vy {P}rocesses and {I}nfinitely {D}ivisible
  {D}istributions}.
\newblock Cambridge University Press.

\bibitem[\protect\citeauthoryear{Sauter and Schwab}{Sauter and
  Schwab}{2004}]{SauterSchwab2004}
Sauter, S. and C.~Schwab (2004).
\newblock {\em Randelementmethoden: Analyse, Numerik und Implementierung
  schneller Algorithmen}.
\newblock Vieweg+ Teubner Verlag.

\bibitem[\protect\citeauthoryear{Tadmor}{Tadmor}{1986}]{Tadmor1986}
Tadmor, E. (1986).
\newblock The exponential accuracy of fourier and chebyshev differencing
  methods.
\newblock {\em SIAM Journal on Numerical Analysis\/}~{\em 23\/}(1), 1--10.

\bibitem[\protect\citeauthoryear{Teichmann}{Teichmann}{2015}]{Teichmann2015}
Teichmann, J. (2015).
\newblock Tractable american options problems.
\newblock Talk at ETH Z{\"u}rich, September 2015, available on
  \url{https://people.math.ethz.ch/~jteichma/its_talks/teichmann_talk_zurich_sep2015.pdf}.

\bibitem[\protect\citeauthoryear{Trefethen}{Trefethen}{2011}]{TrefethenMythTalk}
Trefethen, L.~N. (2011).
\newblock Talk: Six myths of polynomial interpolation and quadrature.
\newblock \url{https://people.maths.ox.ac.uk/trefethen/mythstalk.pdf}.

\bibitem[\protect\citeauthoryear{Trefethen}{Trefethen}{2013}]{Trefethen2013}
Trefethen, L.~N. (2013).
\newblock {\em {Approximation Theory and Approximation Practice}}.
\newblock SIAM books.

\bibitem[\protect\citeauthoryear{Wloka}{Wloka}{1987}]{Wloka-english}
Wloka, J. (1987).
\newblock {\em Partial {D}ifferential {E}quations}.
\newblock Cambridge University Press.

\end{thebibliography}

\end{document}